\documentclass[11pt]{article}
\usepackage{amsmath,amssymb}               
\usepackage{amsthm}{
   \newtheorem{condition}{Condition}
   
   \newtheorem{proposition}{Proposition}
   \newtheorem{theorem}{Theorem}
   \newtheorem{lemma}{Lemma}
   \newtheorem{corollary}{Corollary}
}
\usepackage[mathscr]{euscript}
\usepackage[round]{natbib}

\usepackage{amsfonts}
\usepackage{graphicx}
\usepackage{xcolor}
\usepackage{color,soul}
\usepackage[hidelinks]{hyperref}
\usepackage{xcolor}
\usepackage{setspace,array}
\usepackage{enumerate}
\usepackage[shortlabels]{enumitem}
\newcommand{\round}[1]{\ensuremath{\lfloor#1\rceil}}
\newcolumntype{C}[1]{>{\centering\let\newline\\\arraybackslash\hspace{0pt}}m{#1}}

\usepackage{algorithm,algpseudocode}

\usepackage[T1]{fontenc}
\usepackage[utf8]{inputenc}
\usepackage[margin=0.9in]{geometry}
\allowdisplaybreaks

\DeclareMathOperator*{\argmax}{argmax}
\DeclareMathOperator*{\expit}{expit}
\DeclareMathOperator*{\logit}{logit}

\usepackage{xr}
\makeatletter
\newcommand*{\addFileDependency}[1]{
  \typeout{(#1)}
  \@addtofilelist{#1}
  \IfFileExists{#1}{}{\typeout{No file #1.}}
}
\makeatother

\usepackage{lscape}

\newcommand{\indep}{\rotatebox[origin=c]{90}{$\models$}}

\DeclareMathAlphabet{\mathpzc}{OT1}{pzc}{m}{it}
\newcommand{\hd}{\mathpzc{h}}

\title{Propensity score augmentation\\in matching-based estimation of causal effects}

\usepackage{authblk}
\author[1]{Ernesto Ulloa-P\'erez}
\author[3,2]{Marco Carone}
\author[2,3]{Alex Luedtke}
\affil[1]{Department of Biostatistics, Epidemiology and Informatics, University of Pennsylvania}
\affil[2]{Department of Statistics, University of Washington}
\affil[3]{Department of Biostatistics, University of Washington}

\begin{document}

\maketitle

\doublespacing

\begin{abstract}
When assessing the causal effect of a binary exposure using observational data, confounder imbalance across exposure arms must be addressed. Matching methods, including propensity score-based matching, can be used to deconfound the causal relationship of interest. They have been particularly popular in practice, at least in part due to their simplicity and interpretability. However, these methods can suffer from low statistical efficiency compared to many competing methods. In this work, we propose a novel matching-based estimator of the average treatment effect based on a suitably-augmented propensity score model. Our procedure can be shown to have greater statistical efficiency than traditional matching estimators whenever prognostic variables are available, and in some cases, can nearly reach the nonparametric efficiency bound. In addition to a theoretical study, we provide numerical results to illustrate our findings. Finally, we use our novel procedure to estimate the effect of circumcision on risk of HIV-1 infection using vaccine efficacy trial data.	
\end{abstract}

\section{Introduction}

In many studies, the goal is to estimate the causal effect of a binary exposure, say treatment versus control, on an outcome of interest. In randomized trials, since by design treatment level is independent of patient characteristics, treatment effects can often be inferred through simple, unadjusted comparisons. For example, the average treatment effect can be estimated by the difference in mean observed outcomes across treatment levels. In contrast, in observational studies, patients receiving different treatment levels may also differ in baseline characteristics, thereby giving rise to potential confounding of the treatment-outcome relationship. In order to infer about treatment effects, a lack of balance in patient characteristics across treatment levels must be addressed using appropriate methods.

Matching methods are commonly used to infer causal effects using data from observational studies. Such methods outline how to turn the original dataset into a matched dataset in which the distribution of confounders is appropriately balanced across treatment groups. Often, this is accomplished by pairing each patient to one (or many) patients with similar confounder  values in the opposite treatment arm. Once a matched dataset has been constructed, treatment effects can be estimated similarly as they would had the data been generated from a randomized trial.  Since the conditional treatment probability given baseline covariates --- referred to as the propensity score --- is a balancing score, in that treatment level and baseline covariates are independent conditionally upon the propensity score value, matching on the propensity score suffices to balance the distribution of confounders between treatment groups \citep{rosenbaum1983}. Propensity score matching has been popular in practice because it provides a principled means of dimension reduction, allowing pairing of patients with dissimilar covariate values that may nevertheless have similar propensity score values. There are many ways to construct a matched set based on propensity score values --- \cite{austin2011} provides a comprehensive review. Below, we focus on one to several (1:$M$) propensity score matching with replacement aimed at targeting the average treatment effect; throughout this paper, we refer to this scheme as 1:$M$ matching.

Matching estimators are appealing because of their simplicity, the ease with which their output can be visualized and communicated, and the wide availability of software \citep{stuart2010}. However, such estimators are generally inefficient, failing to achieve the nonparametric efficiency bound for the average treatment effect; this fact was formally shown in \cite{abadie2016}, who derived the large-sample properties of the 1:$M$ matching estimator. Since there exist several alternative methods for estimating the average treatment effect that in fact do attain this bound \citep{robins1994,van2006},  this limitation may discourage some practitioners from using the 1:$M$ matching estimator. In this work, we address this limitation by showing that the inefficiency of this estimator can typically be reduced --- sometimes substantially --- by augmenting the propensity score model using a carefully-constructed synthetic covariate. Our proposal is a natural extrapolation of the fact that, even when the propensity score is known, the asymptotic variance of the matching estimator can be reduced by instead estimating the propensity score, as shown by \cite{abadie2016}. Here, we characterize the maximal efficiency gain possible through propensity score model augmentation, and show that this gain can be achieved with a one-dimensional augmentation. The augmented 1:$M$ estimator we propose specifically leverages information from available prognostic variables, baseline covariates unrelated to treatment level but predictive of the outcome, which are often ignored by traditional matching methods since they do not feature in the propensity score model. Finally, we also provide results suggesting that certain implementations of the 1:$M$ matching estimator may in fact be (nearly) nonparametric efficient.
 
This paper is structured as follows. In Section \ref{sec:secc2}, we review existing results on the large-sample behavior of the 1:$M$ matching estimator. In Section \ref{sec:secc3}, we study propensity score model augmentation, establish the maximal efficiency gain possible through augmentation, and use these results to propose a novel augmented 1:$M$ matching estimator of the average treatment effect. We also establish the large-sample theory of our proposal. In Section \ref{sec:secc4}, we show results from numerical studies that illustrate potential efficiency gains as well as the behavior of our proposed estimator. Then, in Section \ref{sec:secc5}, we illustrate the use of our method through an application in which we estimate the causal effect of circumcision on HIV infection using data from a HIV vaccine trial. We conclude with a discussion of our findings in Section \ref{sec:secc6}. All proofs of technical results are provided in the Supplementary Material.

\section{Review of results on the 1:$M$ matching estimator} \label{sec:secc2}

\subsection{Statistical setup} \label{subssec:statsetup}

Suppose that the available data consist of independent observations $O_1,O_2,\ldots,O_n$ from an unknown distribution $P_0$. Here, the data unit $O_i:=(V_i, A_i, Y_i)$ represents the information gathered on patient $i$, with  $V_i\in\mathcal{V}\subseteq \mathbb{R}^p$ a vector of baseline covariates, $A_i\in\{0,1\}$ the treatment level received, and $Y_i\in\mathcal{Y}\subseteq \mathbb{R}$ the outcome of interest. For convenience, we also define $W_i:=(1,V_i^{\top})^{\top}\in\mathcal{W}\subseteq \mathbb{R}^{p+1}$. For each covariate level $w\in\mathcal{W}$, we denote the propensity score evaluated at $w$ by $\pi_0(w) := P_0(A = 1\,|\,W=w)$. Throughout this article, we assume nonparametric models for the conditional distribution of $Y$ given $(A,W)$ and for the marginal distribution of $W$. Furthermore, we assume that the true propensity score follows a logistic regression model, namely that $\pi_0=\pi_{\theta_0}$ for some $\theta_0\in\mathbb{R}^{p+1}$, where for each $\theta\in\mathbb{R}^{p+1}$ we define $\pi_{\theta}:w\mapsto \expit(\theta^{\top} w)$. We will denote by $\mathcal{M}_0$ this propensity score model. Relaxation of this modeling assumption is discussed in Section \ref{sec:secc6}.

The causal estimand we focus on in this paper is the average treatment effect. To define it precisely, we make use of the potential outcomes framework, which allows us to define treatment effects \citep{rubin1974}. For $a \in \{0,1\}$, we denote by $Y(a)$ the potential outcome under treatment level $a$. Then, the average treatment effect is given by $\mathbb{E}\left[Y(1) - Y(0)\right]$, where $\mathbb{E}$ denotes expectation of the true joint distribution of counterfactual outcomes. We denote the outcome regression and propensity-reduced outcome regression, respectively, as $\mu_0(a,w):=E_0\left(Y\,|\,A=a,W=w\right)$ and $\bar{\mu}_0(a,p):=E_0\left\{Y\,|\, A=a,\pi_0(W)=p\right\}$. Here and below, we use the notation $E_0$, $var_0$ and $cov_0$ to denote, respectively, an expectation, variance and covariance computed under the true sampling distribution $P_0$. Under the following standard causal conditions: \begin{enumerate}[i.]
    \item (no interference) any patient's treatment level is independent of any other patient's potential outcomes;
    \item (consistency) the observed outcome is $Y=AY(1)+(1-A)Y(0)$;
    \item (exchangeability) $Y(a)$ and $A$ are independent given $W$ $P_0$-almost surely for each $a\in\{0,1\}$;
    \item (positivity) $\pi_0(W)\in(0,1)$ $P_0$-almost surely;
\end{enumerate}the average treatment effect can be identified as\[\psi_0:=E_0\left\{\mu_0(1,W)-\mu_0(0,W)\right\}=E_0\left\{\bar{\mu}_0(1,\pi_0(W))-\bar{\mu}_0(0,\pi_0(W))\right\}.\] 
The above is a summary of the distribution $P_0$ of the observed data unit $O$.

\subsection{Definition of the 1:$M$ matching estimator}

The identification formula above motivates matching on the propensity score instead of the entire covariate vector. If $\pi$ is the propensity score used to match observations, the $\pi$--specific 1:$M$ matching estimator of the average treatment effect can be written as
\begin{equation*}
\psi_n(\pi):= \frac{1}{n}\sum_{i=1}^{n}(2A_i-1)\left[Y_i - \frac{1}{M}\sum_{j=1}^nI\{j \in J_M(i,\pi)\}\,Y_j\right],
\end{equation*}
where we denote by $J_M(i,\pi)$ the matched set for the $i^{th}$ observation based on the propensity score $\pi$
$$\left\{j: A_j = 1-A_i, \textstyle\sum_{k=1}^{n} I\{A_k = 1-A_i\}I\left\{|\pi(W_j) -\pi(W_i) |\leq |\pi(W_k) -\pi(W_i) |\right\} \leq M \right\}.$$ Writing the number $K_{M,\pi}(i):= \sum_{j= 1}^{n}I\{i \in J_M(j,\pi)\}$  of times observation $i$ occurs as a match when propensity score $\pi$ is used, the $\pi$--specific matching estimator can also be represented as 
\begin{equation*}
\psi_n(\pi) = \frac{1}{n}\sum_{i=1}^{n}(2A_i-1)\left\{1 + \frac{K_{M,\pi}(i)}{M}\right\}Y_i\ .
\end{equation*}
Since the true index value $\theta_0$ of the propensity score $\pi_0$ in the logistic regression model $\mathcal{M}_0$ is typically unknown, an estimator of $\theta_0$ must be used instead. Suppose that $\theta_n := \argmax_{\theta}\ell_n(\theta)$ is the maximum likelihood estimator of $\theta_0$, where 
$\ell_n(\theta) := \sum_{i = 1}^n\left[ A_i \log\pi_{\theta}(W_i)) + (1-A_i)\log\{1-\pi_{\theta}(W_i)\}\right]$ is the log-likelihood of $\theta$. Then, the resulting 1:$M$ matching estimator available for use is $\psi_n:=\psi_n(\pi_{\theta_n})$.

\subsection{Large-sample behavior of the 1:$M$ matching estimator}

The large-sample properties of $\psi_n$ and its oracle counterpart $\psi_n(\pi_0)$ were carefully established in \cite{abadie2016} and \citet{abadie2006}, respectively. Here, we summarize two of their key findings, which our proposal heavily builds upon.

The first finding we focus on provides an asymptotic distributional result for the $\pi_0$--specific 1:$M$ matching estimator. Specifically, under some regularity conditions, \citet{abadie2006} showed that $n^{1/2}\left\{\psi_n(\pi_0)-\psi_0\right\}$ tends to a normal random variable with mean zero and variance $\sigma^2_M:=\sigma^2_{1}+\sigma^2_{2,M}$, where the variance components are defined as
\begin{align*}
\sigma^2_1\,&:=\, E_0\left[\frac{\bar{\sigma}^2(1,\pi_0(W))}{\pi_0(W)}+\frac{\bar{\sigma}^2(0,\pi_0(W))}{1-\pi_0(W)}+\left\{\bar{\mu}_0(1,\pi_0(W))-\bar{\mu}_0(0,\pi_0(W))-\psi_0\right\}^2\right],\\
\sigma^2_{2,M}\,&:=\,\frac{1}{2M}\, E_0\left[\bar{\sigma}^2(1,\pi_0(W))\left\{\frac{1}{\pi_0(W)}-\pi_0(W)\right\}+\bar{\sigma}^2(0,\pi_0(W))\left\{\frac{1}{1-\pi_0(W)}-1+\pi_0(W)\right\}\right],
\end{align*} and for each $a\in\{0,1\}$ and $w\in\mathcal{W}$, $\bar{\sigma}^2(a,\pi_0(w)):=var_0\{Y\,|\,A=a,\pi_0(W)=\pi_0(w)\}$ denotes the conditional variance of $Y$ given $A=a$ and $\pi_0(W)=\pi_0(w)$ computed under  $P_0$. In particular, this result highlights that, as expected,  larger values of $M$ render $\psi_n(\pi_0)$ relatively less variable in large samples, but that, irrespective of $M$, $\psi_n(\pi_0)$ is not nonparametric efficient except in  trivial cases.

The second finding central to the developments in this paper describes the extension of the above distributional results to the 1:$M$ matching estimator based on an estimated propensity score index. For technical reasons pertaining to a certain non-smoothness in the behavior of matching estimators \citep{andreou2012alternative}, the estimator studied in \cite{abadie2016} is based not on the maximum likelihood estimator $\theta_n$ but rather on a discretized version $\theta_{n,k}$ of $\theta_n$, defined as $\theta_{n,k}:=\round{kn^{1/2}\theta_n}/(kn^{1/2})$, where $k$ is a user-specified discretization constant and $\round{\cdot}:\mathbb{R}^{p+1}\rightarrow \mathbb{Z}^{p+1}$ is the componentwise nearest-integer function. If either $k$ or $n$ is large, $\theta_{n,k}$ is close to $\theta_n$ in each sample. Under regularity conditions, they established, for each $z\in\mathbb{R}$, that
\begin{equation}
\label{eq:asymtauhat}
\lim_{k \rightarrow\infty}\lim_{n \rightarrow \infty }P_0\left[n^{1/2}\left\{\psi_n(\pi_{\theta_{n,k}})-\psi_0\right\}/\sigma_{M,*} \leq z\right] = \Phi(z)
\end{equation} with asymptotic variance given by $\sigma^2_{M,*}:=\sigma_M^2-c(\theta_0)^{\top} \mathscr{I}(\theta_0)^{-1}c(\theta_0)$, where $\Phi$ denotes the standard normal distribution function, $\mathscr{I}(\theta_0):=E_0\left[\pi_0(W)\{1-\pi_0(W)\}WW^{\top}\right]$ is the Fisher information for $\theta$ at $\theta=\theta_0$ in the logistic regression model, assumed to be invertible, and the $\mathbb{R}^{p+1}$--valued vector $c(\theta_0)$ is defined as 
\begin{align*}
    c(\theta_0)\ &:=\ E_0\left[\pi_0(W)\, cov_0\left\{W,\mu_0(0,W)\,\middle|\,A=0,\pi_0(W)\right\}\right]\\
    &\hspace{.5in}+E_0\left[\{1-\pi_0(W)\}\, cov_0\left\{W,\mu_0(1,W)\,\middle|\,A=1,\pi_0(W)\right\}\right].
\end{align*}
This result suggests that the 1:$M$ matching estimator $\psi_{n}(\pi_{\theta_{n,k}})$ based on a (discretization-regularized) estimator of the propensity score is expected to be more efficient than the $\theta_0$--specific  1:$M$ matching estimator $\psi_n(\pi_0)$, with a reduction in asymptotic variance given by $c(\theta_0)^{\top}\mathscr{I}(\theta_0)^{-1}c(\theta_0)\geq 0$. Even so, $\psi_{n}(\pi_{\theta_{n,k}})$ typically still does not achieve the nonparametric efficiency bound.

\section{Propensity score augmentation for the 1:$M$ matching estimator}\label{sec:secc3}
\subsection{Using model augmentation to gain efficiency}

Given that estimation of the propensity score increased the asymptotic efficiency of the 1:$M$ matching estimator relative to the (oracle) 1:$M$ matching estimator based on the true propensity score, it is natural to wonder whether further gains are possible by augmentation of the propensity score model. To investigate this question, we first make use of the second result of \cite{abadie2016} stated above to characterize potential gains resulting from the use of a given augmentation of the propensity score model.

Suppose  $q:=(q_1,q_2,\ldots,q_m)^{\top}$ is an arbitrary $m$-dimensional augmentation function, with $q_j:\mathcal{W}\rightarrow\mathbb{R}$ satisfying that $var_0\{q_j(W)\}<\infty$ for each $j=1,2,\ldots,m$. We consider the 1:$M$ matching estimator based on a propensity score estimated using the augmented logistic regression model $\mathcal{M}(q):=\left\{\pi_{q,\theta,\gamma}:\theta\in\mathbb{R}^{p+1},\gamma\in\mathbb{R}^m\right\}$, where we define \[\pi_{q,\theta,\gamma}:w\mapsto \expit\big\{\theta^{\top} w+\gamma^{\top} q(w)\big\}\ .\] This corresponds to fitting a regression model not only with covariate vector $W$ but also a synthetic covariate vector $q(W)$ created by transformation of $W$. For any choice of $q$, the propensity score model $\mathcal{M}(q)$ is correctly specified since $\pi_0$ is recovered by setting $\vartheta:=(\theta,\gamma)$ equal to $\vartheta_{0,m}:=(\theta_0,0_m)$, where $0_m$ is an $m$--dimensional zero vector. To simplify notation, we denote $\vartheta_{0,m}$ by $\vartheta_0$ as we typically consider the number $m$ of augmentation covariates as fixed. We note that $\pi_{q,\vartheta_{0}}=\pi_0$ irrespective of $q$. Among regularity conditions needed for identification in this augmented model, the information matrix for $\vartheta$ in $\mathcal{M}(q)$ at $\vartheta=\vartheta_{0}$, given by \[\mathscr{I}_{q}(\vartheta_0):=E_0\left[\pi_0(W)\{1-\pi_0(W)\}\begin{bmatrix}WW^{\top} & Wq(W)^{\top}\\ q(W)W^{\top} & q(W)q(W)^{\top} \end{bmatrix}\right],\] is assumed to be invertible.  We denote by $\vartheta_{q,n}:=\argmax_{\vartheta}\ell_{q}(\vartheta;D_n)$ the maximum likelihood estimator of $\vartheta_{0}$, where $\ell_{q}(\vartheta;D_n):= \sum_{i = 1}^n\left[ A_i \log\pi_{q,\vartheta}(W_i)) + (1-A_i)\log\{1-\pi_{q,\vartheta}(W_i)\}\right]$ is the log-likelihood of $\vartheta$ based on model $\mathcal{M}(q)$ for the propensity score. Using once more the results of \cite{abadie2016}, we find that the (sample size-scaled) asymptotic variance of the 1:$M$ matching estimator $\psi_{n,q,k}:=\psi_n(\pi_{q,\vartheta_{q,n,k}})$ using an estimated propensity score based on model $\mathcal{M}(q)$ and a discretized regularization is given by $\sigma^2_{q,M}:=\sigma^2_M-\mathrm{gain}(q)$, where we define $\mathrm{gain}(q):=c_{q}(\vartheta_0)^{\top}\mathscr{I}_{q}(\vartheta_0)^{-1}c_{q}(\vartheta_0)$ and write $c_{q}(\vartheta_0):=(c(\theta_0)^{\top},\bar{c}_{q}(\vartheta_0)^{\top})^{\top}$ with
\begin{align*}
    \bar{c}_{q}(\vartheta_0)\ &:=\ E_0\left[\pi_0(W)\, cov_0\left\{q(W),\mu_0(0,W)\,\middle|\,A=0,\pi_0(W)\right\}\right]\\
    &\hspace{.5in}+E_0\left[\{1-\pi_0(W)\}\, cov_0\left\{q(W),\mu_0(1,W)\,\middle|\,A=1,\pi_0(W)\right\}\right].
\end{align*} This implies, in particular, that there is no loss but only potential gains in terms of asymptotic variance from augmenting the propensity score model with a covariate vector $q(W)$.

\subsection{Maximizing the efficiency gain}

Given the potential efficiency gains resulting from augmentation by any covariate vector $q(W)$, we consider constructing a 1:$M$ matching estimator with a propensity score estimated using an augmented model $\mathcal{M}(q)$. However, as the expression for $\textrm{gain}(q)$ implies, the efficiency gain possible depends on the choice of $q$. Thus, many candidates for $q$ could be considered to increase or even perhaps maximize the efficiency gain. In the following theorem, we display an augmented propensity score model through which the efficiency gain is maximized. Before stating our formal result, we define the transformation $h_0:\mathcal{W}\rightarrow\mathbb{R}$ pointwise as 
\begin{equation*}
h_0(w) := \frac{\mu_0(1, w)  - \bar{\mu}_0(1,\pi_0(w))}{\pi_0(w)} + \frac{\mu_0(0, w)  - \bar{\mu}_0(0,\pi_0(w))}{1-\pi_0(w)}\ .
\end{equation*} We also denote by $\mathcal{Q}_0$ the set obtaining by collecting, for each $m\in\{1,2,\ldots\}$, each augmentation function $q:\mathcal{W}\rightarrow \mathbb{R}^{m}$ such that $E_0\{ q_j(W)^2\}<\infty$ for each $j=1,2,\ldots,m$ and $\mathscr{I}_{q}(\vartheta_0)$ is invertible. The following theorem describes the efficiency gain resulting from using \[\mathcal{M}(h_0)=\left\{w\mapsto \expit\big\{\theta^{\top} w+\gamma h_0(w)\big\}:\theta\in\mathbb{R}^{p+1},\gamma\in\mathbb{R}\right\},\]the propensity score model augmented with the synthetic covariate $h_0(W)$. The theorem requires four onditions to hold. The first three are as follows:

\setcounter{condition}{0}
\begin{condition}
\label{cond:pos:appendix}
Positivity: for some constants $u$ and $l$ such that $0 < l < u <1$, we have that $l < \pi_0(W) < u$ almost surely.
\end{condition}

\begin{condition} 
\label{cond:qmd:appendix}
Boundedness: $\mathcal{W}$ is a bounded subset of $\mathbb{R}^{p+1}$. 
\end{condition}

\begin{condition}
\label{cond:iden:appendix}
Identifiability: (i) consistency: for each observation $Y = AY(A) + (1-A)Y(1-A)$; (ii) no unmeasured confounding: for $a = 0,1$ $Y(a) \indep A | W $; (iii) no interference:  the random variables $Y(1), Y(0)$ and $A$ are such that for every $i\neq j$, we have that $Y_i(0) \indep A_j$, and $Y_i(1) \indep A_j$.
\end{condition}

The fourth condition --- referred to as Condition \ref{cond:bounds:appendix}--- imposes restrictions on the conditional mean of the outcome given the augmented propensity score and treatment status. More specifically, for $a \in \{0,1\}$, the condition requires that $E_0[Y|A= a, \pi_{\theta_0}(W) = p]$ is L-Lipschitz continuous in $p$ and uniformly bounded over $\theta_0$ and $p$. Additionally, Condition \ref{cond:bounds:appendix} also requires the moment condition that there exists a $\delta>0$ such that $E[|Y|^{2+\delta}\mid A = a, \pi_0(W) = p]$ is uniformly bounded over $\theta_0$ and $p$. For brevity, the detailed form of the condition is deferred to Supplementary Appendix~\ref{sec:prelim}.

\begin{theorem}[Gain in efficiency from using augmentation]
\label{theo1}
Suppose Conditions 
\ref{cond:pos:appendix}-\ref{cond:iden:appendix} and \ref{cond:bounds:appendix} hold. The following statements hold true:
\begin{enumerate}[(a)]
\item $\mathrm{gain}(h_0)=\sup_{q\in\mathcal{Q}_0}\mathrm{gain}(q)$;
\item $\mathrm{gain}(h_0)=E_0\left[\pi_0(W)\{1-\pi_0(W)\}\{h_0(W)\}^2\right]$.
\end{enumerate}
\end{theorem}
This result is significant since it indicates that a one-dimensional augmentation covariate suffices to achieve maximal efficiency gain among all possible augmented propensity score models. Moreover, it indicates that if $\mu_0(a,w)=\bar{\mu}_0(a,\pi_0(w))$ for $P_0$--almost every $(a,w)$ --- this occurs, for example, if $W$ is one-dimensional --- then there is no efficiency to be gained, and in fact, the 1:$M$ matching estimator has the same asymptotic variance irrespective of whether the true propensity score is used or instead estimated. The optimal gain depends on the unknown distribution $P_0$ --- we provide some insights on when larger gains can be expected in Section \ref{subsec:releffgd}.

The following corollary provides a comparison between the asymptotic variance of the augmentation-based 1:$M$ matching estimator and the nonparametric efficiency bound \[\sigma^2_{NP}:=E_0\left[\frac{\sigma^2(1,W)}{\pi_0(W)}+\frac{\sigma^2(0,W)}{1-\pi_0(W)}+\left\{\mu_0(1,W)+\mu_0(0,W)-\psi_0\right\}^2\right],\]where $\sigma^2(a,w):=var_0\left(Y\,|\,A=a,W=w\right)$ is the conditional variance of $Y$ given $(A,W)=(a,w)$.

\begin{corollary}[Gain in efficiency from using optimal univariate augmentation] Under the conditions of Theorem~\ref{theo1}, the difference $\delta_M:=\sigma^2_M-\mathrm{gain}(h_0)-\sigma^2_{NP}$
between the asymptotic variance of the optimally-augmented 1:$M$ matching estimator and the nonparametric efficiency bound is given by
\begin{align*}
&\delta_M\,=\,\frac{1}{2M}\,E_0\left[\left\{\frac{1}{\pi_0(W)}-\pi_0(W)\right\}\left\{\sigma^2(1,W)+\zeta(1,W)\right\}\right.\\
&\hspace{1.2in}\left.+\left\{\frac{1}{1-\pi_0(W)}-1+\pi_0(W)\right\}\left\{\sigma^2(0,W)+\zeta(0,W)\right\}\right]\,\geq\,0\ ,
\end{align*}where we define $\zeta(a,w):=var_0\left\{\mu_0(a,W)\,|\,\pi_0(W)=\pi_0(w)\right\}$.
\end{corollary}

In particular, this result indicates that the asymptotic inefficiency of the optimally-augmented 1:$M$ matching estimator --- the amount by which its (sample size-scaled) asymptotic variance exceeds the nonparametric efficiency bound --- can be made arbitrarily small by choosing a large value of $M$. However, our theoretical guarantee pertains to the augmented matching estimator with fixed $M$, and does not consider the case where $M$ increases with sample size. In contrast, \cite{he2024propensity} study the asymptotics of an unaugmented matching estimator based on a diverging number of matches, showing the asymptotic variance of such an estimator is lower than that of the unaugmented matching estimator, regardless of $M$. Additionally, under the $M\rightarrow \infty$ regime, the unaugmented matching estimator can attain the nonparametric efficiency bound under the condition that $\mu_0(a, w) = \bar{\mu}_0(a,\pi_0(w))$ and $\sigma^2(a,w)= \bar{\sigma}^2(a,w)$, for $a =0,1$ and all $w \in \mathcal{W}$. When $\pi_0$ is injective this condition directly holds, but otherwise it can be quite strong, especially if $W$ is multivariate. In contrast, our result does not rely on $\mu_0(a, w) = \bar{\mu}_0(a,\pi_0(w))$ or $\sigma^2(a,w)= \bar{\sigma}^2(a,w)$ for $\sigma^2_M-\mathrm{gain}(h_0)$ to be arbitrarily close to ~$\sigma^2_{NP}$. In future work, it would be interesting to apply the theoretical developments from \cite{he2024propensity} to study whether the optimally-augmented 1:$M$ matching estimator achieves the asymptotic efficiency bound in an appropriate $M\rightarrow\infty$ regime. 

\subsection{Estimating the optimal augmentation covariate}

In practice, since the optimal augmentation function $h_0 $ is unknown, it must be estimated. In this section, we establish that, under regularity conditions, the augmented 1:$M$ matching estimators using an estimator of $h_0$ is asymptotically equivalent to the corresponding estimator using $h_0$ itself.  We show this in two steps. First,  in Proposition \ref{propp:1}, we consider the behavior of the 1:$M$ matching estimator resulting from use of the augmented propensity score model $\mathcal{M}(h_{0n})$, where $h_{01},h_{02},\ldots$ is a sequence of deterministic functions such that $\int \left\{h_{0n}(w)-h_0(w)\right\}^2P_{W,0}(dw)$ tends to zero. Here, $P_{W,0}$ denotes the marginal distribution of $W$ implied by $P_0$. This result allows accounting for the fact that, in applications, a proxy of $h_0$ would be used instead of $h_0$. It does not however account for the randomness of this proxy. Thus, in Corollary \ref{corp:2}, we show that the results of Proposition \ref{propp:1} still hold when the proxy used for $h_0$ is a random function $h_n$ constructed on an independent sample such that $\int\left\{h_{n}(w)-h_0(w)\right\}^2P_{W,0}(dw)$ tends to zero $P_0$--almost surely.

Our theoretical results establishing the asymptotic normality of our proposed matching estimator require nine regularity conditions. Conditions \ref{cond:pos:appendix}-\ref{cond:iden:appendix} and \ref{cond:bounds:appendix} have been described in Theorem \ref{theo:theo1} above. 
Condition \ref{cond:qmd2:appendix} pertains to $\mathcal{H}$, a collection of $\mathcal{W}\rightarrow\mathbb{R}$ functions that contains the sequence of augmentation covariate functions utilized in Proposition \ref{propp:1} as well as $h_0$. 

\begin{condition} 
\label{cond:qmd2:appendix}
Conditions on $\mathcal{H}$:
(i) $\mathcal{H}$ is a collection of bounded $\mathcal{W}\rightarrow\mathbb{R}$ functions, so that  \allowbreak
$\sup_{\hd \in\mathcal{H}}\sup_{w\in\mathcal{W}} |\hd(w)|<\infty$. 
(ii) The class of functions $\mathcal{H}$ is totally bounded in $L^2(P_{W,0})$. (iii) The class of functions $\mathcal{H}$ is such that for every $\hd \in \mathcal{H}$, $E(r_\hd(W)r_\hd(W)^{\rm{T}})$ is positive definite.
\end{condition}
Condition \ref{cond:ps:appendix} pertains to the evaluation of the propensity score at $W$, and is also utilized by \cite{abadie2016}.
\begin{condition}
\label{cond:ps:appendix}
The evaluation of the propensity score $\pi$ at $W$ is continuously distributed, with continuous density.
\end{condition}
Condition \ref{cond:mle:appendix} requires identifiability of $\vartheta_0$ and is deferred to Supplementary Appendix ~\ref{sec:prelim}. Finally, Conditions \ref{cond:lambdan:appendix} and \ref{cond:ulan:last:appendix} are utilized to derive the limits of specific terms that arise when deriving asymptotic distribution of the matching estimator. An interpretation of these technical conditions is provided in the Supplementary Material. All supplementary conditions are also stated in the Appendix. 

We employ the strategy of \cite{abadie2016} to derive large-sample distributional results for the 1:$M$ matching estimator based on a discretization-regularized maximum likelihood propensity score estimator within an augmented model.  As before, given a discretization constant $k>0$, we define a partition of cubes with sides of length $(kn^{1/2})^{-1}$ in $\mathbb{R}^{p+2}$. For a given augmentation function $q$, the discretized estimator $\vartheta_{q,n,k}$ is set to equal the midpoint of the cube to which $\vartheta_{q,n}$ belongs. In our notation, the discretized estimator is thus defined as $\vartheta_{q,n,k}:=\round{kn^{1/2}\vartheta_{q,n}}/(kn^{1/2})$, where, similarly as before, $s:\round{\cdot}:\mathbb{R}^{p+2} \rightarrow \mathbb{Z}^{p+2}$ is the componentwise nearest-integer function. The 1:$M$ matching estimator based on a discretization-regularized maximum likelihood propensity score estimator with augmentation functions $h_{0n}$ and $h_n$, respectively, are then given by $\psi_{n,h_{0n},k}:=\psi_n(\pi_{h_{0n},\vartheta_{h_{0n},n,k}})$ and $\psi_{n,h_{n},k}:=\psi_n(\pi_{h_{n},\vartheta_{h_n,n,k}})$. 

Given the sequence $h_{01},h_{02},\ldots$ of deterministic functions approximating the optimal augmentation function $h_0$, Proposition \ref{propp:1} states that $n^{1/2}\left(\psi_{h_{0n},n,k}-\psi_0\right)$ tends to a mean-zero normal distribution with oracle variance $\sigma^2_{M,\mathrm{opt}}:=\sigma^2_M-\mathrm{gain}(h_0)$. 

\begin{proposition}[Guarantee when a fixed sequence of augmentation covariates is used]
\label{propp:1}
Suppose Conditions \ref{cond:pos:appendix}-\ref{cond:ps:appendix} and \ref{cond:bounds:appendix}-\ref{cond:ulan:last:appendix} hold. 
Then, for each $z\in\mathbb{R}$, it holds that
\begin{align*}
\lim_{k \downarrow 0 }\lim_{n \rightarrow \infty }P_0\left[n^{1/2}\left(\psi_{n,h_{0n},k}-\psi_0\right)/\sigma_{M,\mathrm{opt}} \leq z\right] = \Phi(z)\ .
\end{align*}
\end{proposition}

The above result gives a large-sample distributional approximation relevant to the matching estimator based on the discretization-regularized propensity score estimator involving augmentation by $h_{0n}$. While the involved probability sequence is based on sampling from the true distribution $P_0$ at each sample size $n$, a stronger result wherein the probability sequence is based on sampling from certain local perturbations of $P_0$ can also be proven, and is established in Proposition 
\ref{prop1} in the Supplementary Material. This strengthened result ascribes a certain level of regularity to the matching estimator studied, which is particularly relevant given the irregularity that may perhaps be expected given how non-smooth the matching operation is. A key distinction between this result and the main result of \cite{abadie2016} is that the propensity score model used here incorporates a covariate --- in our case, an augmentation covariate --- that is changing with $n$ but eventually settles down. Nevertheless, in practice, it would be the case that $h_0$ is unknown and no deterministic approximating sequence $h_{0n}$ is available, limiting the applicability of the result. The next result overcomes this limitation by allowing use of a random approximating sequence $h_n$ in the propensity score augmentation step.

To allow consideration of a random estimator $h_n$ of the optimal augmentation function $h_0$, we require $h_n$ to be constructed on a different dataset than  the resulting (discretized) maximum likelihood estimator $\vartheta_{h_n,n,k}$ of $\vartheta_{h_n,0}$ in model $\mathcal{M}(h_n)$ and matching estimator $\psi_{n,h_n,k}$. In practice, this implies that the original dataset must be partitioned into two subsets: subset $A_n$ used for estimation of $h_0$, say of size $1<m_n<n$, and subset $B_n$ used for the remainder of the matching estimation procedure, of size $n_{\mathrm{eff}}:=n-m_n$, with $m_n$ tending to infinity with $n$.  We are then able to prove that the same large-sample distributional result holds for the 1:$M$ matching estimator based on a discretization-based estimate of the propensity score model augmented with an estimator of the optimal augmentation function as with the (oracle) optimal augmentation function itself, except for the loss of effective sample size resulting for the need to estimate $h_0$ on a separate sample.

\begin{corollary}[Guarantee when augmentation covariate is estimated]
\label{corp:2}
Suppose Conditions \ref{cond:pos:appendix}-\ref{cond:ps:appendix} and \ref{cond:bounds:appendix}-\ref{cond:ulan:last:appendix} hold, and that $\int \left\{h_n(w) - h_0(w)\right\}^2P_{W,0}(dw)$ tends to zero $P_0$--almost surely. Then, for each $z\in\mathbb{R}$, it $P_0$--almost surely holds that
\begin{align*}
  \lim_{k \rightarrow\infty}\lim_{n \rightarrow \infty }P_0\left[n_{\mathrm{eff}}^{1/2}\left(\psi_{n,h_n,k}-\psi_0\right)/\sigma_{M,\mathrm{opt}} \leq z\,\middle|\,A_n\right] = \Phi(z)\ .
\end{align*}
\end{corollary}

To estimate the optimal augmentation function $h_0$, we must first estimate the outcome regression $\mu_0$ and the propensity-reduced outcome regression $\bar{\mu}_0$. Flexible regression techniques can be used for this purpose. The nuisance functions $\mu_0$ and $\bar{\mu}_0$ can naturally be estimated in sequence. First, an estimate of $\mu_0(a,w)$ can be obtained by using any regression technique for learning the mean of outcome $Y$ given $W=w$ within the subset of observations with $A=a$. Regression tools used can include parametric, semiparametric, or nonparametric approaches, or a combination thereof pooled via an ensembling strategy. Then, based on the fact that \[\bar{\mu}_0(a,p)=E_0\left\{Y\,\middle|\,A=a,\pi_0(W)=\pi_0(w)\right\}=E_0\left\{\mu_0(a,W)\,\middle|\,A=a,\pi_0(W)=\pi_0(w)\right\},\]an estimate of $\bar{\mu}_0(a,p)$ can be obtained using a mean regression of outcome $\mu_n(A,W)$ on $\pi_n(W)$ within the subset of observations with $A=a$, where $\pi_n$ is a model-based estimator of the propensity score. Here again, a variety of regression tools could be considered. This nested regression strategy is more likely to result in estimates of $\mu_0(a,w)$ and $\bar{\mu}_0(a,\pi_0(w))$ that are compatible with each other, and hence may be preferable in finite samples compared to the use of separate regressions to estimate these nuisance functions. We summarize the steps involved in constructing the final matching estimator in Algorithm \ref{al1}.

\begin{algorithm}[ht]
\begin{algorithmic}[1]
    \vspace{.15in}
    \State \textbf{partition data into two subsets $A_n$ and $B_n$;}
    \vspace{.1in}
    \State \textbf{using subset $A_n$ of data, obtain $h_n$ as follows:}
    \vspace{.05in}
    \State \quad compute MLE $\theta_n$ of $\theta_0$ in $\mathcal{M}_0$ by performing logistic regression of $A$ on $W$, and set $\pi_n:=\pi_{\theta_n}$;
    \State \quad for $a \in \{0, 1\}$, obtain $\mu_n(a,\cdot)$ by regressing $Y$ on $W$ using observations with $A=a$;
    \State \quad for $a \in \{0, 1\}$, obtain $\bar{\mu}_n(a,\cdot)$ by regressing $\mu_n(A,W)$ on $\pi_{n}(W)$ using observations with $A=a$.
    \State \quad define  $h_n:w\mapsto \{\mu_n(1,w)-\bar{\mu}_n(1,\pi_n(w)\}/\pi_n(w)+\{\mu_n(0,w)-\bar{\mu}_n(0,\pi_n(w)\}/\{1-\pi_n(w)\}$;
    \vspace{.1in}
    \State \textbf{using subset $B_n$ of data, compute $\psi_{n,h_n,k}$ as follows:}
    \vspace{.05in}
    \State \quad compute MLE $\vartheta_{h_n,n}$ of $\vartheta_{0}$ in $\mathcal{M}(h_n)$ by performing logistic regression of $A$ on $(W,h_n(W)$;
    \State \quad compute 1:$M$ matching estimator $\psi_{n,h_n,k}$ based on matching using $\pi_{h_n,\vartheta_{h_n,n,k}}$.
  \vspace{.1in}
\end{algorithmic}
\caption{Optimally-augmented propensity score 1:$M$ matching with sample-splitting.}
\label{al1}
\end{algorithm}

In practice, the size $m_n$ of $A_n$ can be chosen to be relatively small compared to the size $n_{\mathrm{eff}}$ of $B_n$ since, while $h_n$ must be consistent in a suitable sense, there is no minimal convergence rate required. In our simulation studies, we selected $m_n=0.05\times n$ to minimize the variance inflation induced by sample-splitting. We note that relatively poor estimation of $h_0$, due, for example, to the use of insufficiently flexible regression techniques or an overly small value of $m_n$, is not expected to be problematic since \textit{any} augmentation of the propensity score model is expected to possibly help but never hurt the performance of the resulting matching estimator. 

\section{Numerical studies}\label{sec:secc4}

\subsection{Motivating questions}

The theory provided in this paper suggests that efficiency gains are possible via propensity score model augmentation, and that a regularization of the resulting 1:$M$ matching estimator has desirable large-sample properties. Nevertheless, there are critical questions that are best addressed through numerical studies, including: \begin{enumerate}
    \item how substantial can the theoretical efficiency gains derived from augmentation be, and when can meaningful gains be expected?
    \item does the proposed matching estimator achieve predicted efficiency gains in practice, and is the normal approximation derived useful to describe the finite-sample behavior of the estimator?
    \item does the use of sample-splitting to estimate the optimal augmentation function and calculate the resulting estimator appear critical to the validity of the derived normal approximation?
\end{enumerate}
We begin by investigating the first question in the context of a simple data-generating mechanism for which certain closed-form results can be derived based on our analytic expressions. We then perform simulation studies using a collection of scenarios in order to provide additional insight on the first question and to address the second and third questions as well.

\subsection{Illustration of predicted gains}\label{subsec:releffgd}

The efficiency gain resulting from propensity score model augmentation depends on various aspects of the data-generating mechanism $P_0$. Based on the form of the optimal augmentation function, we expect that the gain will tend to be higher when the contrast between the outcome regression and propensity-reduced outcome regression is large. This occurs, for example, when $W$ includes variables that are strongly associated with the outcome but less so with treatment. 

We considered a simple data-generating mechanism based upon which to perform some analytic and numeric efficiency gain computations. Specifically, we considered a setting in which $W:=(W_1,W_2)$ consists of two independent mean-zero normal random variables with unit variance, $A$ follows a Bernoulli distribution with success probability $\pi_0(w_1,w_2):=\expit\left(\theta_0+\theta_1w_1+\theta_2w_2\right)$ conditionally upon $W=(w_1,w_2)$, and $Y$ is normal random variable with mean \[\mu_0(a,w):=a\left(\beta_0+\beta_1w_1+\beta_2w_2\right)+(1-a)\left(\gamma_0+\gamma_1w_1+\gamma_2w_2\right)\] and variance $\sigma^2$ conditionally upon $W=(w_1,w_2)$ and $A=a$. The average treatment effect equals $\beta_0-\gamma_0$. Furthermore, the propensity-reduced outcome regression $\bar{\mu}_0(a,\pi_0(w))$ can be shown to equal \[a\left\{\beta_0+(\beta_1\theta_1+\beta_2\theta_2)\left(\frac{\theta_1w_1+\theta_2w_2}{\theta_1^2+\theta_2^2}\right)\right\}+(1-a)\left\{\gamma_0+(\gamma_1\theta_1+\gamma_2\theta_2)\left(\frac{\theta_1w_1+\theta_2w_2}{\theta_1^2+\theta_2^2}\right)\right\},\] which implies, for example, that \[\mu_0(1,w)-\bar{\mu}_0(1,\pi_0(w))=\frac{\theta_2(\beta_1\theta_2-\beta_2\theta_1)w_1+\theta_1(\beta_2\theta_1-\beta_1\theta_2)w_2}{\theta_1^2+\theta_2^2}.\]A similar expression holds for  $\mu_0(0,w)-\bar{\mu}_0(0,\pi_0(w))$, replacing $\beta_j$ by $\gamma_j$ for each $j=1,2,3$. Suppose that $W_2$ is a precision variable, in the sense that it has an effect on $Y$ but not on $A$; in the context of this data-generating mechanism, this is achieved by setting $\theta_2=0$ but ensuring that either $\beta_2\neq 0$ or $\gamma_2\neq 0$. In this case, the difference in outcome regressions simplify to $\mu_0(1,w)-\bar{\mu}_0(1,\pi_0(w))=\beta_2w_2$ and $\mu_0(0,w)-\bar{\mu}_0(0,\pi_0(w))=\gamma_2 w_2$. For simplicity, we may consider the case in which $\gamma_2=\beta_2$ so that there is no effect modification by $W_2$. In this case, we can explicitly calculate the nonparametric efficiency bound to be  \[\sigma^2_{NP}=2\sigma^2\left\{1+\exp(\theta_1^2/2)\right\}+(\gamma_1-\beta_1)^2\ .\] The asymptotic variances of the 1:$M$ matching estimator based on unaugmented and optimally augmented propensity score models are, respectively, $\sigma^2_{M,*}=\sigma^2_{M}-\mathrm{gain}(\varnothing)$ and $\sigma^2_{M,\mathrm{opt}}=\sigma^2_M-\mathrm{gain}(h_0)$, where we can explicitly compute that \[\sigma^2_M=\sigma^2_{NP}+2\beta_2^2\left\{1+\exp(\theta_1^2/2)\right\}+\delta_M\]with $\delta_M:=(\sigma^2+\beta_2^2)\,\{1+2\exp(\theta_1^2/2)\}/(2M)$, and furthermore, as expected, that\[
\mathrm{gain}(\varnothing)=\frac{\beta_2^2}{E_0\left[\pi_0(W)\{1-\pi_0(W)\}\right]}\ \leq\  \beta_2^2\,E_0\left[\frac{1}{\pi_0(W)\{1-\pi_0(W)\}}\right]=\mathrm{gain}(h_0)\ .\] Additionally, further simplification yields that  $\mathrm{gain}(h_0)=2\beta_2^2\,\{1+\exp(\theta_1^2/2)\}$. We thus find that the relative efficiency of the unaugmented 1:$M$ matching estimator to its augmented counterpart equals \begin{align*}
    \mbox{relative efficiency}\ &=\ 1+\bar{\beta}_2^2\left[\frac{2\{1+\exp(\theta_1^2/2)\}-1/E_0\left[\pi_0(W)\{1-\pi_0(W)\}\right]}{2\{1+\exp(\theta_1^2/2)\}+(\bar{\gamma}_1-\bar{\beta}_1)^2+\frac{1}{2M}(1+\bar{\beta}_2^2)\{1+2\exp(\theta_1^2/2)\}}\right],
\end{align*}where we have defined the standardized coefficients $\bar{\beta}_1:=\beta_1/\sigma$, $\bar{\beta}_2:=\beta_2/\sigma$ and $\bar{\gamma}_1:=\gamma_1/\sigma$. In Figure \ref{releff_example}, for different choices of $M$, we display the relative efficiency as a function of $\theta_1$ for $\bar{\beta}_2=1$ and as a function of $\bar{\beta}_2$ for $\theta_1=1$. In these displays, for simplicity, we have set $\bar{\beta}_1=\bar{\gamma}_1$, corresponding to the absence of effect modification by $W_1$. We observe that the relative gains in efficiency resulting from the use of augmentation compared to standard propensity score-based 1:$M$ matching grow with the strength of the associations between $Y$ and $W_2$  and between $A$ and $W_1$, and that these gains can be substantial even under moderate associations. If any of these two associations is null, then no gains in efficiency result from augmentation. We also note that, as expected, the relative efficiency of augmentation grows with the number $M$ of matches though with diminishing marginal improvements as $M$ grows.
\begin{figure}
\label{releff_example}
\centering
\includegraphics[scale=0.8]{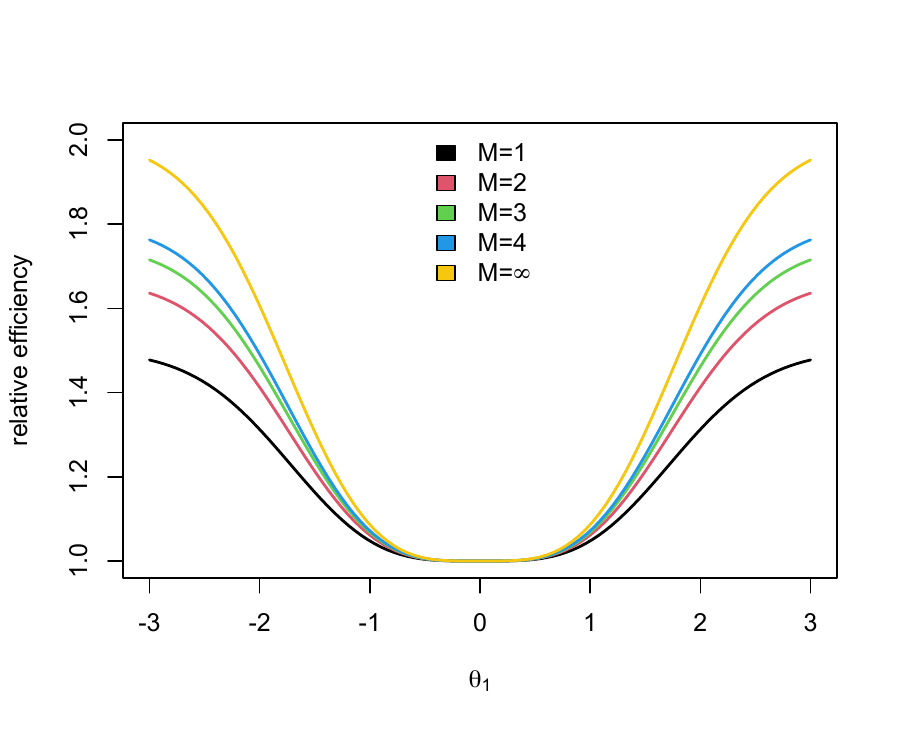}\vspace{-.6in}
\includegraphics[scale=0.8]{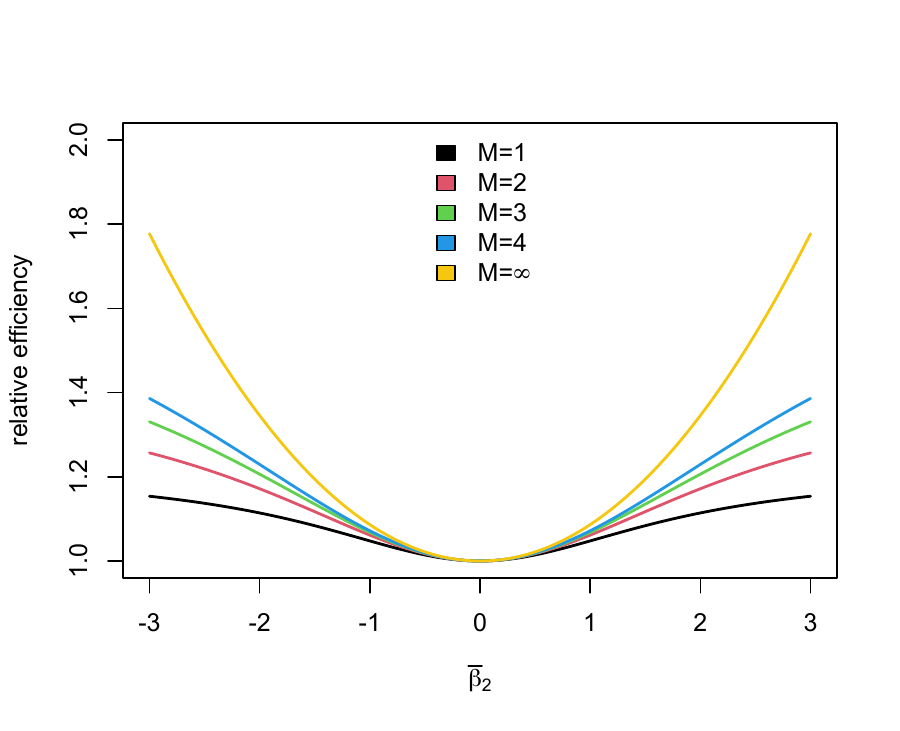}
\caption{Relative efficiency of 1:$M$ matching estimator with optimal propensity augmentation versus unaugmented propensity in analytic example described in Section 4.1. The left plot displays the relative efficiency as a function of $\theta_1$ for $\bar{\beta}_2=1$. The right plot displays the relative efficiency as a function of $\beta_2$ for $\theta_1=1$. In both plots, we have set  $\bar{\beta}_1=\bar{\gamma}_1=1$, and show the relative efficiency value for different values of $M$.}
\end{figure}

\subsection{Simulation studies with sample-splitting}

To verify that the proposed estimation procedure behaves as suggested by our theoretical results, we ran simulations under each of four different scenarios. In each simulation scenario, data were generated as follows. First, a covariate vector $W:=(W_1,W_2)$ was generated, where $W_1$ is a uniform random variable that takes values between -2 and 2, and $W_2$ is a standard binary random variable. Then, given $W=(w_1,w_2)$, a Bernoulli random variable $A$ with success probability $\pi_0(w_1,w_2)$ was drawn. Finally, given $W=(w_1,w_2)$ and $A=a$, a normal random variable $Y$  with mean $\mu_0(a,w)$ and unit variance was generated. Table \ref{datagen} provides the specific form of $\pi_0(w_1,w_2)$ and $\mu_0(a,w)$ defining the four simulation scenarios considered. 
They encompass settings where: both covariates are confounders (scenario 1), one of the two measured covariates are either a precision variable (scenario 2) or an instrument (scenario 3), and a combination where one covariate is a weak instrument and the second is nearly a precision variable (scenario 4). Lastly, in all settings, one or both confounders also serve as effect modifiers.

\begin{table}
\begin{center}\begin{tabular}{|C{3.8em}||C{9em}|C{9em}|C{9em}|}
\hline
scenario & $\logit \pi_0(w)$   & $\mu_0(1,w)$ & $\mu_0(0,w)$       \\ \hline\hline
1  & $0.2 + 1.5 w_1 - w_2$ & $2 + 4 w_1 - 3w_2$                & $1 + 3w_1 + 3w_2$            \\ 
2  & $0.2 + 1.5 w_1$ & $2 + 4 w_1 - 3w_2$                & $1 + 3w_1 + 3w_2$            \\ 
3        & $0.2 + 1.5 w_1 - w_2$  & $2 - 3w_2$           & $1 + 3w_2$        \\ 
4        & $0.2 + 1.5 w_1 - 0.1 w_2 $           & $2 +0.1 w_1 - 3w_2$         & $1 + 0.1 w_1 + 3w_2$    \\ \hline
\end{tabular}\end{center}
\caption{Description of data-generating mechanism for various simulation scenarios considered.}
\label{datagen}
\end{table}

To validate numerically our asymptotic theory, we generated 5,000 random datasets of size 5,000 in each scenario outlined. We implemented a 95/5 sample-splitting scheme. Specifically, we split each dataset randomly into two parts, one with $n_\mathrm{eff}= 4,750$ observations and the other with $m_n=250$. The smaller sample was used to obtain an estimate $h_n$ of the optimal augmentation function $h_0$, whereas the larger sample was used to compute the 1:1 matching estimator based on propensity score augmentation by $h_n$. All comparisons between empirical and theoretical variances were based on sample size $n_{\mathrm{eff}}$.
To construct $h_n$, we fitted (correctly-specified) parametric regression models to estimate $\mu_0$ and $\pi_0$, and estimated $\bar{\mu}_0$ using the super learner \citep{van2007super} with the following candidate prediction algorithms: highly adaptive lasso estimator \citep{benkeser2016highly}, gradient-boosted regression trees \citep{chen2016xgboost}, and multivariate adaptive regression splines \citep{friedman1991multivariate}. We implemented these estimators using the R package \texttt{sl3} \citep{coyle2021sl3-rpkg}.

\begin{table}
\begin{center}\begin{tabular}{|C{3.55em}||C{4em}|C{4em}|C{4em}|C{4em}|C{4em}|C{4em}|C{4em}|}
\hline
scenario & efficiency bound  & theor. var. (aug) & emp. var. (aug) &  asympt. bias (aug) & theor. var. (no aug) & emp. var. (no aug) &  asympt. bias (no aug)\\ \hline\hline 
1 & 20.22 & 30.55 & 30.69 & 0.09 & 34.22 & 31.39 & 0.02 \\ 
2 &  19.17 & 31.82 & 33.73 & 0.04 & 42.62 & 40.98 & 0.12 \\ 
3 & 18.83 & 28.97 & 30.99 & 0.03 & 34.62 & 34.24 & 0.02 \\ 
4 &  17.78 & 29.94 & 29.75 & 0.04 & 39.99 & 38.12 & 0.08  \\ 
  \hline
\end{tabular}\end{center}
\caption{Theoretical versus empirical variance of the 1:1 matching with augmented propensity score estimator and 95/5 sample-splitting across simulation scenarios at sample size 5,000. The second, third, and sixth columns provide the variance suggested by theory for a nonparametric efficient estimator, the optimally augmented matching estimator based on sample-splitting, and the unaugmented matching estimator. The fourth and seventh columns display the empirical Monte Carlo variance of the augmented (sample-split) matching estimator and the unaugmented matching estimator implemented over 5,000 simulated datasets. The fifth and eighth columns show the Monte Carlo bias of the augmented and unaugmented estimators, respectively.} 
\label{table2}
\end{table}

Results of this simulation are summarized in Table \ref{table2}. Overall, the empirical variance of the augmented 1:1 matching estimator (fourth column) was close to the variance predicted by our theoretical results once adjusted for sample size depletion due to sample-splitting (third column), though in general the empirical variance was slightly above the theoretical variance. In these scenarios, none of the estimation procedures considered achieved nonparametric efficiency. In scenarios 2 and 4, in which $W_2$ was a strong precision variable, augmentation led to a significant improvement in efficiency, even despite the reduction in effective sample size resulting from sample-splitting. In scenarios 1 and 3, the augmented and unaugmented estimators had a similar variance.

\subsection{Simulation studies without sample-splitting}

A potential drawback of our proposed augmentation method is that it relies on sample-splitting. This reduces the effective sample size and introduces a tuning parameter: the proportion of the data to use for estimation of the optimal augmentation function. Given these observations, it is natural to consider foregoing sample-splitting altogether.

Through a simulation study, we examined the extent to which the behavior of the augmented propensity score 1:1 matching estimator differs when sample-splitting is used versus not. While our theoretical results for the augmented 1:$M$ matching estimator have only been shown to be valid when sample-splitting is used, as highlighted above, it is of practical interest to explore whether this is necessary or merely sufficient for the sake of our technical arguments. In this simulation study, we considered the same four data-generating mechanisms previously described. For each scenario, we compared the performance of the augmented matching estimator without sample-splitting to the unaugmented matching estimator.

The implementation of the augmented estimator was the same as described in the previous subsection except for the use of the full sample for estimation of both the optimal augmentation function and the average treatment effect.

\begin{figure}
\centering
\includegraphics[scale=0.4]{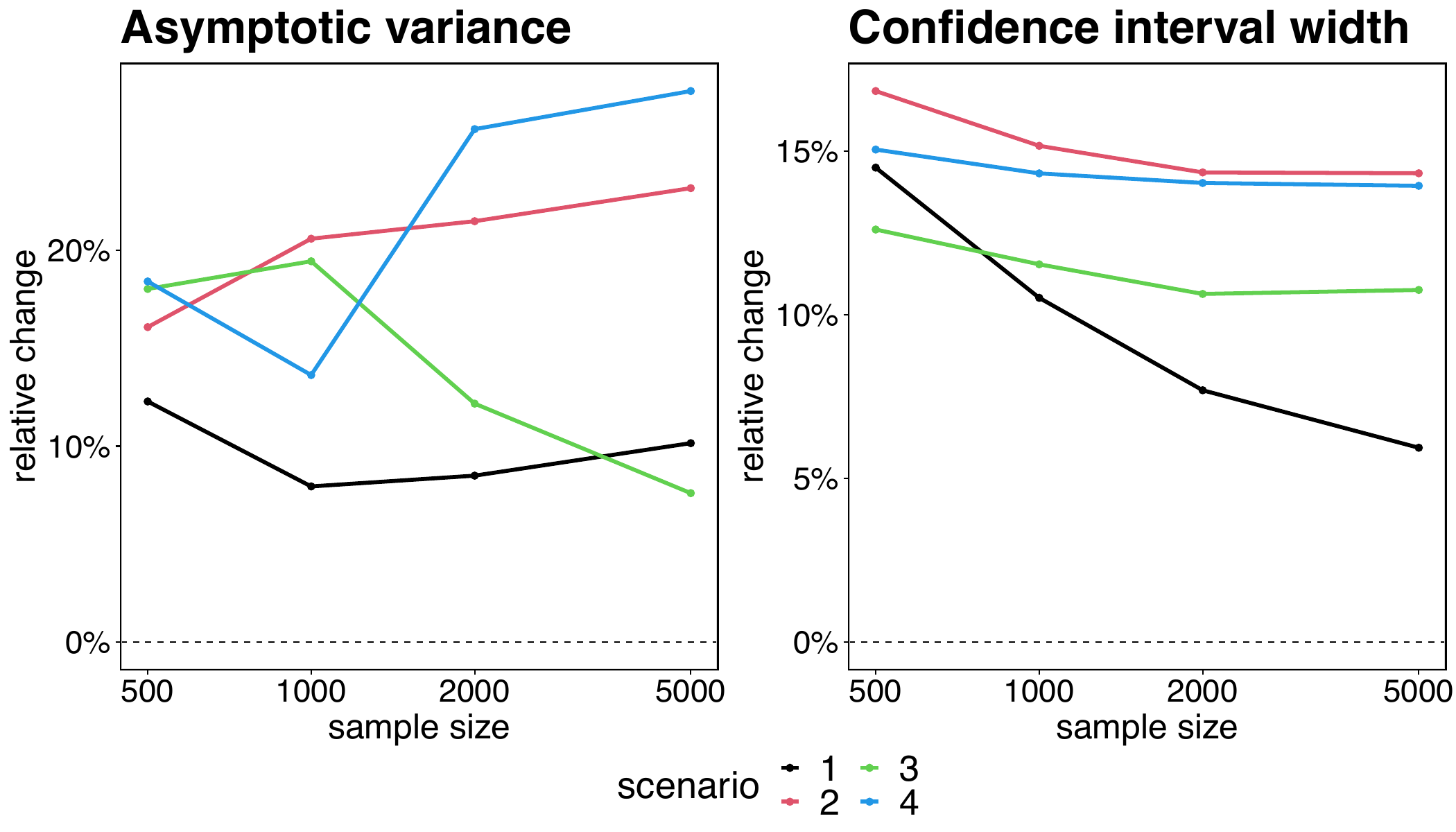}
\caption{Empirical Monte Carlo percent reduction in variance (sample size-scaled) and length of Wald-type confidence intervals (on each simulated data set) comparing the augmented 1:1 matching estimator (without sample-splitting) and its unaugmented counterpart across scenarios and sample sizes. A positive relative reduction indicates improvements resulting from augmentation.}
\label{fig2}
\end{figure}

\begin{figure}
\centering
\includegraphics[scale=0.4]{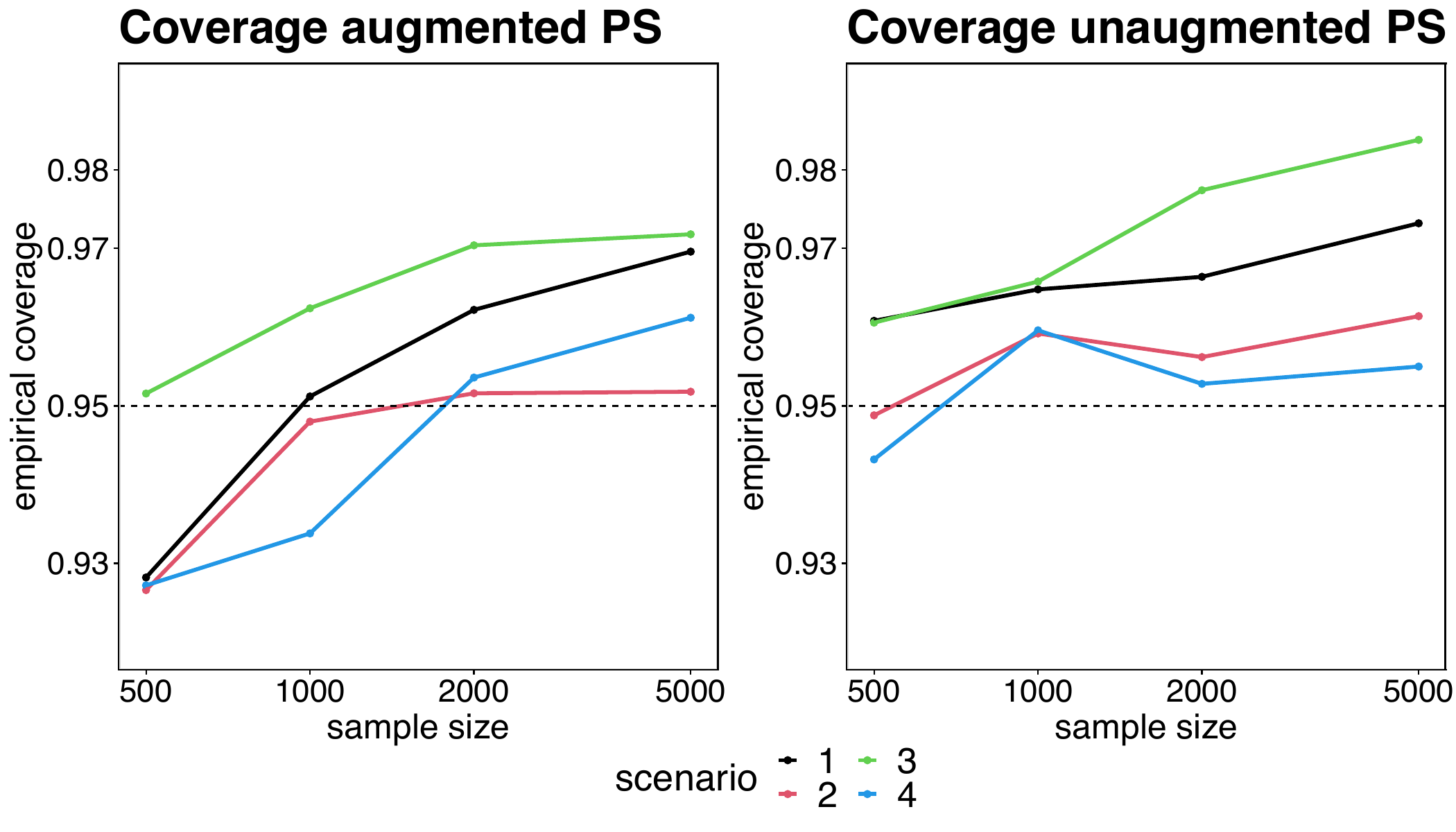}
\caption{Empirical coverage of Wald-type confidence intervals  based on the augmented 1:1 matching estimator (without sample-splitting) and its unaugmented version across scenarios and sample sizes.}
\label{fig3}
\end{figure}

Results are summarized in Table \ref{table3}. Once more, the Monte Carlo variance of the augmented matching estimator was relatively close to the variance predicted by our theoretical results, with some under-dispersion seen in certain settings. While bias is not explicitly reported on in this table, we found all estimators to have negligible bias, and the percent reduction in empirical variance reported in Table \ref{table3} was essentially equal to the percent reduction in empirical mean squared error in these simulations. In Figure \ref{fig2}, the empirical percent change in sample size-scaled variance and length of Wald-type confidence intervals resulting from augmentation is displayed not only for sample size $n=5000$, but also for smaller sample sizes $n\in\{500, 1000, 2500\}$. In most scenarios considered, Wald-type confidence intervals for the average treatment effect based on the augmented 1:1 matching estimator (without sample-splitting) and on estimates of the theoretical variance suffered from moderate undercoverage in smaller sample sizes considered, and slight overcoverage at sample size $n=5000$ (Figure \ref{fig3}). The unaugmented matching estimator with estimated variance had slight overcoverage with modest and larger sample sizes. The overcoverage seen in large samples mirrors the fact that the empirical variance of both the unaugmented and augmented estimators was somewhat smaller than predicted by theory.

\begin{table}
\begin{center}\begin{tabular}{|C{3.8em}||C{7em}|C{6em}|C{7em}|C{6em}|C{7em}|}
\hline
scenario  & theoretical var. (aug) & empirical var. (aug) & theoretical var. (no aug) & empirical var. (no aug) &  empirical change in var.\\ \hline\hline 
1 & 30.55 & 28.8 & 34.22 & 32.1 &  --10\%\\ 
2 & 31.82 &  30.5  & 42.62 & 39.7 &  --23\% \\ 
3 & 28.97 &  29.2 & 34.62 & 31.6  & --7\%  \\ 
4 & 29.94 &   27.4 & 39.99 &  38.2 &  --28\% \\ 
  \hline
\end{tabular}\end{center}
\caption{Theoretical versus empirical (sample size-scaled) variance of the 1:1 matching with augmented propensity score estimator and no sample-splitting across simulation scenarios at sample size 5,000. The second and fourth columns provide the variance suggested by theory for an optimally augmented matching estimator based on the known augmentation function and for its unaugmented counterpart, respectively. The third and fifth columns display the empirical Monte Carlo variance of the augmented matching estimator (without sample-splitting) and the unaugmented matching estimator implemented over 5,000 simulated datasets. The fifth column shows the percent change in variance between the augmented and unaugmented estimators.} 
\label{table3}
\end{table}

\section{Analyzing the effect of circumcision on risk of HIV infection} \label{sec:secc5}

We have employed our proposed propensity score augmentation matching method to estimate the causal effect of circumcision on risk of HIV-1 infection. To do so, we analyzed data from the Step Study, a randomized trial conducted by the HIV Vaccine Trials Network to assess the efficacy of a recombinant adenovirus type-5 (Ad5) HIV vaccine. The trial enrolled seronegative men between 18 and 46 years of age at 27 different sites in North America, the Caribbean, South America, and Australia. More details on the study design  can be found in \cite{duerr2012extended}.

Exposure, outcome of interest, and measured potential confounders were defined as follows. Circumcision, the exposure of interest, was self-reported initially but reassessed by clinicians through physical exam in the subsequent follow-up study. The outcome considered is HIV-1 infection occurring within the first two years since enrollment. Potential confounders recorded at baseline were participant age ($<30$ or $\geq 30$), region (North America + Australia, or Other), ethnicity (Black, Hispanic, Mestizo/Mestiza, White, or Other), attenuated Ad5 titer, evidence of sexually transmitted disease, and sexual risk factors, including number of male partners (0, $<4$, or $\geq 4$), unprotected insertive anal sex (yes or no), and unprotected receptive anal sex (yes or no) in the last three months before enrollment. 

To estimate the causal effect of circumcision on risk of HIV-1 infection, we conducted 1:1 propensity score matching, and first estimated the propensity score through logistic regression with all measured potential confounders as main terms. We compared the results obtained to matching with the augmented propensity score, which included  the same main terms and the estimated optimal augmentation covariate. We estimated $\mu_0$ using the Super Learner algorithm \citep{van2007super}, and $\bar{\mu}_0$ using the nested procedure and the univariable the Super Learner algorithm as well. In our analysis, sample-splitting was not implemented. Point estimates were obtained using the \texttt{Matching} package in \texttt{R} \citep{sekhon2008multivariate}. The variance of resulting matching estimators were estimated by subtracting from the estimator of $\sigma^2_M$ provided by the \texttt{Matching} package an empirical plug-in estimator of their corresponding gain ---$\mathrm{gain}(h_0)$ or $\mathrm{gain}(\varnothing)$---based on $\mu_n$, $\bar{\mu}_n$ and $\pi_n$.

A total of 1,488 study participants were followed for at least two years and provided complete covariate information. Of these participants, 828 were circumcised, and the other 660 were not. In all, 109 participants developed an HIV-1 infection in the first two years of follow-up, 63 of which were circumcised and the other 46 were not. The average treatment effect was estimated to be -0.01 using the standard 1:1 propensity score-based matching estimator, with estimated standard error 0.028, and -0.04 using the augmented 1:1 propensity score-based matching estimator, with estimated standard error 0.025. 

\section{Discussion} \label{sec:secc6}

Our findings suggest that 1:$M$ matching using a propensity score estimated through the augmented model $\mathcal{M}(h_0)$ may yield substantial efficiency gains. Additionally, we found that the data generating mechanism plays an important role in the variance and efficiency gain of the augmented matching estimator.\cite{brookhart2006variable} found that including precision variables in propensity score regression and sub-classification methods increased their precision. This is also true for matching estimators, and the resulting gain in efficiency from including precision variables is even larger when the optimal augmentation covariate is included. Thus, we recommend including any precision variable available into the augmentation covariate. Interestingly, a slight modification of Theorem \ref{theo1} shows that there is no need to include any precision variable in the propensity score model, since either doing so or not leads to the same asymptotic variance after optimal augmentation. Thus, if researchers have at their disposal several precision variables, they can extract all the relevant information they contain by including them all in the augmentation covariate, without any need to increase the number of parameters in the propensity score model beyond the single augmentation covariate. This result can be practically advantageous in settings in which the number of precision variables is large and sample size is limited. \cite{brookhart2006variable} also found that including instruments in propensity score regression and sub-classification methods can be detrimental to their performance. This result aligns with our findings in Section \ref{subsec:releffgd}, where we showed that the variance of the matching estimator increased as the strength of the instrument increased. The use of the optimal augmentation covariate did attenuate this loss in precision. Nevertheless, we recommend excluding instruments in both the propensity score model and the augmentation covariate.  

While our proposed approach provides a means of performing propensity score matching with greater statistical efficiency, it does suffer from three potential limitations. First, our approach builds upon correct specification of the logistic regression model for the propensity score. Nevertheless, because any propensity score can be approximated by a logistic regression model based on covariate transformations, our methodology could be implemented with $W$ replaced by the evaluation of a rich collection of basis functions on $W$. Alternatively, as in \cite{abadie2016}, it appears possible to extend our theoretical results to arbitrary parametric regression models for the propensity score. Second, the need for sample-splitting in order to estimate $h_0 $ in our procedure results in a reduced effective sample size. Nevertheless, our simulation results suggest that using the complete sample to build an estimator of $h_0$ and then to estimate $\psi_0$  does not invalidate the asymptotic distributional approximation provided in our theory. Whether sample-splitting is fundamentally necessary or only an artifact of our technical proofs remains to be formally studied. Third, similarly as in the work of \cite{abadie2016}, our theoretical results only pertain to the augmented 1:$M$ matching estimator incorporating a discretization regularization, but not to its unregularized counterpart, even though our simulation studies suggest that this is perhaps not a significant issue for practical purposes. Here again, it is of interest to determine whether our results can also be established for the unregularized augmented 1:$M$ matching estimator.

In this work we have focused on the average treatment effect, however the same ideas used here could be explored in the context of the average treatment effect among the treated.  \cite{abadie2016} showed that when estimating the latter causal contrast, estimation of the propensity score may or may not result in efficiency improvements. Identifying the optimal augmentation covariate and elucidating its impact on efficiency remains an open question.

\bibliography{supplement-ref.bib}

\appendix

\section{Appendix}

In this section, we briefly outline a the remaining set of conditions for Theorem \ref{theo:theo1} and Proposition \ref{propp:1} to hold. The derivations and proofs can be found in the Supplement of this work. 

Let $\vartheta_{0}:= (\theta_0,0)$, and for any $\vartheta\in\mathbb{R}^{p+2}$ and $\hd: \mathcal{W}\rightarrow\mathbb{R}$, let $r_\hd(w):= (w,\hd(w))$, $\pi_{\hd,\vartheta}(w):= {\rm expit}[\vartheta^{\top} r_\hd(w)]$, and let $P_{\hd,\vartheta}$ denote the distribution for which
\begin{align}
    \frac{dP_{\hd,\vartheta}}{dP_0}(w,a,y)&= \frac{\pi_{\hd,\vartheta}(w)^a[1-\pi_{\hd,\vartheta}(w)]^{1-a}}{\pi_0(w)^a[1-\pi_0(w)]^{1-a}}. \label{eq:Pthetah}
\end{align}
We write $E_{\hd,\vartheta}$ to denote the expectation operator under $P_{\hd,\vartheta}$. 

Fix $g\in \mathbb{R}^{p+2}$ and let $\bar{P}_{\hd,\vartheta}$ denote the marginal distribution of $(W,A)$ under sampling from $P_{\hd,\vartheta}$. Let $\vartheta_n := \vartheta_{0} + g/\sqrt{n}$. In the following definitions, we let $\vartheta$ and $\vartheta'$ be generic elements of $\mathbb{R}^{p+2}$, $\hd$ be a $\mathcal{W}\rightarrow\mathbb{R}$ function, $z:=(w,a)$ denote a realization of a covariate-treatment tuple, $\mathcal{Z}$ its support, and $d_n:= \{z_i\}_{i=1}^n$ denote a dataset consisting of $n$ such tuples. We let $\ell_\hd(\vartheta;z)$ denote the log-likelihood for $\vartheta\in\mathbb{R}^{p+2}$ that would arise under sampling from $\bar{P}_{\hd,\vartheta}$, $U_\hd(\vartheta;z):= \frac{\partial}{\partial \vartheta} \ell_\hd(\vartheta;z)$ denote the score function for $\vartheta$, and $\mathscr{I}_{\hd,\vartheta}:=E_{\hd,\vartheta}[U_\hd(\vartheta;Z)U_\hd(\vartheta;Z)^{\top}]$ be the information matrix for $\vartheta$. Additionally, let $\ell_{\hd}(\vartheta; d_n):=\sum_{i=1}^n \ell_\hd(\vartheta;z_i)$, $\Lambda_{\hd}(\vartheta,\vartheta';d_n):=\ell_{\hd}(\vartheta; d_n) - \ell_{\hd}(\vartheta';d_n)$, let $\bar{p}_{\hd,\vartheta}:=d\bar{P}_{\hd,\vartheta}/d\bar{P}_0$, where $\bar{P}_0$ is the marginal distribution of $(W,A)$ under sampling from $P_0$, and let $P_{W,0}$ denote the marginal distribution of $W$ under $P_0$. Finally, some of our proofs use deterministic sequences in a family of $\mathcal{W} \rightarrow \mathbb{R}$ functions which we denote by $\mathcal{H}$.

\renewcommand{\thecondition}{S\arabic{condition}}
\setcounter{condition}{0}

\begin{condition}
\label{cond:bounds:appendix}
There exists an $\epsilon >0$ such that, for all $a\in \{0,1\}$ and $\tilde{\vartheta} \in \{\vartheta \in \mathbb{R}^{p+2}: \|\vartheta - \vartheta_0\| \leq \epsilon\}$, both of the following hold:  $E_{\hd, \tilde{\vartheta}}[Y\mid A = a, \pi_{\hd, \tilde{\vartheta}}(W) = p]$ is L-Lipschitz continuous in $p$ and uniformly bounded over $\tilde{\vartheta}$, $\hd$, and $p$; and there exists a $\delta >0$ such that $E_{\hd, \tilde{\vartheta}}[|Y|^{2+\delta}\mid A=a,  \pi_{\hd, \tilde{\vartheta}}(W)=p]$ is uniformly bounded over $\tilde{\vartheta}$, $\hd$, and $p$. 
\end{condition}

\begin{condition}
\label{cond:mle:appendix}
Consistency: $\vartheta_0$ is identifiable, in the sense that for any given observation $\{a_i,w_i\}$, we have that  $P_{\hd,\vartheta}(a_i,w_i) \neq P_{\hd, \vartheta_0}(a_i,w_i)$ for all $\vartheta \neq \vartheta_0$.
\end{condition}

\begin{condition}
\label{cond:lambdan:appendix}
For every sequence of vectors $(\vartheta_n)_{n=1}^{\infty}$ such that $\vartheta_n= \vartheta_0 + g/\sqrt{n}$ for some $g \in \mathbb{R}^{p+1}$ and every sequence of functions $(h_{0n})_{n=1}^{\infty}$ such that $\|h_{0n} - h_0\|_{L_2(P)} \rightarrow 0$, we have that, for every $\epsilon > 0$,
\begin{equation}
P_{h_{0n},\vartheta_n}\left(\left|\frac{\lambda_{n}}{\lambda^*} - 1\right| > \epsilon\right) \xrightarrow{n \to \infty} 0.
\end{equation}
Above 
$\lambda_{n}$ and $\lambda^*$ are as defined in Equations \ref{eqdefl:lembrasymn} and \ref{eqdeflstr:lembrasymn} found in the proof of Lemma \ref{lem:brasym} in the Supplementary material.
\end{condition}

\begin{condition}
\label{cond:ulan:last:appendix}
Given $\hd$, $d_n$, $\vartheta,$ and $\vartheta'$, let $\Lambda_{\hd}(\vartheta,\vartheta';d_n)=\ell_{\hd}(\vartheta; d_n) - \ell_{\hd, d_n}(\vartheta')$ be the log-likelihood ratio. Then, for any $\vartheta_n,\vartheta_n'$, sequence of vectors in $\mathbb{R}^{p+2}$ contiguous to $\vartheta_0$ such that $\delta_n = \sqrt{n}(\vartheta_n - \vartheta_0)$ and $\delta_n'  = \sqrt{n}(\vartheta'-\vartheta_0)$ are bounded sequences in $\mathbb{R}^{p+2}$, there exists a sequence of processes $\{\Delta_{h_{0n}}(\vartheta): \vartheta \in \mathbb{R}^{p+2}\}_{n=1}^{\infty}$ such that
\begin{align*}
\Lambda_{h_{0n},n}(\vartheta_n'|\vartheta_n) &= (\delta_n'-\delta_n)^{\rm{T}}\Delta_{h_{0n}}(\vartheta_0 + \delta_n/\sqrt{n}) - \frac{1}{2}(\delta_n' - \delta_n)^{\rm{T}}\mathscr{I}_{h_0, \vartheta_0}(\delta_n' - \delta_n) + o_{P_{h_0, \vartheta_0}}(1)  \\
&= (\delta_n'-\delta_n)^{\rm{T}}\Delta_{h_{0n}}(\vartheta_0 + \delta_n'/\sqrt{n}) + \frac{1}{2}(\delta_n' - \delta_n)^{\rm{T}}\mathscr{I}_{h_0, \vartheta_0}(\delta_n' - \delta_n) + o_{P_{h_0, \vartheta_0}}(1).
\end{align*}
Moreover, the sequence $\Delta_{h_{0n},n}(\vartheta_n)$ is asymptotically normally distributed with zero mean and variance $\mathscr{I}_{ h_0,\vartheta_0}$ under $P_{h_{0n},\vartheta_n}$ as $n \rightarrow \infty$. 
\end{condition}

\section*{ \centering Supplementary Material}

\section{Preliminaries} \label{sec:prelim}

\subsection{Notation}

Fix a sequence of natural numbers $\{m_n\}_{n=1}^\infty$ for which $\lim_{n\rightarrow \infty } m_n/n = 0$. In our upcoming arguments, we will generally condition on the iid sequence $S := \{(W_{-i},A_{-i},Y_{-i})\}_{i=1}^{\infty}$ from $P_0$. The first $m_n$ observations of this sequence are used to obtain the estimate $h_n$ of $h_0$. Some of our proofs will also employ deterministic sequences in a family of $\mathcal{W} \rightarrow \mathbb{R}$ functions denoted by $\mathcal{H}$. We will denote these deterministic sequences in $\mathcal{H}$ via $(h_{0n})_{n=1}^\infty$, and arbitrary functions in $\mathcal{H}$ by $\hd$. Our results will depend on $(h_{0n})_{n=1}^\infty$ satisfying certain properties, and, as we will show later on, when the random sequence $(h_n)_{n=1}^\infty$ of estimators of $h_0$ satisfies these properties with probability one, the conclusions of these arguments will also hold with probability one. 

Let $P_0$ denote the distribution of each of the $n+m_n$ iid copies of $(W,A,Y)$ that were observed. Let $\mathcal{W}, \mathcal{Y}$ be the support of $W$ and $Y$, respectively. We follow the convention that all vectors are column vectors and we let $v_{k}$ denote the k-th entry of a vector $v$. By assumption there exists some $\theta_0\in \mathbb{R}^{p+1}$ such that
\begin{align*}
    \logit P_0(A = 1 \mid W) = \theta_0^{\rm{T}} W.
\end{align*}

Hereafter, we let $\vartheta_{0}:= (\theta_0,0)$, $\pi_0(w):= P_0(A=1|W=w)$, and $\rho(x) := \frac{d}{dx}{\rm expit}(x) = \frac{e^x}{(1+e^x)^2}$.

For any $\vartheta\in\mathbb{R}^{p+2}$ and $\hd: \mathcal{W}\rightarrow\mathbb{R}$, let $r_\hd(w):= (w,\hd(w))$, $\pi_{\hd,\vartheta}(w):= {\rm expit}[\vartheta^{\rm{T}} r_\hd(w)]$, and let $P_{\hd,\vartheta}$ denote the distribution for which
\begin{align}
    \frac{dP_{\hd,\vartheta}}{dP_0}(w,a,y)&= \frac{\pi_{\hd,\vartheta}(w)^a[1-\pi_{\hd,\vartheta}(w)]^{1-a}}{\pi_0(w)^a[1-\pi_0(w)]^{1-a}}. \label{eq:Pthetah}
\end{align}
We write $E_{\hd,\vartheta}$ to denote the expectation operator under $P_{\hd,\vartheta}$. Unless otherwise stated, $E[\cdot]$ will be used to denote $E_{ h_0,\vartheta_0}[\cdot]$. Note that $P_0=P_{\hd, \vartheta_0}$, regardless of the chosen function $\hd$. Moreover, under sampling $(W,A,Y)$ from $P_{\hd,\vartheta}$, $W$ has the same marginal distribution as it would under sampling from $P_0$, and $Y\mid A,W$ has the same conditional distribution as it would under sampling from $P_0$. The two distributions only differ in the conditional distribution of $A\mid W$. In particular, $P_{\hd,\vartheta}(A = 1 \mid W=w)$ is equal to $\pi_{\hd,\vartheta}(w)$, rather than to $\pi_0(w)$.

Fix $g\in \mathbb{R}^{p+2}$ and let $\bar{P}_{\hd,\vartheta}$ denote the marginal distribution of $(W,A)$ under sampling from $P_{\hd,\vartheta}$. Let $\vartheta_n := \vartheta_{0} + g/\sqrt{n}$. In the following definitions, we let $\vartheta$ and $\vartheta'$ be generic elements of $\mathbb{R}^{p+2}$, $\hd$ be a $\mathcal{W}\rightarrow\mathbb{R}$ function, $z:=(w,a)$ denote a realization of a covariate-treatment tuple, $\mathcal{Z}$ its support, and $d_n:= \{z_i\}_{i=1}^n$ denote a dataset consisting of $n$ such tuples. We let $\ell_\hd(\vartheta;z)$ denote the log-likelihood for $\vartheta\in\mathbb{R}^{p+2}$ that would arise under sampling from $\bar{P}_{\hd,\vartheta}$, $U_\hd(\vartheta;z):= \frac{\partial}{\partial \vartheta} \ell_\hd(\vartheta;z)$ denote the score function for $\vartheta$, and $\mathscr{I}_{\hd,\vartheta}:=E_{\hd,\vartheta}[U_\hd(\vartheta;Z)U_\hd(\vartheta;Z)^{\rm{T}}]$ be the information matrix for $\vartheta$. In our arguments, we will use that the aforementioned score function takes the following form:
\begin{align}
\label{defu}
	U_\hd(\vartheta;z)
	&= a  \nabla_\vartheta \log \pi_{\hd,\vartheta}(w) + (1-a)\nabla_\vartheta\log\{1-\pi_{\hd,\vartheta}(w)\}. 
\end{align}
Additionally, define the following function to denote the Jacobian of the score function:
\begin{align}
\label{defK}
K_{\hd}(\vartheta;z) = \frac{\partial}{\partial \vartheta} U_{\hd}(\vartheta; z). 
\end{align}
In a slight abuse of notation, we also define $\ell_{\hd}(\vartheta; d_n):=\sum_{i=1}^n \ell_\hd(\vartheta;z_i)$, $U_{\hd}(\vartheta; d_n):=\sum_{i=1}^n U_\hd(\vartheta;z_i)$, and $K_{\hd}(\vartheta;d_n) := \sum_{i=1}^n K_{\hd}(\vartheta;z_i)$. Additionally, let $\Lambda_{\hd}(\vartheta,\vartheta';d_n):=\ell_{\hd}(\vartheta; d_n) - \ell_{\hd}(\vartheta';d_n)$, and let $\bar{p}_{\hd,\vartheta}:=d\bar{P}_{\hd,\vartheta}/d\bar{P}_0$, where $\bar{P}_0$ is the marginal distribution of $(W,A)$ under sampling from $P_0$. By Equation \eqref{eq:Pthetah}, $\bar{p}_{\hd,\vartheta}(w,a)=\pi_{\hd,\vartheta}(w)^a[1-\pi_{\hd,\vartheta}(w)]^{1-a}/\{\pi_0(w)^a[1-\pi_0(w)]^{1-a}\}$. Finally, $P_{W,0}$ denotes the marginal distribution of $W$ under $P_0$. The following section outlines the sufficient conditions for the proofs of the results in our work to hold. 

\subsection{Conditions}

\begin{condition}
\label{cond:pos:first}
Positivity: for some constants $u$ and $l$ such that $0 < l < u <1$, we have that $l < \pi_0(W) < u$ almost surely.
\end{condition}

\begin{condition} 
\label{cond:qmd}
 
Boundedness: $\mathcal{W}$ is a uniformly bounded subset of $\mathbb{R}^{p+1}$. 
\end{condition}

\begin{condition}
\label{cond:iden}
Identifiability: (i) consistency: for each observation $Y = AY(A) + (1-A)Y(1-A)$; (ii) no unmeasured confounding: for $a = 0,1$ $Y(a) \indep A | W $; (iii) no interference:  the random variables $Y(1), Y(0)$ and $A$ are such that for every $i\neq j$, we have that $Y_i(0) \indep A_j$, and $Y_i(1) \indep A_j$.
\end{condition}

\begin{condition} 
\label{cond:qmd2}
Conditions on $\mathcal{H}$:
(i) $\mathcal{H}$ is a collection of bounded $\mathcal{W}\rightarrow\mathbb{R}$ functions, so that  \allowbreak
$\sup_{\hd \in\mathcal{H}}\sup_{w\in\mathcal{W}} |\hd(w)|<\infty$. 
(ii) The class of functions $\mathcal{H}$ is totally bounded in $L^2(P_{W,0})$. 
(iii) The class of functions $\mathcal{H}$ is such that for every $\hd \in \mathcal{H}$, $E(r_\hd(W)r_\hd(W)^{\rm{T}})$ is positive definite.
\end{condition}

\begin{condition}
\label{cond:ps}
The evaluation of the propensity score $\pi$ at $W$ is continuously distributed, with continuous density.
\end{condition}

\renewcommand{\thecondition}{S\arabic{condition}}
\setcounter{condition}{0}

\begin{condition}
\label{cond:bounds}
There exists an $\epsilon >0$ such that, for all $a\in \{0,1\}$ and $\tilde{\vartheta} \in \{\vartheta \in \mathbb{R}^{p+2}: \|\vartheta - \vartheta_0\| \leq \epsilon\}$, both of the following hold:  $E_{\hd, \tilde{\vartheta}}[Y\mid A = a, \pi_{\hd, \tilde{\vartheta}}(W) = p]$ is L-Lipschitz continuous in $p$ and uniformly bounded over $\tilde{\vartheta}$, $\hd$, and $p$; and there exists a $\delta >0$ such that $E_{\hd, \tilde{\vartheta}}[|Y|^{2+\delta}\mid A=a,  \pi_{\hd, \tilde{\vartheta}}(W)=p]$ is uniformly bounded over $\tilde{\vartheta}$, $\hd$, and $p$. 
\end{condition}

\begin{condition}
\label{cond:mle}
Consistency: $\vartheta_0$ is identifiable, in the sense that for any given observation $\{a_i,w_i\}$, we have that  $P_{\hd,\vartheta}(a_i,w_i) \neq P_{\hd, \vartheta_0}(a_i,w_i)$ for all $\vartheta \neq \vartheta_0$.
\end{condition}

\begin{condition}
\label{cond:lambdan}
For every sequence of vectors $(\vartheta_n)_{n=1}^{\infty}$ such that $\vartheta_n= \vartheta_0 + g/\sqrt{n}$ for some $g \in \mathbb{R}^{p+1}$ and every sequence of functions $(h_{0n})_{n=1}^{\infty}$ such that $\|h_{0n} - h_0\|_{L_2(P)} \rightarrow 0$, we have that, for every $\epsilon > 0$,
\begin{equation}
P_{h_{0n},\vartheta_n}\left(\left|\frac{\lambda_{n}}{\lambda^*} - 1\right| > \epsilon\right) \overset{n \rightarrow \infty }{\longrightarrow} 0.
\end{equation}
Above 
$\lambda_{n}$ and $\lambda^*$ are as defined in Equations \ref{eqdefl:lembrasymn} and \ref{eqdeflstr:lembrasymn}, which can be found in the proof of Lemma \ref{lem:brasym}.
\end{condition}

\begin{condition}
\label{cond:ulan:last}
Given $\hd$, $d_n$, $\vartheta,$ and $\vartheta'$, let $\Lambda_{\hd}(\vartheta,\vartheta';d_n)=\ell_{\hd}(\vartheta; d_n) - \ell_{\hd, d_n}(\vartheta')$ be the log-likelihood ratio. Then, for any $\vartheta_n,\vartheta_n'$, sequence of vectors in $\mathbb{R}^{p+2}$ contiguous to $\vartheta_0$ such that $\delta_n = \sqrt{n}(\vartheta_n - \vartheta_0)$ and $\delta_n'  = \sqrt{n}(\vartheta'-\vartheta_0)$ are bounded sequences in $\mathbb{R}^{p+2}$, there exists a sequence of processes $\{\Delta_{h_{0n}}(\vartheta): \vartheta \in \mathbb{R}^{p+2}\}_{n=1}^{\infty}$ such that
\begin{align*}
\Lambda_{h_{0n},n}(\vartheta_n'|\vartheta_n) &= (\delta_n'-\delta_n)^{\rm{T}}\Delta_{h_{0n}}(\vartheta_0 + \delta_n/\sqrt{n}) - \frac{1}{2}(\delta_n' - \delta_n)^{\rm{T}}\mathscr{I}_{h_0, \vartheta_0}(\delta_n' - \delta_n) + o_{P_{h_0, \vartheta_0}}(1)  \\
&= (\delta_n'-\delta_n)^{\rm{T}}\Delta_{h_{0n}}(\vartheta_0 + \delta_n'/\sqrt{n}) + \frac{1}{2}(\delta_n' - \delta_n)^{\rm{T}}\mathscr{I}_{h_0, \vartheta_0}(\delta_n' - \delta_n) + o_{P_{h_0, \vartheta_0}}(1).
\end{align*}
Moreover, the sequence $\Delta_{h_{0n},n}(\vartheta_n)$ is asymptotically normally distributed with zero mean and variance $\mathscr{I}_{ h_0,\vartheta_0}$ under $P_{h_{0n},\vartheta_n}$ as $n \rightarrow \infty$. 
\end{condition}

We require Conditions \ref{cond:pos:first}-\ref{cond:ps} and Conditions \ref{cond:bounds}-\ref{cond:ulan:last} to derive the asymptotic results of the proposed matching estimator. Conditions \ref{cond:pos:first} and \ref{cond:iden} are standard conditions for identification of the average treatment effect. Regarding estimation of the average treatment effect, in the first part of our proof, we show a uniform version of Theorem 7.2 in \cite{van2000asymptotic} over probability models indexed by $\vartheta$ and $\hd \in \mathcal{H}$, where $\mathcal{H}$ is a function class. Specifically, the uniformity is over the function class $\mathcal{H}$. In order to show uniform quadratic mean differentiability we require Conditions \ref{cond:qmd} and \ref{cond:qmd2}, which impose restrictions on $\mathcal{W}$ and $\mathcal{H}$, respectively. In Condition \ref{cond:mle} we require that the probability model is identifiable for the maximum likelihood estimator to consistently estimate the propensity score parameters. Conditions \ref{cond:ps}, \ref{cond:bounds}, and \ref{cond:lambdan}, are used to show joint asymptotic normality of a martingale whose definition is based on the score function at $\vartheta_n$, and the matching estimator evaluated at a propensity score that takes the values $\vartheta_n$ and a sequence of deterministic functions $(h_{0n})_{n=1}^{\infty}$ that are converging to $h_{0}$. Specifically, Condition \ref{cond:ps} pertains to the distribution of the propensity score evaluated at $W$ and is required to leverage the results shown in \cite{abadie2016}. Next, Condition \ref{cond:bounds}, which is similar to Assumption 4 in \cite{abadie2016}, imposes restrictions on the conditional mean of $Y$ given the propensity score and treatment value $a$.  Condition \ref{cond:lambdan} helps us control the asymptotic variance of the aforementioned martingale. Based on the asymptotic normality of the martingale, we then show joint asymptotic normality of the matching estimator, the maximum likelihood estimator, and the log likelihood ratio between $\vartheta_0$ and its shift $\vartheta_n$. Finally, in order to obtain the marginal distribution of the matching estimator, we use the discretization technique from \cite{andreou2012alternative}, which requires Condition \ref{cond:ulan:last}. 

The following two sections ---Sections \ref{supsec:sec2} and \ref{supsec:sec3}--- state and show the proofs of Theorem \ref{theo1} and Corollary \ref{cor1}, respectively. Section \ref{subsec:sec4} shows the lemmas we require to show Proposition \ref{prop1}. Next, in Section \ref{supsec:sec5}, we show the proof of Proposition \ref{prop1}. Section \ref{supsec:sec6} shows the proof of the Corollary 2 that allows the results of Proposition 1 to hold even when the sequence of functions to be random. 


\section{Proof of Theorem~\ref{theo:theo1}}\label{supsec:sec2} 
The following theorem and its corollary pertain to the efficiency gain term in the asymptotic variance of the matching estimator. Before stating the theorem, we introduce additional notation and definitions as well as remind the reader of previous definitions. 

Let $q(w):=(q_1(w), \ldots, q_{m}(w))$ denote a $m$-dimensional vector that consists of the evaluations of $m$ real-valued satisfying that $\text{var}(q_j(W))$ for each $j = 1,...,m$.

We let $\pi_{q, \theta, \gamma}:= \text{expit}[\theta^{\rm{T}}w + \gamma^{\rm{T}}q(w)]$, where $\theta \in \mathbb{R}^{p+1}$ and $\gamma \in \mathbb{R}^{m}$. In the special case where $m=1$, $q_{1}=\hd$ for some $\hd$ in $\mathcal{H}$, it holds that $\pi_{q, \theta, \gamma}=\pi_{\hd,\vartheta}$. Recall that $\rho(x) := \frac{d}{dx}{\rm expit}(x)=\frac{e^x}{(1+e^x)^2}$, let $$\rho_{q, \theta, \gamma}(w):= \frac{e^{\theta^{\rm{T}}w + \gamma^{\rm{T}}q(w)}}{(1+e^{\theta^{\rm{T}}w + \gamma^{\rm{T}}q(w)})^2}.$$ Additionally, we will let $\vartheta_{0,m}$ be a vector of dimension $m+p+1$, where its first $p+1$ entries of have the same values as those of $\theta_0$, and the remaining $m$ are 0. In a slight abuse of notation, we will let $\vartheta_{0,m}$ be denoted by $\vartheta_{0}$, where the number of additional 0's is based on the length of $q$--- which is usually included in expressions where $\vartheta_{0,m}$ appears.

Let $P_{q,\theta,\gamma}$ denote the distribution for which 
\begin{equation*}
    \frac{dP_{q,\theta,\gamma}}{dP}(w,a,y) = \frac{\pi_{q, \theta, \gamma}(w)^a[1-\pi_{q, \theta, \gamma}(w)]^{1-a}}{\pi(w)^a[1-\pi(w)]^{1-a}}.
\end{equation*}
Let $\bar{\mu}_{q,\theta,\gamma}(a,p) := E_{q, \theta, \gamma}[Y|A = a, \pi_{q, \theta, \gamma}(W) = p]$ denote the conditional expectation under $P_{q,\theta,\gamma}$ of $Y$ given $a$ and the augmented propensity score evaluated at $p$, and let $\mu_{q, \theta, \gamma}(a,w) = E_{q, \theta, \gamma}[Y|A = a, W = w]$ denote the conditional mean of $Y$ given treatment $a$ and covariates $w$, also under $P_{q,\theta,\gamma}$. 

Given augmentation covariates $q$, we will denote the corresponding model of the augmented propensity score by 
\begin{align*}
	\mathcal{M}(q) :=\{\text{logit}[\pi_{q, \theta, \gamma}(w)] &= \theta^{\rm{T}}w + \gamma^{\rm{T}}q(w): q_i(w)  \in L^2(P_{W,0}) \; \forall i\in \{1,\dots,m\}, \; \mathscr{I}_{q,\vartheta_0} \text{ is invertible}\},
\end{align*} 
where $$\mathscr{I}_{q,\vartheta_0}:=E\left[\pi_0(W)\{1-\pi_0(W)\}\begin{bmatrix}WW^{\rm{T}} & Wq(W)^{\rm{T}}\\ q(W)W^{\rm{T}} & q(W)q(W)^{\rm{T}} \end{bmatrix}\right],$$ denotes the information matrix for $\vartheta$ at $\vartheta=\vartheta_0$. Let $(\mathscr{I}_{q,\vartheta_0})^{-1}$ be equal to the inverse of the information matrix.

By \cite{abadie2016}, the efficiency gain in the asymptotic variance of the 1:$M$ matching estimator  using an estimated propensity score based on model $\mathcal{M}(q)$ and a discretized regularization is given by  $c_{q,\vartheta_0}^{\rm{T}}(\mathscr{I}_{q,\vartheta_0})^{-1}c_{q,\vartheta_0}$, where $c_{q,\vartheta_0}=(c_{\theta_0}^{\rm{T}},\bar{c}_{q,\vartheta_0}^{\rm{T}})^{\rm{T}}$ is defined by the vectors
\begin{align*}
    c_{\theta_0}\ &:=\ E\left[\pi_0(W)\, \mathrm{cov}\left\{W,\mu_0(0,W)\,\middle|\,A=0,\pi_0(W)\right\}\right]\\
    &\hspace{.5in}+E\left[\{1-\pi_0(W)\}\, \mathrm{cov}\left\{W,\mu_0(1,W)\,\middle|\,A=1,\pi_0(W)\right\}\right],
\end{align*}
and
\begin{align*}
   \bar{c}_{q,\vartheta_0}\ &:=\ E\left[\pi_0(W)\, \mathrm{cov}\left\{q(W),\mu_0(0,W)\,\middle|\,A=0,\pi_0(W)\right\}\right]\\
    &\hspace{.5in}+E\left[\{1-\pi_0(W)\}\, \mathrm{cov}\left\{q(W),\mu_0(1,W)\,\middle|\,A=1,\pi_0(W)\right\}\right].
\end{align*} 
For convenience we let ${\rm{gain}}(q)$ denote such efficiency gain. 

Before showing the proof of Theorem \ref{theo:theo1}, recall that $h_0$ is defined as 
\begin{equation}
\label{defh}
h_0(w) = \frac{1}{\pi_{h_0, \vartheta_0}(w)}\left[\mu_{h_0, \vartheta_0}(1, w)  - \bar{\mu}(1,\pi_{h_0, \vartheta_0}(w) )\right] + \frac{1}{1-\pi_{h_0, \vartheta_0}(w)}\left[\mu_{h_0, \vartheta_0}(0, w)  - \bar{\mu}(0,\pi_{h_0, \vartheta_0}(w))\right].
\end{equation}
Conditions \ref{cond:pos:first}, \ref{cond:qmd}, and \ref{cond:bounds} imply that $h_0 \in L^2(P_{W,0})$.
Finally, denote by $\mathcal{Q}_0$ the set obtaining by collecting, for each $m\in\{1,2,\ldots\}$, each augmentation function $q:\mathcal{W}\rightarrow \mathbb{R}^{m}$ such that $E\{q_j(W)^2\}<\infty$ for each $j=1,2,\ldots,m$ and $\mathscr{I}_{q,\vartheta_0}$ is invertible. 
\begin{theorem}
\label{theo:theo1}
Suppose Conditions \ref{cond:pos:first}-\ref{cond:iden}, and \ref{cond:bounds} hold. The following statements hold true:
\begin{enumerate}[(a)]
\item $\mathrm{gain}(h_0)=\sup_{q\in\mathcal{Q}_0}\mathrm{gain}(q)$;
\item $\mathrm{gain}(h_0)=E\left[\pi_0(W)\{1-\pi_0(W)\}\{h_0(W)\}^2\right]$.
\end{enumerate}
\end{theorem}

\begin{proof}
Fix $q\in\mathcal{Q}_0$. By \cite{abadie2016}, $\mathrm{gain}(q)$ is  $c_{q,\vartheta_0}^{\rm{T}}(\mathscr{I}_{q,\vartheta_0})^{-1}c_{q,\vartheta_0}$, where 
\begin{align}
\label{eq1:theo1}
c_{q,\vartheta_0}&=\begin{pmatrix}
E\left[
\mathrm{cov}(W_1,\mu_{q, \vartheta_0}(A,W)|\pi_{q, \vartheta_0}(W),A)\rho_{q, \vartheta_0}(W)\left(\frac{A}{\pi_{q, \vartheta_0}(W)^{2}}+\frac{1-A}{(1-\pi_{q, \vartheta_0}(W)^{2}}\right)\right]\\
\vdots\\
E\left[\mathrm{cov}(q_{m}(W),\mu_{q, \vartheta_0}(A,W)|\pi_{q, \vartheta_0}(W),A)\rho_{q, \vartheta_0}(W)\left(\frac{A}{\pi_{q, \vartheta_0}(W)^{2}}+\frac{1-A}{(1-\pi_{q, \vartheta_0}(W)^{2}}\right)\right]
\end{pmatrix}.
\end{align} 
Let $$g_{q,\vartheta_0}(a,w):=\frac{a(1-\pi_{q, \vartheta_0}(w))}{\pi_{q, \vartheta_0}(w)} + \frac{(1-a)\pi_{q, \vartheta_0}(w)}{1-\pi_{q, \vartheta_0}(w)}.$$ Conditioning on $a$, and $\pi_{q, \vartheta_0}(w)$, $g_{q,\vartheta_0}(a,w)$ is fixed. For an arbitrary $q_j(W)$, we have that
\begin{align}
\label{eq2theo1}
&E\left\{\mathrm{cov}(q_j(W),\mu_{q, \vartheta_0}(A,W)|\pi_{q, \vartheta_0}(W),A)\pi_{q, \vartheta_0}(W)(1-\pi_{q, \vartheta_0}(W))\left(\frac{A}{\pi_{q, \vartheta_0}(W)^2} + \frac{1-A}{(1-\pi_{q, \vartheta_0}(W))^2}\right)\right\} \nonumber \\ \nonumber
= &E\left\{\mathrm{cov}(q_j(W),\mu_{q, \vartheta_0}(A,W)|\pi_{q, \vartheta_0}(W),A)\left(\frac{A(1-\pi_{q, \vartheta_0}(W))}{\pi_{q, \vartheta_0}(W)} + \frac{(1-A)\pi_{q, \vartheta_0}(W)}{(1-\pi_{q, \vartheta_0}(W))}\right)\right\} \\ \nonumber
= &E\left\{\mathrm{cov}(q_j(W),\mu_{q, \vartheta_0}(A,W)|\pi_{q, \vartheta_0}(W),A)g_{q,\vartheta_0}(A,W)\right\} \\ \nonumber
= &E\bigg\{E\left[q_j(W)\mu_{q, \vartheta_0}(A,W)|\pi_{q, \vartheta_0}(W),A)\right]g_{q,\vartheta_0}(A,W)\\ \nonumber
& \quad - E\left[q_j(W)|\pi_{q, \vartheta_0}(W),A\right]E\left[\mu_{q, \vartheta_0}(A,W)|\pi_{q, \vartheta_0}(W),A\right]g_{q,\vartheta_0}(A,W)\bigg\} \\ \nonumber
= &E\bigg\{E\left[q_j(W)\mu_{q, \vartheta_0}(A,W)|\pi_{q, \vartheta_0}(W),A\right]g_{q,\vartheta_0}(A,W)\\ \nonumber
& \quad -E\left[q_j(W)|\pi_{q, \vartheta_0}(W),A\right]\bar{\mu}_{q, \vartheta_0}(A,\pi_{q, \vartheta_0}(W))g_{q,\vartheta_0}(A,W)\bigg\}  \\ \nonumber
=&E\left\{E\left[q_j(W)\mu_{q, \vartheta_0}(A,W)g_{q,\vartheta_0}(A,W)|\pi_{q, \vartheta_0}(W),A\right] - E\left[q_j(W)\bar{\mu}_{q, \vartheta_0}(A,\pi_{q, \vartheta_0}(W))g_{q,\vartheta_0}(A,W)|A,\pi_{q, \vartheta_0}(W)\right]\right\}  \\ \nonumber
=&E\left\{q_j(W)g_{q,\vartheta_0}(A,W)\mu_{q, \vartheta_0}(A,W) - q_j(W)\bar{\mu}_{q, \vartheta_0}(A,\pi_{q, \vartheta_0}(W))g_{q,\vartheta_0}(A,W)\right\} \\ \nonumber
=&E\left\{q_j(W)g_{q,\vartheta_0}(A,W)(\mu_{q, \vartheta_0}(A,W)-\bar{\mu}_{q, \vartheta_0}(A,\pi_{q, \vartheta_0}(W)))\right\} \\ \nonumber
=& E\left\{q_j(W)g_{q,\vartheta_0}(1,W)(\mu_{q, \vartheta_0}(1,W)-\bar{\mu}_{q, \vartheta_0}(1,\pi_{q, \vartheta_0}(W)))\pi_{h_0, \vartheta_0}(W)\right\} \\ \nonumber
&+E\left\{q_j(W)\left[g_{q,\vartheta_0}(0,W)(\mu_{q, \vartheta_0}(0,W)-\bar{\mu}_{q, \vartheta_0}(0,\pi_{q, \vartheta_0}(W)))(1-\pi_{h_0, \vartheta_0}(W))\right]\right\}\\ 
=&E\left\{q_j(W)\left[(\mu_{q, \vartheta_0}(1,W)-\bar{\mu}_{q, \vartheta_0}(1,\pi_{q, \vartheta_0}(W)))(1-\pi_{h_0, \vartheta_0}(W)) + (\mu_{q, \vartheta_0}(0,W)-\bar{\mu}_{q, \vartheta_0}(0,\pi_{q, \vartheta_0}(W)))\pi_{h_0, \vartheta_0}(W)\right]\right\}
\end{align}
The above identity also holds for an arbitrary covariate $W_j$. Next, by definition of $h_0$ in \eqref{defh}, and because $\pi_{q, \vartheta_0}(w) = \pi_{h_0, \vartheta_0}(w)$ for all $w \in \mathcal{W}$, we have that 
\begin{align*}
    (\mu_{q, \vartheta_0}(1,w)-\bar{\mu}_{q, \vartheta_0}(1,\pi_{q, \vartheta_0}(w)))(1-\pi_{h_0, \vartheta_0}(w)) + (\mu_{q, \vartheta_0}(0,w)-\bar{\mu}_{q, \vartheta_0}(0,\pi_{q, \vartheta_0}(w)))\pi_{h_0, \vartheta_0}(w)=\left[s_{q, \vartheta_0}(w)\right]^2h_0(w),
\end{align*}
where $s_{q, \vartheta_0}(w):= \sqrt{\pi_{q, \vartheta_0}(w)(1-\pi_{q, \vartheta_0}(w))}$. Therefore, the left-hand side of \eqref{eq2theo1} is equal to 
\begin{align*}
E\left\{[s_{q, \vartheta_0}(w)q_j(W)][s_{q, \vartheta_0}(w)h_0(W)]\right\}.
\end{align*}
The above identity implies that the vector $c_{q,\vartheta_0}$ can be re-expressed as 
\begin{align}
\label{eq2:theo1}
c_{\tilde{q},\vartheta_0}&:= \begin{pmatrix}
E\left[\tilde{W}_1\tilde{h}_0(W)\right]\\
\vdots\\
E\left[\tilde{q}_{m}(W)\tilde{h}_0(W) \right] 
\end{pmatrix}, 
\end{align}
where, for each $i\in \{1,\ldots,p\}$, $\tilde{w}_i:=s_{q, \vartheta_0}(w) w_i$, for each $j\in \{1,\ldots,m\}$ and $w \in \mathcal{W}$, $\tilde{q}_i(w):= s_{q, \vartheta_0}(w)q_i(w)$, and $\tilde{h}_0(w) :=  s_{q, \vartheta_0}h_0(w)$. Next, let $\tilde{q}(w):= (\tilde{q}_1(w),...\tilde{q}_{m}(w))$ denote the $m$ real valued-vector, where for each $i\in \{1,\ldots,m\}$, each entry is equal to the corresponding $\tilde{q}_{i}$ function evaluated at $w$. Finally, $\mathscr{I}_{q, \vartheta_0}$ can be re-expressed as

$$\mathscr{I}_{\tilde{q},\vartheta_0}:=E\begin{bmatrix}\tilde{W}\tilde{W}^{\rm{T}} & \tilde{W}\tilde{q}(W)^{\rm{T}}\\ \tilde{q}(W)\tilde{W}^{\rm{T}} & \tilde{q}(W)\tilde{q}(W)^{\rm{T}} \end{bmatrix},$$ 
Because $\mathscr{I}_{\tilde{q},\vartheta_0}$ is invertible, the covariates $\{\tilde{w}_1,\ldots,\tilde{q}_{m}(w)\}$ must be linearly independent. 
Thus, we can construct an orthonormal basis of the space spanned by the functions in $\{\tilde{w}_1,\ldots,\tilde{q}_{m}\}$ as follows. Let $B$ be a matrix such that $B^2=(\mathscr{I}_{\tilde{q},\vartheta_0})^{-1}$. 
Then, for every $i \in \{1,\ldots,p+m+1\}$, let $a_i$ be a $\mathcal{W}\rightarrow\mathbb{R}$ function defined as $a_i(w) := \sum_{j=1}^{p+1}B_{ij}\tilde{w}_{j} + \sum_{j=1}^{m}B_{ij}\tilde{q}_{j}(w)$, and, for a given $w$, let $a(w):=(a_1(w),\ldots a_{p+m+1}(w))$. By construction, the functions $\{a_1,...a_{p+m+1}\}$ are an orthonormal basis of the space spanned by the functions $\{\tilde{w}_1,\ldots,\tilde{q}_{p+m+1}\}$ with Gram matrix $\mathscr{I}_{\tilde{q},\vartheta_0}$. Additionally, we have that
\begin{align}
\label{eq3:theo1}
E[a_i(W)\tilde{h}_0(W)]= \sum_{j=1}^{p+1}B_{ij}E[\tilde{w}_j\tilde{h}_0(W)] + \sum_{j=1}^{m}B_{ij}E[\tilde{q}_j(W)\tilde{h}_0(W)].     
\end{align} 
Before proceeding further, we introduce some notation. For $f_1,\ldots,f_j\in L^2(P_{W,0})$, let ${\rm span}\{f_1,\ldots,f_j\}$ denote the linear span of $f_1,\ldots,f_j$. For a subspace $\mathcal{S}$ of $L^2(P_{W,0})$, let $\mathcal{S}^\perp$ denote the orthogonal complement of $\mathcal{S}$ in $L^2(P_{W,0})$ and let $\Pi(f|\mathcal{S})$ denote the $L^2(P_{W,0})$-projection of $f$ onto $\mathcal{S}$. When $\mathcal{S}$ is a one-dimensional space spanned by the function $a_i$, we write $\Pi(f\mid a_i)$ to mean $\Pi(f\mid {\rm span}\{a_i\})$. Given this notation, we can decompose $\tilde{h}_0$ as $\tilde{h}_0 = \Pi (\tilde{h}_0|{\rm span}\{a_1,\ldots,a_{m+p+1}\}^\perp) + \sum_{i=1}^{m+p+1} \Pi (\tilde{h}_0|a_i)$.  Then, letting \begin{align}
\label{eq4:theo1}
c_{a,\vartheta_0}&:= \begin{pmatrix}
E\left[a_1(W)\tilde{h}_0(W)\right]\\
\vdots\\
E\left[a_{p+m+1}(W)\tilde{h}_0(W) \right] 
\end{pmatrix},
\end{align} 
and by Equations \ref{eq2:theo1} and \ref{eq3:theo1}, $\mathrm{gain}(q)$ can be upper bounded as follows: 
\begin{align}
\mathrm{gain}(q) &= c^{\rm{T}}_{q,\vartheta_0}(\mathscr{I}_{q,\vartheta_0})^{-1} c_{q,\vartheta_0}= c^{\rm{T}}_{\tilde{q},\vartheta_0}(\mathscr{I}_{\tilde{q},\vartheta_0})^{-1}c_{\tilde{q},\vartheta_0} = c^{\rm{T}}_{\tilde{q},\vartheta_0}BBc_{\tilde{q},\vartheta_0}=
c_{a,\vartheta_0}^{\rm{T}}c_{a,\vartheta_0} \nonumber \\ 
&=\sum_{i = 1}^{m+p+1} E^2[a_i(W)\tilde{h}_0(W)]  = \sum_{i = 1}^{m+p+1}  E^2\left[a_i(W)\Pi (\tilde{h}_0|a_i)(W)\right] \label{eq:MqbEff} \\ 
&\leq  \sum_{i = 1}^{m+p+1} E[a_i^2(W)]E\left[\Pi (\tilde{h}_0|a_i)^2(W)\right] =\sum_{i = 1}^{m+p+1} E\left[\Pi (\tilde{h}_0|a_i)^2(W) \right]\leq E\left[\tilde{h}_0^2(W)\right].\nonumber
\end{align}
Now, consider the sub-model that contains $h_0$ as an augmentation covariate: 
\begin{align*}
\mathcal{M}(h_0) =\{\logit(\pi_{h_0, \bar{\theta}}(w)) &=  \theta_1^{\rm{T}}w + \theta_2 h_0(w): \; \theta_1 \in \mathbb{R}^{p+1}, \; p \in \mathbb{N}, \theta_2 \; \in \mathbb{R},  \mathscr{I}_{h_0, \vartheta_0} \text{ is invertible}\}.
\end{align*}
For $q(w)=h_0(w)$ it holds that $\mathcal{M}(h_0)=\mathcal{M}(q)$. Applying \eqref{eq:MqbEff} for this particular choice of $q$ and then using Parseval's identity shows that
\begin{align*}
    \mathrm{gain}(h_0)&= \sum_{i = 1}^{p+2}  E^2\left[a_i(W)\Pi (\tilde{h}_0|a_i)(W)\right] = E\left[\Pi(\tilde{h}_0\mid {\rm span}\{a_1,\ldots,a_{p+2}\})(W)^2\right].
\end{align*}
Because $q=h_0$, it holds that $\tilde{h}_0=\Pi(\tilde{h}_0\mid {\rm span}\{a_1,\ldots,a_{p+2}\})$. Hence, the above yields that $\mathrm{gain}(h_0)=E[\tilde{h}_0^2(W)]$, which shows that the inequality in \eqref{eq:MqbEff} is saturated for this particular choice of $q$, showing part (a) of Theorem \ref{theo:theo1}. Finally, by definition of $\tilde{h}$, we see that $\mathrm{gain}(h_0)=E\left[\pi_0(W)\{1-\pi_0(W)\}\{h_0(W)\}^2\right]$.

\end{proof}

\section{Proof of Corollary~\ref{cor1}}\label{supsec:sec3} 

In the statement and proof of the next corollary, we let $\bar{\pi}(a|w):= \pi_{h_0, \vartheta_0}(w)^a[1-\pi_{h_0, \vartheta_0}(w)]^{1-a}$. We also write $\bar{\pi}_1(W)$ to mean $\bar{\pi}(1|w)$. For notational brevity, and because all expectations are taken with respect to $P_{h_0, \vartheta_0}$, we will drop the subscripts $\vartheta_0,h_0$ from our notation. For instance, we will use the notation $\mu(a,w)$ and $\bar{\mu}(a, \pi_{1}(w))$ to denote $\mu_{h_0, \vartheta_0}(a,w)$ and $\bar{\mu}_{h_0, \vartheta_0}(a, \pi_{1}(w))$, respectively. Additionally, we let $\sigma^2_{NP}$ denote the nonparametric efficiency bound: $$\sigma^2_{NP}:=E\left[\frac{\sigma^2(1,W)}{\pi(W)}+\frac{\sigma^2(0,W)}{1-\pi(W)}+\left\{\mu(1,W)+\mu(0,W)-\psi\right\}^2\right].$$ 

\begin{corollary}Under the conditions of Theorem~\ref{theo:theo1}, the difference $\delta_M:=\sigma^2_M-\mathrm{gain}(h_0)-\sigma^2_{NP}$
between the asymptotic variance of the optimally-augmented 1:$M$ matching estimator and the nonparametric efficiency bound is given by
\begin{align*}
&\delta_M\,=\,\frac{1}{2M}\,E_0\left[\left\{\frac{1}{\pi_0(W)}-\pi_0(W)\right\}\left\{\sigma^2(1,W)+\zeta(1,W)\right\}\right.\\
&\hspace{1.2in}\left.+\left\{\frac{1}{1-\pi_0(W)}-1+\pi_0(W)\right\}\left\{\sigma^2(0,W)+\zeta(0,W)\right\}\right]\,\geq\,0\ ,
\end{align*}where we define $\zeta(a,w):=var_0\left\{\mu_0(a,W)\,|\,\pi_0(W)=\pi_0(w)\right\}$.
\label{cor1}
\end{corollary}

\begin{proof}
First, note that we can re-express $\delta_m$ as
\begin{align}
\delta_m &= \frac{1}{2M}\sum_{a=0}^1E\left\{\left[\frac{1}{\bar{\pi}(a|W)} - \bar{\pi}(a|W)\right]\left[\sigma^2(a,W) +\zeta(a,W)\right]\right\}.
\end{align} Let $\gamma(w) = \mu(1,w)-\mu(0,w)$, and $\bar{\gamma}(w) = \bar{\mu}(1,w) - \bar{\mu}(0,w)$.  Now, let $s_M$, be a $(0,1) \rightarrow \mathbb{R}$ function defined as 
\begin{align*}
    s_M(p) &:= p^{-1} + (2M)^{-1}\left(p^{-1} - p\right). 
\end{align*}
We write $s_M\circ \bar{\pi}(a|w)$ to mean $s_M(\bar{\pi}(a|w))$. Then,  $\sigma_M^2$ can be expressed as 
\begin{align*}
\sigma_M^2 = &E\left[(\bar{\gamma}(W)-\psi_0)^2\right]  + \sum_{a=0}^1 E\left[\bar{\sigma}^2(a,\bar{\pi}_1(W)) s_M\circ\bar{\pi}(a|W)\right]. 
\end{align*}
Next, the variance of an asymptotically efficient estimator is equal to 
\begin{align*}
\sigma_{NP}^2&= E\left[\left(\frac{2A-1}{\bar{\pi}(A|W)}(Y-\mu(A,W)) + (\gamma(W) - \psi_0)\right)^2\right] = E\left[\sum_{a = 0}^{1}\frac{\sigma^2(a,W)}{\bar{\pi}(a|W)}\right] + E\left[(\gamma(W) - \psi_0)^2\right].
\end{align*}
Hence, 
\begin{align*}
\sigma_M^2 - \sigma^2_{NP} &= E\left[(\bar{\gamma}(W)-\psi_0)^2\right]  + \sum_{a=0}^1 E\left[\bar{\sigma}^2(a,\bar{\pi}_1(W)) s_M\circ\bar{\pi}(a|W)\right] \\
&\quad - \left\{E\left[\sum_{a = 0}^{1}\frac{\sigma^2(a,W)}{\bar{\pi}(a|W)}\right] + E\left[(\gamma(W) - \psi_0)^2\right] \right\}  \\
&= E\left[(\bar{\gamma}(W)-\psi_0)^2\right]
- E\left[(\gamma(W) - \psi_0)^2\right]+ \sum_{a=0}^1 E\left[\bar{\sigma}^2(a,\bar{\pi}_1(W)) s_M\circ\bar{\pi}(a|W) - \frac{\sigma^2(a,W)}{\bar{\pi}(a|W)}\right].
\end{align*}
Let $\Gamma(w) := {\rm{Var}}(\gamma(W)|\bar{\pi}_{1}(W) = \bar{\pi}_{1}(w))$. Noting that $\tau=E[\bar{\gamma}(W)]=E[\gamma(W)]$ and $E[\gamma(W)|\bar{\pi}_1(W)]=\bar{\gamma}(W)$ almost surely, it is not difficult to show that $E\left[(\bar{\gamma}(W)-\psi_0)^2\right]
- E\left[(\gamma(W) - \psi_0)^2\right] = 
-E[\Gamma(W)]$. Hence, the above rewrites as follows:
\begin{align}
\label{eq2:cor2}
\sigma_M^2 - \sigma^2_{NP} = &-E[\Gamma(W)] +\sum_{a=0}^1 E\left[\bar{\sigma}^2(a,\bar{\pi}_1(W)) s_M\circ\bar{\pi}(a|W) - \frac{\sigma^2(a,W)}{\bar{\pi}(a|W)}\right].
\end{align}
We now begin to study the latter term above. By the law of total variance, the following holds for $P_0$-almost all $w$:
\begin{align*}
    &\sum_{a=0}^1 \left[\bar{\sigma}^2(a,\bar{\pi}_1(w)) s_M\circ\bar{\pi}(a|w)\right]=\sum_{a=0}^1 \left[E\left\{\sigma^2(a,W)|\bar{\pi}_1(W)=\bar{\pi}_1(w)\right\}s_M\circ\bar{\pi}(a|w)+\zeta(a,w) s_M\circ\bar{\pi}(a|w)\right].
\end{align*}
Integrating $w$ on both sides against the marginal distribution under $P_0$ and applying the law of total expectation then shows that
\begin{align*}
    \sum_{a=0}^1 &E\left[\bar{\sigma}^2(a,\bar{\pi}_1(W)) s_M\circ\bar{\pi}(a|W)\right]= \sum_{a=0}^1 E\left[\sigma^2(a,W)s_M\circ\bar{\pi}(a|W)+\zeta(a,w)s_M\circ\bar{\pi}(a|W)\right].
\end{align*}
Plugging this into the latter term in \eqref{eq2:cor2} and simplifying, we find that
\begin{align*}
    \sum_{a=0}^1 E\left[\bar{\sigma}^2(a,\bar{\pi}_1(W)) s_M\circ\bar{\pi}(a|W) - \frac{\sigma^2(a,W)}{\bar{\pi}(a|W)}\right] 
    &= \sum_{a=0}^1 E\left[\frac{\zeta(a,w)}{\bar{\pi}(a|W)}\right] + \delta_m.
\end{align*}
Hence,
\begin{align}
\label{eq4:cor2}
\sigma_M^2 - \sigma^2_{NP} =&-E[\Gamma(W)] + \sum_{a = 0}^1 E\left[\frac{\zeta(a,w)}{\bar{\pi}(a|W)}\right] + \delta_m.
\end{align}
Let $\Delta(a,w):=\mu(a,w)-\bar{\mu}(a,\bar{\pi}_1(w))$. Now, because
\begin{align}
\label{def:effhstr}
    \mathrm{gain}(h_0)=E\left[\left(\sqrt{\pi_{h_0, \vartheta_0}(W)(1-\pi_{h_0, \vartheta_0}(W))}h_0(W)\right)^2\right]
\end{align}
we have that $\mathrm{gain}(h_0)$ is equal to
\begin{align*}
\mathrm{gain}(h_0)=&E\left[\bar{\pi}_1(W)\bar{\pi}_0(W)h_0^2(W)\right] 
= \sum_{a = 0}^1 E\left[\frac{1-\bar{\pi}(a|W)}{\bar{\pi}(a|W)}\Delta^2(a,W)\right]+2E\left[\Delta(1,W)\Delta(0,W)\right].
\end{align*}
Next, for a given $a \in \{0,1\}$, and for $P_0$-almost all $w$, $\zeta(a,w) =E\left[\Delta^2(a,W)|\bar{\pi}_1(W)=\bar{\pi}_1(w)\right]$. 
Applying the law of total expectation on the right-hand side above and plugging this in yields that
\begin{align*}
\mathrm{gain}(h_0)&=\sum_{a = 0}^1 E\left[\left(\frac{1}{\bar{\pi}(a|W)}-1\right)\zeta(a,w)\right] +2E\left\{\Delta(1,W)\Delta(0,W)\right\}.
\end{align*}
Now, let $\xi(w) := \mathrm{cov}\left(\mu(1,W), \mu(0,W)|\bar{\pi}_1(W) = \bar{\pi}_1(w)\right)$, and note that $\Gamma(w) = \sum_{a = 0}^1\zeta(a,w) - 2\xi(w)$ for $P_0$-almost all $w$. 
Plugging this into the above yields that
\begin{align}
\label{eq7:cor2}
\mathrm{gain}(h_0)&=-E\left[\Gamma(W)\right] - 2E\left[\xi(W)\right] + \sum_{a=0}^1 E\left\{\frac{\zeta(a,w)}{\bar{\pi}(a|W)}\right\}  +2E\left[\Delta(1,W)\Delta(0,W)\right].  
\end{align}
By the tower property, $E\left[\mu(a, W)\bar{\mu}(1-a,\bar{\pi}_1(W))\right]= E\left[\bar{\mu}(a, \bar{\pi}_1(W)\bar{\mu}(1-a,\bar{\pi}_1(W))\right]$ for $a \in \{0,1\}$.
Hence,
\begin{align}
\label{eq8:cor2}
&E\left[\Delta(1,W)\Delta(0,W)\right] 
=E\left[\mu(1, W)\mu(0, W)\right] - E\left\{\bar{\mu}(0,\bar{\pi}_1(W))\bar{\mu}(1,\bar{\pi}_1(W))\right\}.
\end{align} 
Moreover, the definition of conditional covariance and the tower property imply that
\begin{align}
\label{eq9:cor2}
E\left[\xi(W)\right]=&E[\mu(1,W)\mu(0,W)]-E\left\{\bar{\mu}(1,\bar{\pi}_1(W)\bar{\mu}(0,\bar{\pi}_1(W))\right\}.
\end{align}
Combining Equations \ref{eq7:cor2}, \ref{eq8:cor2}, and \ref{eq9:cor2}, we have that
\begin{align*}
\mathrm{gain}(h_0) =&-E[\Gamma(W)] + \sum_{a = 0}^1 E\left[\frac{\zeta(a,w)}{\bar{\pi}(a|W)}\right].
\end{align*}
Equation \ref{eq4:cor2} and the above identity show that
\begin{align*}
\sigma_M^2 - \sigma^2_{NP} = \mathrm{gain}(h_0)+ \delta_m.
\end{align*}
\end{proof}

\section{Supporting lemmas of Proposition \ref{prop1}} \label{subsec:sec4}

As a first step, we want to show that model $\{\bar{P}_{\hd,\vartheta}; \vartheta \in \mathbb{R}^{p+2}, \hd \in \mathcal{H}\}$ is uniformly differentiable in quadratic mean (UDQM) at every $\vartheta \in \Theta$, where $\Theta$ is an open subset of $\mathbb{R}^{p+2}$. By UDQM, we mean that 
   $$\sup_{\hd \in \mathcal{H}}\int\left[\sqrt{\bar{p}_{\hd, \vartheta + g}(z)} - \sqrt{\bar{p}_{\hd,\vartheta}(z)} - \frac{1}{2}g^T U_{\hd}(\vartheta;z)\sqrt{\bar{p}_{\hd,\vartheta}(z)}\right]^2 d\bar{P}(z) = o(\|g\|^2), \; g \rightarrow 0 $$ 
Before we do this, we establish quadratic mean differentiability for each fixed $\hd$.

\begin{lemma}\label{lem:qmd}
Assume Conditions \ref{cond:pos:first}, \ref{cond:qmd}, and \ref{cond:qmd2} hold, then for each $\hd \in \mathcal{H}$, $\{\bar{P}_{\hd,\vartheta}; \vartheta \in \mathbb{R}^{p+2}\}$ is quadratic mean differentiable.
\end{lemma}

\begin{proof}[Proof of Lemma~\ref{lem:qmd}]
Fix $\hd \in \mathcal{H}$. In what follows, we will show that the conditions of Lemma~7.6 in \cite{van2000asymptotic} are satisfied, which will give the result.

Let $s_{\hd,\vartheta} := \sqrt{\bar{p}_{\hd,\vartheta}}$. We first show that, for each $z$, $\vartheta \mapsto s_{\hd,\vartheta}(z)$ is continuously differentiable. Fix $z$, noting that 
$\pi_{\hd,\vartheta}(w) = {\rm expit}[\vartheta^{\rm{T}}  r_\hd(w)]$, we see that $\vartheta\mapsto \pi_{\hd,\vartheta}(w)$ is the composition of two continuously differentiable functions, namely ${\rm expit}(\cdot)$ and $\vartheta \mapsto \vartheta^{\rm{T}} r_\hd(w)$, and so is itself continuously differentiable. 
Next, let $\dot{p}_{\hd,\vartheta}(z) := \nabla_\vartheta \bar{p}_{\hd,\vartheta}(z)$. 
It can be shown that
\begin{align*}
\dot{p}_{\hd,\vartheta}(z)
&= \frac{(-1)^a\nabla_\vartheta \pi_{\hd,\vartheta}(w)}{\pi(w)^a[1-\pi(w)]^{1-a}}
 = \frac{(-1)^a \rho(\vartheta^{\rm{T}} r_\hd(w))r_\hd(w)}{\pi(w)^a[1-\pi(w)]^{1-a}}.
\end{align*}
Hence, $\vartheta\mapsto \dot{p}_{\hd,\vartheta}(z)$ is continuous in $\vartheta$. Finally, let $\dot{s}_{\hd,\vartheta}(z):= \nabla_\vartheta s_{\hd,\vartheta}(z)$ be the gradient of $\vartheta \mapsto s_{\hd,\vartheta}(z)$.  Noting that $p_{\hd,\vartheta}(z)>0$ for all $\vartheta$, we see that
$$\dot{s}_{\hd,\vartheta}(z):= \nabla_\vartheta  s_{h, \vartheta}(z) = \nabla_\vartheta  \sqrt{\bar{p}_{\hd,\vartheta}(z)} = \dot{p}_{\hd,\vartheta}(z)/\left[2s_{\hd,\vartheta}(z)\right].$$
Because $\vartheta \mapsto \bar{p}_{\hd,\vartheta}(z)$ is a strictly positive, continuous function and that $x \mapsto \sqrt{x}$ and $x \mapsto 1/x $ are continuous on $(0,\infty)$, we see that $\vartheta\mapsto \bar{p}_{\hd,\vartheta}(z)^{-1/2}$ is continuous. Hence, $z\mapsto \dot{s}_{\hd,\vartheta}(z)$ is continuous, and so the map $\vartheta \mapsto s_{\hd,\vartheta}$ is continuously differentiable.

We now show that the elements of the matrix $\mathscr{I}_{\hd,\vartheta} = \int (\dot{p}_{\hd,\vartheta}/\bar{p}_{\hd,\vartheta})(\dot{p}_{\hd,\vartheta}^{\rm{T}}/\bar{p}_{\hd,\vartheta}) \bar{p}_{\hd,\vartheta} d\mu$ are well defined and continuous in $\vartheta$. Recalling that $\dot{p}_{\hd,\vartheta}(z)=\{(-1)^a \rho(\vartheta^{\rm{T}} r_\hd(w))r_\hd(w)\}/\{\pi(w)^a[1-\pi(w)]^{1-a}\}$, we see that
\begin{align*}
(\dot{p}_{\hd,\vartheta}/\bar{p}_{\hd,\vartheta})(\dot{p}_{\hd,\vartheta}^{\rm{T}}/\bar{p}_{\hd,\vartheta}) &= \frac{1}{[\pi_{\hd,\vartheta}(w)^a(1-\pi_{\hd,\vartheta}(w))^{1-a}]^2}\rho(\vartheta^{\rm{T}} r_\hd(w))r_\hd(w) r_\hd(w)^{\rm{T}} \rho(\vartheta^{\rm{T}} r_\hd(w)) \\
&= \frac{\rho(\vartheta^{\rm{T}} r_\hd(w))^2}{[\pi_{\hd,\vartheta}(w)^a(1-\pi_{\hd,\vartheta}(w))^{1-a}]^2} r_\hd(w)r_\hd(w)^{\rm{T}} 
\end{align*}
Below we use $E_{\bar{P}_{W}}$ to denote the expectation operator under $\bar{P}_{W}$. We have that
\begin{align*}
 \mathscr{I}_{\hd,\vartheta}&= \int (\dot{p}_{\hd,\vartheta}/\bar{p}_{\hd,\vartheta})(\dot{p}_{\hd,\vartheta}^{\rm{T}}/\bar{p}_{\hd,\vartheta}) d\bar{P}_{\hd,\vartheta} \\
 &= \int_{\mathcal{W}}\left( \sum_{a \in \{1, 0\}} \frac{\rho(\vartheta^{\rm{T}} r_\hd(w))^2r_\hd(w)r_\hd(w)^{\rm{T}}}{[\pi_{\hd,\vartheta}(w)^a(1-\pi_{\hd,\vartheta}(w))^{1-a}]^2}[\pi_{\hd,\vartheta}(w)^a(1-\pi_{\hd,\vartheta}(w))^{1-a}]\right)d\bar{P}_{W}\\
 &= \int_{\mathcal{W}}\left( \sum_{a \in \{1, 0\}} \frac{\rho(\vartheta^{\rm{T}} r_\hd(w))^2r_\hd(w)r_\hd(w)^{\rm{T}}}{\pi_{\hd,\vartheta}(w)^a(1-\pi_{\hd,\vartheta}(w))^{1-a}}\right)d\bar{P}_{W} \\
 &= \int_{\mathcal{W}}\frac{\rho(\vartheta^{\rm{T}} r_\hd(w))^2r_\hd(w)r_\hd(w)^{\rm{T}}}{\pi_{\hd,\vartheta}(w)(1-\pi_{\hd,\vartheta}(w))}d\bar{P}_{W} \\
 &= E_{\bar{P}_W}\left[\frac{\rho(\vartheta^{\rm{T}} r_\hd(W))^2r_\hd(W)r_\hd(W)^{\rm{T}}}{\pi_{\hd,\vartheta}(W)(1-\pi_{\hd,\vartheta}(W))}\right],
\end{align*}
the third equality uses that $1/t + 1/(1-t)=1/[t(1-t)]$. 
By Condition \ref{cond:qmd2}, $\hd(\cdot)$ is bounded. Hence, there exists a finite constant $C_u$ such that 
	\begin{equation*}
	   0\leq \frac{\rho(\vartheta^{\rm{T}} r_\hd(w))^2}{\pi_{\hd,\vartheta}(w)(1-\pi_{\hd,\vartheta}(w))} \leq C_u \ \ \textnormal{for all } w \in \mathcal{W}, \textnormal{ and all } \vartheta \in \mathbb{R}^{p+2}.
	\end{equation*}
Combining the bound on $\hd(\cdot)$ with the fact that $\mathcal{W}$ is bounded shows that the entries of the matrix $r_{\hd}(w)r_{\hd}(w)^{\rm{T}}$ are also bounded in $w$. Therefore, $\mathscr{I}_{\hd,\vartheta}$ is well-defined and finite. 
Next, note that the entries of the matrix  $$r_{\hd}(w)r_{\hd}(w)^{\rm{T}}\frac{\rho_{\hd,\vartheta}(w)}{\pi_{\hd,\vartheta}(w)^2(1-\pi_{\hd,\vartheta}(w))}$$ correspond to the composition of continuous functions in $\vartheta$ and, as such, are also continuous in $\vartheta$ for each $w$.  Therefore, for every $w$ and every $\mathbb{R}^{p+2}$-valued sequence $\{\vartheta_n\}_{n=1}^\infty$ that converges to $\vartheta$, we have that 
\begin{equation*}
r_{\hd}(w)r_{\hd}(w)^{\rm{T}} \frac{\rho(\vartheta_n^{\rm{T}} r_\hd(w))^2}{\pi_{\hd, \vartheta_n}(w)(1-\pi_{\hd, \vartheta_n}(w))} \longrightarrow  r_{\hd}(w)r_{\hd}(w)^{\rm{T}} \frac{\rho(\vartheta^{\rm{T}} r_\hd(w))^2}{\pi_{\hd,\vartheta}(w)(1-\pi_{\hd,\vartheta}(w))}.
\end{equation*}
Because each entry of $r_{\hd}(w)r_{\hd}(w)^{\rm{T}} \frac{\rho(\vartheta_n^{\rm{T}} r_\hd(w))^2}{\pi_{\hd, \vartheta_n}(w)(1-\pi_{\hd, \vartheta_n}(w))}$ is bounded by a constant, the Dominated Convergence Theorem implies that
\begin{equation*}
     E_{\bar{P}_W}\left[\frac{\rho(\vartheta_n^{\rm{T}} r_\hd(W))^2r_\hd(W)r_\hd(W)^{\rm{T}}}{\pi_{\hd, \vartheta_n}(W)(1-\pi_{\hd, \vartheta_n}(W))}\right] \longrightarrow E_{\bar{P}_W}\left[\frac{\rho(\vartheta^{\rm{T}} r_\hd(W))^2r_\hd(W)r_\hd(W)^{\rm{T}}}{\pi_{\hd,\vartheta}(W)(1-\pi_{\hd,\vartheta}(W))}\right]
\end{equation*}

Therefore. $\mathscr{I}_{\hd,\vartheta}$ is continuous in $\vartheta$. By Lemma~7.6 in \cite{van2000asymptotic}, the model $\{\bar{P}_{\hd,\vartheta}; \vartheta \in \mathbb{R}^{p+2}\}$ is quadratic mean differentiable.
\end{proof}

\begin{lemma}\label{lem:uqmd}
If the conditions of Lemma~\ref{lem:qmd} hold, then $\{\bar{P}_{\hd,\vartheta}; \vartheta \in \mathbb{R}^{p+2}, \hd \in \mathcal{H}\}$ is uniformly differentiable in quadratic mean.
\end{lemma}

Before proving the above result, we establish the following useful lemma. Concisely stating this lemma requires a slight abuse of notation. 
We first fix an open, bounded, convex set $H$ that contains the image of $\hd$. For each $\eta\in H$, we then define $\pi_{\eta, \vartheta}(w) :=  {\rm expit}(\vartheta^{\rm{T}} (w,\eta))$, $\bar{p}_{\eta, \vartheta}(z) := \pi_{\eta, \vartheta}(w)^a(1-\pi_{\eta, \vartheta}(w))^{1-a}/\{\pi(w)^a(1-\pi(w))^{1-a}\}$, and $U_{\eta}(\vartheta;z) = a\frac{\partial}{\partial \vartheta} \log\pi_{\eta, \vartheta}(w) + (1-a)\frac{\partial}{\partial \vartheta}\log\{1-\pi_{\eta, \vartheta}(w)\}$. For each $\vartheta,g\in \mathbb{R}^{p+2}$ and $z\in\mathcal{Z}$, define the following $H\rightarrow \mathbb{R}$ function:
\begin{align}
    f_{\vartheta,g,z}(\eta):= \left[\bar{p}_{\eta, \vartheta + g}(z)^{1/2} - \bar{p}_{\eta, \vartheta}(z)^{1/2} - \frac{1}{2}g^T U_{\eta}(\vartheta;z) \bar{p}_{\eta, \vartheta}(z)^{1/2} \right]^2. \label{eq:fdef}
\end{align}
\begin{lemma}\label{lem:Lipschitz}
Fix $\vartheta\in\mathbb{R}^{p+2}$ and a bounded neighborhood $\mathcal{N}\subset\mathbb{R}^{p+2}$ of zero. If the conditions of Lemma~\ref{lem:qmd} hold, there exists an $L>0$ such that, for all $g\in \mathcal{N}$ and $z\in\mathcal{Z}$, $f_{\vartheta,g,z}(\cdot)$ is $L$-Lipschitz.
\end{lemma}

\begin{proof}[Proof of Lemma~\ref{lem:Lipschitz}] 
Fix $\vartheta\in\mathbb{R}^{p+2}$ and a bounded neighborhood $\mathcal{N}$ of zero. Though not indicated in the notation, all constants displayed in this proof may depend on the choice of $\vartheta$ and $\mathcal{N}$. We show that $\big |\frac{\partial}{\partial \eta} f_{\vartheta,g,z}(\eta)\big |$ is uniformly bounded over $g\in\mathcal{N}$, $\eta\in H$, and $z\in\mathcal{Z}$.

For each $v \in \mathbb{R}^{p+2}$, $\eta\in H$, and $z=(w,a)\in\mathcal{Z}$, let $q_{\eta, v}(z) := [\frac{\partial}{\partial \eta}\bar{p}_{\eta, v}(z)]/[\bar{p}_{\eta, v}(z)^{1/2}]$. It can be shown that 
\begin{align}
    q_{\eta, v}(z) =  \frac{1}{\{\pi(w)^a[1-\pi(w)]^{1-a}\}^{1/2}}\left\{\frac{(-1)^{1-a}\pi_{\eta, v}(w)[1-\pi_{\eta, v}(w)]}{\{\pi_{\eta, v}(w)^{a}[1-\pi_{\eta, v}(w)]^{1-a}\}^{1/2}}\right\}v_{p+2}.\label{eq:qveta}
\end{align}
Further, 
\begin{align*}
    \frac{\partial}{\partial \eta}U_{\eta}(\vartheta; z) &=  \frac{\partial}{\partial \eta}\left\{a\nabla_\vartheta\log\{\pi_{\eta, \vartheta}(w)\} + (1-a)\nabla_\vartheta\log\{1-\pi_{\eta, \vartheta}(w)\}\right\}  \\
    &= \frac{\partial}{\partial \eta}\left\{(w,\eta)(a-\pi_{\eta, \vartheta}(w))\right\} \\
    &= -\rho(\vartheta^{\rm{T}}(w,\eta))\vartheta_{p+2}(w, \eta  -[a-\pi_{\eta, \vartheta}(w)]).
\end{align*}
For each $v \in \mathbb{R}^{p+2}$, $\eta\in H$, $z\in\mathcal{Z}$, and $g\in\mathcal{N}$, let $\beta_{\vartheta,g,\eta,z}:=\bar{p}_{\eta, \vartheta + g}(z)^{1/2} - \bar{p}_{\vartheta,\eta}(z)^{1/2} - \frac{1}{2}g^{\rm{T}} U_{\eta}(\vartheta;z)\bar{p}_{\vartheta,\eta}(z)^{1/2}$. We have that
\begin{align}
\label{derf}
&\frac{\partial}{\partial \eta}f_{\vartheta,g,z}(\eta) =\\ \nonumber
&= 2\beta_{\vartheta,g,\eta,z}\left\{\frac{d}{d\eta}\bar{p}_{\eta, \vartheta + g}(z)^{1/2} - \frac{d}{d\eta}\bar{p}_{\eta, \vartheta}(z)^{1/2} - \frac{d}{d\eta}[\frac{1}{2}g^T U_{\eta}(\vartheta;z)\bar{p}_{\eta, \vartheta}(z)^{1/2}] \right\} \\ \nonumber
&= 2 \beta_{\vartheta,g,\eta,z}\Bigg\{\frac{\frac{\partial}{\partial \eta}\bar{p}_{\eta, \vartheta + g}(z)}{2\bar{p}_{\eta, \vartheta + g}(z)^{1/2}} -\frac{\frac{\partial}{\partial \eta}\bar{p}_{\eta, \vartheta}(z)}{2\bar{p}_{\eta, \vartheta}(z)^{1/2}}  \\ \nonumber
&\hspace{5.25em}- \frac{1}{2} g^{\rm{T}} \left(\frac{\partial}{\partial\eta} U_{\eta}(\vartheta;z)\right)\bar{p}_{\eta, \vartheta}(z)^{1/2} - \frac{1}{2} g^{\rm{T}} U_{\eta}(\vartheta;z)\left(\frac{\frac{\partial}{\partial \eta}\bar{p}_{\eta, \vartheta}(z)}{2\bar{p}_{\eta, \vartheta}(z)^{1/2}}\right) \Bigg\} \\
\nonumber
&= \beta_{\vartheta,g,\eta,z}\left\{\frac{\frac{\partial}{\partial \eta}\bar{p}_{\eta, \vartheta + g}(z)}{\bar{p}_{\eta, \vartheta + g}(z)^{1/2}} -\frac{\frac{\partial}{\partial \eta}\bar{p}_{\eta, \vartheta}(z)}{\bar{p}_{\eta, \vartheta}(z)^{1/2}} - g^{\rm{T}} \left(\frac{\partial}{\partial\eta} U_{\eta}(\vartheta;z)\right)\bar{p}_{\eta, \vartheta}(z)^{1/2} - \frac{1}{2} g^{\rm{T}} U_{\eta}(\vartheta;z)\left(\frac{\frac{\partial}{\partial \eta}\bar{p}_{\eta, \vartheta}(z)}{\bar{p}_{\eta, \vartheta}(z)^{1/2}}\right)\right\}   \\
\nonumber
&= \beta_{\vartheta,g,\eta,z}\left\{q_{\eta, \vartheta +g}(z) - q_{\eta, \vartheta}(z)  - g^{\rm{T}} \left(\frac{\partial}{\partial\eta} U_{\eta}(\vartheta;z)\right)\bar{p}_{\eta, \vartheta}(z)^{1/2} - \frac{1}{2} g^{\rm{T}} U_{\eta}(\vartheta;z)q_{\eta, \vartheta}(z)\right\} \\
\nonumber
&= \beta_{\vartheta,g,\eta,z}\left\{q_{\eta, \vartheta +g}(z) - q_{\eta, \vartheta}(z)  - g^{\rm{T}} \left(\frac{\partial}{\partial\eta} U_{\eta}(\vartheta;z)\right)[\bar{p}_{\eta, \vartheta}(z)^{1/2} + q_{\eta, \vartheta}(z) ]\right\} \\
\nonumber
&= \beta_{\vartheta,g,\eta,z}\left\{q_{\eta, \vartheta +g}(z) - q_{\eta, \vartheta}(z)  + \rho(\vartheta^{\rm{T}}(w,\eta))\vartheta_{p+2}g^{\rm{T}}(w, \eta  -[a-\pi_{\eta, \vartheta}(w)])[\bar{p}_{\eta, \vartheta}(z)^{1/2} + q_{\eta, \vartheta}(z) ]\right\}.
\end{align}
Because all $\hd \in \mathcal{H}$ are bounded functions, the domain of $f_{\vartheta,g,z}$, namely $H$, is bounded. Moreover, because $\mathcal{W}$ is bounded, we see that $\sup_{w\in\mathcal{W},g\in\mathcal{N}}|(\vartheta + g)^{\rm{T}} (w,\eta)|<\infty$, and so there exists a $\delta>0$ so that
\begin{align}
\label{eq:posit}
    \delta \le \inf_{w\in\mathcal{W},g\in\mathcal{N}} \pi_{\eta, \vartheta + g}(w)\le \sup_{w\in\mathcal{W},g\in\mathcal{N}} \pi_{\eta, \vartheta + g}(w)\le 1-\delta.
\end{align}
Plugging this observation and a similar observation about the propensity $\pi$ into \eqref{eq:qveta} 
shows that there exist finite constants $C_1,C_2$ such that, for all $g \in \mathcal{N}$, $\eta\in H$, and $z\in\mathcal{Z}$, $|q_{\eta, \vartheta + g}(z)|\le C_1 (|\vartheta_{p+2}| + |g_{p+2}|)\le C_1|\vartheta_{p+2}| + C_2$, where the latter inequality used that $\mathcal{N}$ is bounded. Let $C:= C_1|\vartheta_{p+2}| + C_2$.
Next, there exists an $M<\infty$ such that $\sup_{g\in\mathcal{N},\eta\in H,z\in\mathcal{Z}}|g^{\rm{T}}(w,\eta - [a - \pi_{\eta, \vartheta}(w)])|<M$.  Moreover, because $\rho(x)\leq 1$ for all $x \in \mathbb{R}$ such that
$$\sup_{g\in\mathcal{N},\eta\in H,z\in\mathcal{Z}}|\rho(\vartheta^{\rm{T}}(w,\eta))\vartheta_{p+2}g^{\rm{T}}(w, \eta - [a-\pi_{\eta, \vartheta}(w)])| < M|\vartheta_{p+2}|.$$ 
Lastly, by the aforementioned arguments, there exists a $D<\infty$ such that $\sup_{g\in\mathcal{N},\eta\in H,z\in\mathcal{Z}}\bar{p}_{\eta, \vartheta + g}(z)\le D$. Therefore, for all $g\in\mathcal{N}$, $\eta\in H$, and $z\in\mathcal{Z}$,
\begin{align*}
    &q_{\eta, \vartheta +g}(z) - q_{\eta, \vartheta}(z)  + \rho(\vartheta^{\rm{T}}(w,\eta))\vartheta_{p+2}g^{\rm{T}}(w, \eta  -[a-\pi_{\eta, \vartheta}(w)])[\bar{p}_{\eta, \vartheta}(z)^{1/2} + q_{\eta, \vartheta}(z) ] \\
    &\leq |q_{\eta, \vartheta +g}(z)| + |q_{\eta, \vartheta}(z)|  + |\rho(\vartheta^{\rm{T}}(w,\eta))\vartheta_{p+2}g^{\rm{T}}(w, \eta  - [a-\pi_{\eta, \vartheta}(w)])| |\bar{p}_{\eta, \vartheta}(z)^{1/2} + q_{\eta, \vartheta}(z) | \\
    &\leq  2C + |\vartheta_{p+2}|M(D^{1/2} + C),
\end{align*}
and also
\begin{align*}
\beta_{\vartheta, g, \eta, z}= |\bar{p}_{\eta, \vartheta + g}(z)^{1/2} - \bar{p}_{\eta, \vartheta}(z)^{1/2} - \frac{1}{2}g^{\rm{T}} U_{\eta}(\vartheta;z)\bar{p}_{\eta, \vartheta}(z)^{1/2}| \leq 2D^{1/2} +  |\vartheta_{p+2}|MD^{1/2}.
\end{align*}
Combining the above bounds with the display in equation  \ref{derf}, we have that 
$$\sup_{g\in\mathcal{N},\eta\in H,z\in\mathcal{Z}}\bigg|\frac{\partial}{\partial \eta}f_{\vartheta,g,z}(\eta)\bigg| < (2D^{1/2} +  |\vartheta_{p+2}|MD^{1/2})[2C + |\vartheta_{p+2}|M(D^{1/2} + C)]=:L.$$
Hence, $f_{\vartheta,g,z}(\cdot)$ is $L$-Lipschitz for all $g \in \mathcal{N}$ and $z \in \mathcal{Z}$.
\end{proof}

\begin{proof}[Proof of Lemma~\ref{lem:uqmd}] 
Fix $\vartheta \in \mathbb{R}^{p+2}$. We show that 
\begin{align*}
   \sup_{\hd \in \mathcal{H}} \int \left(\sqrt{\bar{p}_{\hd, \vartheta + g}(z)} - \sqrt{\bar{p}_{\hd,\vartheta}(z)} - \frac{1}{2}g^T U_{\hd}(\vartheta;z)\sqrt{\bar{p}_{\hd,\vartheta}(z)} \right)^2 d\bar{P}_Z(z)  \rightarrow 0 \text{ as } g \rightarrow 0.
\end{align*}
For $g\in\mathcal{N}$ and $\hd_1,\hd_2\in\mathcal{H}$, define
\begin{align*}
    \Psi_{g}(\hd_1,\hd_2):= \left|\int [f_{\vartheta, g, z}\circ \hd_1(w) - f_{\vartheta, g, z}\circ \hd_2(w)] d\bar{P}_Z(z) \right|
\end{align*}
By two applications of Jensen's inequality and Lemma~\ref{lem:Lipschitz}, the following holds for all $g \in \mathcal{N}$ and $\hd_1,\hd_2\in\mathcal{H}$:
\begin{align}
   \Psi_{g}(\hd_1,\hd_2)
   &\leq \int |f_{\vartheta, g, z}\circ \hd_1(w) - f_{\vartheta, g, z}\circ \hd_2(w)| d\bar{P}_Z(z) \nonumber \\
   &\leq L_{\vartheta} \int |\hd_1(w) - \hd_2(w)|dP_{W,0}(w) \le L_{\vartheta}\|\hd_1-\hd_2\|_{L^2(P_{W,0})}, \label{boundf}
\end{align}
where $\|\cdot\|_{L^2(P_{W,0})}$ denotes the $L^2(P_{W,0})$ norm, defined as $\|f\|_{L^2(P_{W,0})}:=[\int f(w)^2 dP_{W,0}(w)]^{1/2}$. 
Because $(\mathcal{H},L^2(P_{W,0}))$ is totally bounded, for every $\delta>0$, there exists a finite $\delta$-cover of $\mathcal{H}$. Now, fix $\epsilon>0$. By \eqref{boundf}, $\|\hd_1-\hd_2\|_{L^2(P_{W,0})} \le \epsilon/L_{\vartheta}$ implies $\Psi_g(\hd_1,\hd_2)\le \epsilon$. 
Let $\mathcal{H}_{\epsilon}$ be a minimal $\epsilon/L_{\vartheta}$-cover of $\mathcal{H}$ with finite cardinality. We then define a $\mathcal{H}\rightarrow \mathcal{H}_{\epsilon}$ map $H_{\epsilon}$ so that, for all $\hd \in \mathcal{H}$, $H_{\epsilon}(\hd)$ is such that $\|\hd-H_{\epsilon}(\hd)\|_{L^2(P_{W,0})}\le \epsilon/L_{\vartheta}$, and note that, by the above $\Psi_g(\hd,H_{\epsilon}(\hd))\le \epsilon$ for all $\hd \in \mathcal{H}$. 
Combining this with the triangle inequality and the fact that $f_{\vartheta,g,z}$ is a nonnegative function shows that
\begin{align*}
       0\le \sup_{\hd \in \mathcal{H}}\int f_{\vartheta, g, z}\circ \hd(w) d\bar{P}_Z(z)&\le \sup_{\hd \in \mathcal{H}}\int f_{\vartheta, g, z}\circ H_\epsilon(\hd)(w) d\bar{P}_Z(z) \\
       &\quad+ \sup_{\hd \in \mathcal{H}}\left|\int \left[f_{\vartheta, g, z}\circ \hd(w)-f_{\vartheta, g, z}\circ H_{\epsilon}(\hd)(w)\right] d\bar{P}_Z(z)\right| \\
        &= \sup_{\hd \in \mathcal{H}}\int f_{\vartheta, g, z}\circ H_{\epsilon}(\hd)(w) d\bar{P}_Z(z) +  \sup_\hd \Psi_g(\hd,H_{\epsilon}(\hd))   \\
       &\le \sup_{\hd \in \mathcal{H}}\int f_{\vartheta, g, z}\circ H_{\epsilon}(\hd)(w) d\bar{P}_Z(z)+ \epsilon \\
       &= \sup_{\hd \in \mathcal{H}_\epsilon}\int f_{\vartheta, g, z}\circ \hd(w) d\bar{P}_Z(z)+ \epsilon.
\end{align*}
By Lemma \ref{lem:qmd}, for each $\hd \in \mathcal{H}_{\epsilon}$, $\int  f_{\vartheta, g, z}\circ \hd(w) d\bar{P}_Z(z) \rightarrow 0 \text{ as } g \rightarrow 0$. 
Thus, the supremum on the right-hand side above, which is a supremum over a finite set, converges to zero. Hence,
\begin{align*}
   \sup_{\hd \in \mathcal{H}} \int f_{\vartheta, g, z}\circ \hd(w) d\bar{P}_Z(z) 
   \le \epsilon.
\end{align*} 
As $\epsilon>0$ was arbitrary, the LHS of the above equation converges to 0. Recalling the definition of $f_{\vartheta,g,z}$ from \eqref{eq:fdef} gives the result.
\end{proof}

\begin{lemma}\label{lem:locasymnorm}
Assume that the conditions of Lemma \ref{lem:qmd} hold. Let $g \in \mathbb{R}^{p+2}$ and $\vartheta_n = \vartheta_0 + g/\sqrt{n}$. Then, for all $\hd \in \mathcal{H}$, $E_{\hd, \vartheta_0}[U_{\hd}(\vartheta_0;Z)]  = 0$,  and $\mathscr{I}_{\hd, \vartheta_0}$ exists. Furthermore, for every $\epsilon>0$ as $n \rightarrow \infty$
\begin{align} \label{eq:lasymnorm}
\sup_{\hd \in \mathcal{H}} P_{\hd, \vartheta_0}\left( \bigg| \Lambda_{\hd}(\vartheta_n, \vartheta_0; D_n)  -g^{\rm{T}}\frac{1}{\sqrt{n}}U_{\hd}(\vartheta_0; D_n) + \frac{1}{2}g^{\rm{T}} \mathscr{I}_{\hd, \vartheta_0} g \bigg| > \epsilon \right) \rightarrow 0  
\end{align}
\end{lemma}

\begin{proof}
Let $\hd \in \mathcal{H}$. By assumption, $P_{\hd, \vartheta_0}$ is quadratic mean differentiable. Then, by Theorem 7.2 in \cite{van2000asymptotic} $E_{\hd, \vartheta_0}[U(\vartheta_0;Z)] = 0$ and $\mathscr{I}_{\hd, \vartheta_0}$ exist.
 
Now we show equation \ref{eq:lasymnorm}. First, we introduce some notation: let $\bar{p}_{\hd, n}(z) := \bar{p}_{\hd, \vartheta_0 + g/\sqrt{n}}(z)$ and $s_{\hd}(z) := g^{\rm{T}} U_{\hd}(\vartheta_0; z)$. Note that, since we use $\bar{P}_Z$ as a dominating measure, $\bar{p}_{\hd, \vartheta_0}(z) = 1$. Now, define the random variable $T_{\hd, ni} := 2[ \sqrt{\bar{p}_{\hd, n}(Z_i)} -1]$ and note that 
\begin{align*}
    &\sup_{\hd \in \mathcal{H}} E\left[\left\{\sum_{i=1}^n T_{\hd, ni} - n^{-1/2}\sum_{i=1}^n s_\hd(Z_i) + \frac{1}{4}Ps_\hd^2\right\}^2\right] \\ 
    &\leq \sup_{\hd \in \mathcal{H}} \text{var}\left(\sum_{i=1}^n T_{\hd, ni} - n^{-1/2}\sum_{i=1}^n s_\hd(Z_i) + \frac{1}{4}Ps_\hd^2 \right)\\
    &\quad+\sup_{\hd \in \mathcal{H}}E\left[\sum_{i=1}^n T_{\hd, ni} - n^{-1/2}\sum_{i=1}^n s_\hd(Z_i) + \frac{1}{4}Ps_\hd^2\right]^2  
\end{align*}
We show these two quantities converge to 0. First, by properties of variances and Lemma \ref{lem:uqmd}:
\begin{align}
  0 &\le \sup_{\hd \in \mathcal{H}}   \text{var}\left(\sum_{i=1}^n T_{\hd, ni} - n^{-1/2}\sum_{i=1}^n s_\hd(Z_i) + \frac{1}{4}Ps_\hd^2 \right) \nonumber \\
  &= \sup_{\hd \in \mathcal{H}}   \text{var}\left(\sum_{i=1}^n T_{\hd, ni} - n^{-1/2}\sum_{i=1}^n s_\hd(Z_i)\right) = \sup_{\hd \in \mathcal{H}} \sum_{i=1}^n\text{var}\left( T_{\hd, ni} - n^{-1/2}s_\hd(Z_i)\right) \nonumber \\
  &= \sup_{\hd \in \mathcal{H}}   \text{var}\left( n^{1/2}T_{\hd, n1} - s_\hd(Z_1)\right)
  \leq \sup_{\hd \in \mathcal{H}}  E\left[\left(n^{1/2} T_{\hd, n1} - s_\hd(Z_1)\right)^2\right] \nonumber \\
  &= 2 \sup_{\hd \in \mathcal{H}} \int n \left(\sqrt{\bar{p}_{\hd, n}(z)} -\sqrt{\bar{p}_{\hd, \vartheta_0}(z)} - \frac{1}{2} \frac{g^{\rm{T}}}{\sqrt{n}} U_{\hd}(\vartheta_0; z) \sqrt{\bar{p}_{\hd, \vartheta_0}(z)} \right)^2 d\bar{P}_Z(z)\rightarrow 0. \label{eq:1localsymn}
\end{align}
Next,
\begin{align*}
    \sup_{\hd \in \mathcal{H}} E\left[\sum_{i=1}^n T_{\hd, ni} - n^{-1/2}\sum_{i=1}^n s_\hd(Z_i) + \frac{1}{4}Ps_\hd^2\right]^2 &=  \sup_{\hd \in \mathcal{H}}E\left[\sum_{i=1}^n T_{\hd, ni} + \frac{1}{4}Ps_\hd^2\right]^2 \\
    &= \sup_{\hd \in \mathcal{H}} \left(2n \int[ \sqrt{\bar{p}_{n,\hd}(z)}-1 ]d\bar{P}_{Z}(z) + \frac{1}{4}Ps_\hd^2\right)^2\\
    &= \sup_{\hd \in \mathcal{H}} \left(-n\int \left[ \sqrt{\bar{p}_{\hd, n}(z)}-\sqrt{\bar{p}_{\hd, \vartheta_0}(z)} \right]^2 d\bar{P}_{Z}(z) + \frac{1}{4}Ps_\hd^2\right)^2,
\end{align*}
where for a fixed $\hd$ we have that
\begin{align}
    &\int \left[ \sqrt{\bar{p}_{\hd, n}(z)}-\sqrt{\bar{p}_{\hd, \vartheta_0}(z)} \right]^2 d\bar{P}_{Z}(z) \nonumber \\
    &= \int \left[ \sqrt{\bar{p}_{\hd, n}(z)}-\sqrt{\bar{p}_{\hd, \vartheta_0}(z)} \right]^2 d\bar{P}_{Z}(z) - \int \left[ \sqrt{\bar{p}_{\hd, n}(z)}-\sqrt{\bar{p}_{\hd, \vartheta_0}(z)} - n^{-1/2}\frac{1}{2} s_\hd(z) \right]^2 d\bar{P}_{Z}(z) + o(n^{-1}) \nonumber \\
    &= -n^{-1} \frac{1}{4}Ps_\hd^2+2\int \left[ \sqrt{\bar{p}_{\hd, n}(z)}-\sqrt{\bar{p}_{\hd, \vartheta_0}(z)}\right]\left[ n^{-1/2}\frac{1}{2} s_\hd(z) \right] d\bar{P}_{Z}(z) + o(n^{-1}) \nonumber \\
    &= -n^{-1} \frac{1}{4}Ps_\hd^2 + n^{-1} \frac{1}{2} \int s_\hd(z)^2 d\bar{P}_{Z}(z)  \nonumber \\
    &\quad + 2\int \left[ \sqrt{\bar{p}_{\hd, n}(z)}-\sqrt{\bar{p}_{\hd, \vartheta_0}(z)} - n^{-1/2} \frac{1}{2} s_\hd(z)\right]\left[ n^{-1/2}\frac{1}{2} s_\hd(z) \right] d\bar{P}_{Z}(z) + o(n^{-1}) \label{eq:CSUQMD} \\
    &= n^{-1} \frac{1}{4} Ps_\hd^2 + o(n^{-1}), \nonumber
\end{align}
where the integral in \eqref{eq:CSUQMD} is $o(n^{-1})$ by  Cauchy-Schwarz and Lemma \ref{lem:uqmd}. 
By the UQMD property of $\{P_{\hd,\vartheta}: \vartheta \in \mathbb{R}^{p+2}, \hd \in \mathcal{H}\}$, the $o(n^{-1})$ terms above are uniform over the choice of $\hd \in \mathcal{H}$. Therefore,
\begin{align}
\label{eq:2localsymn}
\sup_{\hd \in \mathcal{H}} E\left[\sum_{i=1}^n T_{\hd, ni} - n^{-1/2}\sum_{i=1}^n s_\hd(Z_i) + \frac{1}{4}Ps_\hd^2\right]^2 \rightarrow 0. 
\end{align}
By \eqref{eq:1localsymn}, \eqref{eq:2localsymn}, and Markov's inequality, we have that, for every $\epsilon >0$,
\begin{align}
    &\sup_{\hd \in \mathcal{H}} P\left( \bigg| \sum_{i=1}^N T_{\hd, ni} - n^{-1/2}\sum_{i=1}^n s_\hd(Z_i) + \frac{1}{4}Ps_\hd^2 \bigg| > \epsilon \right) \nonumber \\
    &\le \epsilon^{-2} \sup_{\hd \in \mathcal{H}} E\left[\left\{\sum_{i=1}^n T_{\hd, ni} - n^{-1/2}\sum_{i=1}^n s_\hd(Z_i) + \frac{1}{4}Ps_\hd^2\right\}^2\right] \overset{n\to\infty}{\longrightarrow} 0. \label{eq:TniUnif}
\end{align}
Next, we express $\Lambda_{\hd}(\vartheta_n, \vartheta_0; d_n)$ through a Taylor expansion of the logarithm. Define the $(-2,\infty)\rightarrow\infty$ function
\begin{align*}
    R(x)&= \begin{cases}
   \frac{\log\left(1+\frac{x}{2}\right) - x/2 + x^2/8}{x^2/4},&\mbox{ if }x\not=0, \\
    0,&\mbox{ otherwise.}
    \end{cases}
\end{align*}
Note that $\log(1+x)=x - \frac{1}{2}x^2 + x^2 R(2x)$. Moreover, by a Taylor expansion of $x\mapsto \log(1+x)$ about $0$, it can be shown that $R(x)\rightarrow 0$ as $x\rightarrow 0$. Then, for a given $\hd \in \mathcal{H}$ we have that 
\begin{align}
\label{eq:4loglktaylor}
   \Lambda_{\hd}(\vartheta_n, \vartheta_0; D_n) = 2 \sum_{i = 1}^n \log \left(1 + \frac{1}{2} T_{\hd, ni}\right)  = \sum_{i = 1}^n T_{\hd, ni} - \frac{1}{4} \sum_{i = 1}^n T_{\hd, ni}^2 + \frac{1}{2} \sum_{i = 1}^n T_{\hd, ni}^2 R(T_{\hd, ni}).
\end{align}
We have already shown uniform convergence of the first term of the right hand side of the above, in the sense described in \eqref{eq:TniUnif}. We now show that $\sum_{i = 1}^n T_{\hd, ni}^2 \rightarrow Ps_\hd ^2$ in an analogous uniform sense. To do this, we will show that both terms of the right hand side in the equation below converge to 0,
\begin{align}
\label{eq:5localsymn}
   \sup_{\hd \in \mathcal{H}} E\left[\left\{ \sum_{i = 1}^n T_{\hd, ni}^2 - Ps_\hd^2\right\} ^2\right] \leq  \sup_{\hd \in \mathcal{H}} \text{var}\left\{ \sum_{i = 1}^n T_{\hd, ni}^2 - Ps_\hd^2 \right\} + \sup_{\hd \in \mathcal{H}}E\left[\sum_{i = 1}^n T_{\hd, ni}^2 - Ps_\hd^2\right]^2.
\end{align}
Let $T_{\hd, n}(z):=2[\sqrt{\bar{p}_{\hd, n}(z)}-1]$, and note that $T_{\hd, ni}=T_{\hd, n}(Z_i)$. 
We begin by showing that the variance term converges to 0 uniformly. To do this, we first upper bound it as follows:
\begin{align*}
       \sup_{\hd \in \mathcal{H}}\text{var}\left\{ \sum_{i = 1}^n T_{\hd, ni}^2 - Ps_\hd^2 \right\} &=   \sup_{\hd \in \mathcal{H}}\text{var}\left\{ \sum_{i = 1}^n T_{\hd, ni}^2\right\} = \sup_{\hd \in \mathcal{H}}\sum_{i = 1}^n\text{var}\left\{ T_{\hd, ni}^2\right\} \leq  n \sup_{\hd \in \mathcal{H}}\int T_{\hd, n}(z)^4 dP_{Z}(z)\\
    &\le n\left(\sup_{\hd \in \mathcal{H},z\in\mathcal{Z}} T_{\hd, n}(z)^2\right)\sup_{\hd \in \mathcal{H}}\int T_{\hd, n}(z)^2dP_{Z}(z).
\end{align*}
Now, by Lemma \ref{lem:supT}, $\sup_{\hd \in \mathcal{H},z\in\mathcal{Z}} T_{\hd, n}(z)^2\ \rightarrow 0$ as $n\rightarrow\infty$, and because of the arguments in \eqref{eq:CSUQMD} from Lemma \ref{lem:locasymnorm}, we have that $\sup_{\hd \in \mathcal{H}}\int T_{\hd, n}(z)^2dP_{Z}(z) = O(n^{-1})$. Hence, $\sup_{\hd \in \mathcal{H}}\text{var}\left\{ \sum_{i = 1}^n T_{\hd, ni}^2 - Ps_\hd^2 \right\} \rightarrow 0.$
For the second term in \eqref{eq:5localsymn}, note that
\begin{align*}
 0 \leq \sup_{\hd \in \mathcal{H}}E\left[\sum_{i = 1}^n T_{\hd, ni}^2 - Ps_\hd^2\right]^2 &=  \sup_{\hd \in \mathcal{H}} \left(nE[T_{\hd, n1}^2] - Ps_\hd^2\right)^2. 
\end{align*}
By the arguments from \eqref{eq:CSUQMD} in Lemma \ref{lem:supT},
\begin{align*}
    \left|nE[ T_{\hd, ni}^2] - Ps_\hd^2\right| &= \left|n 4\int \left[ \sqrt{\bar{p}_{\hd, n}(z)}-\sqrt{\bar{p}_{\hd, \vartheta_0}(z)} \right]^2 d\bar{P}_{Z}(z)  - Ps_\hd^2\right| = \\
    &=  \left|n4 [n^{-1} \frac{1}{4} Ps_\hd^2 + o(n^{-1})]  - Ps_\hd^2\right|= o(n^{-1}) \rightarrow 0,
\end{align*}
uniformly over the choice of $\hd$. Hence, $\sup_{\hd \in \mathcal{H}} \left(nE[ T_{\hd, n1}^2] - Ps_\hd^2\right)^2 \rightarrow 0$. Combining the two preceding arguments we have that 
\begin{align*}
    \sup_{\hd \in \mathcal{H}}\text{var}\left\{ \sum_{i = 1}^n T_{\hd, ni}^2 - Ps_\hd^2 \right\} \rightarrow 0.
\end{align*}
Hence, by Markov's inequality and the aforementioned arguments, we have that for every $\epsilon > 0$
\begin{align}
\label{eq:6localsymn}
\sup_{\hd \in \mathcal{H}} P_{\hd, \vartheta_0}\left(\left| \sum_{i = 1}^n T_{\hd, ni}^2 - Ps_\hd^2\right| > \epsilon \right) \overset{n\to\infty}{\longrightarrow} 0.
\end{align}
Next, we show that for any $\epsilon >0$
\begin{align*}
\sup_{\hd \in \mathcal{H}}P_{\hd, \vartheta_0}\left(\left|\sum_{i = 1}^n T_{\hd, ni}^2 R(T_{\hd, ni})\right| > \epsilon\right)&\overset{n\to\infty}{\longrightarrow} 0.
\end{align*}
First,
\begin{align}
\label{eq:7localsymn}
    E\left[\left|\sum_{i = 1}^n T_{\hd, ni}^2 R(T_{\hd, ni})\right|\right] \leq  \sum_{i = 1}^n E\left[ \left|T_{\hd, ni}^2 R(T_{\hd, ni})\right|\right] \leq n \sup_{z \in \mathcal{Z}, \hd \in \mathcal{H}}R(T_{\hd, n}(z)) E\left[T_{\hd, ni}^2\right].
\end{align}
We now argue that $\sup_{z \in \mathcal{Z}, \hd \in \mathcal{H}}R(T_{\hd, n}(z))\rightarrow 0$ as $n\rightarrow\infty$. We first note that $R$ is continuous and nondecreasing. Also, because $\mathcal{W}$ and $\mathcal{H}$ are bounded, $\sup_{z \in \mathcal{Z}, \hd \in \mathcal{H}}R(T_{\hd, n}(z))<\infty$. Hence,
\begin{align*}
    \sup_{z \in \mathcal{Z}, \hd \in \mathcal{H}}R(T_{\hd, n}(z)) = R\left(\sup_{z \in \mathcal{Z}, \hd \in \mathcal{H}}T_{\hd, n}(z)\right). 
\end{align*}
By Lemma~\ref{lem:supT}, $\sup_{\hd \in \mathcal{H},z\in\mathcal{Z}}T_{\hd, n}(z)^2 \rightarrow 0$, which implies that $\sup_{\hd \in \mathcal{H},z\in\mathcal{Z}}T_{\hd, n}(z) \rightarrow 0$. 
By the above, this implies that $ \sup_{z \in \mathcal{Z}, \hd \in \mathcal{H}}R(T_{\hd, n}(z)) \rightarrow 0$ as $n \rightarrow \infty$. Next, because the arguments outlined in Lemma \ref{lem:locasymnorm} and \eqref{eq:CSUQMD}, we have that $\sup_{\hd \in \mathcal{H}}E[T_{\hd, ni}^2] = O(n^{-1})$. Combining the preceding arguments with \eqref{eq:7localsymn} and applying Markov's inequality, we have that, for any $\epsilon>0$, 
\begin{align}
\label{eq:8localsymn}
\sup_{\hd \in \mathcal{H}}P_{\hd, \vartheta_0}\left(\left|\sum_{i = 1}^n T_{\hd, ni}^2 R(T_{\hd, ni})\right| > \epsilon\right)&\le \epsilon^{-1}\sup_{\hd \in \mathcal{H}}E\left|\sum_{i = 1}^n T_{\hd, ni}^2 R(T_{\hd, ni})\right|\overset{n\to\infty}{\longrightarrow} 0.
\end{align}
Combining the identity in \eqref{eq:4loglktaylor} with the results shown in \eqref{eq:TniUnif}, \eqref{eq:6localsymn}, and \eqref{eq:8localsymn}, and by the triangle inequality, we have that for every $\epsilon > 0$,
\begin{align*}
\sup_{\hd \in \mathcal{H}} P_{\hd, \vartheta_0}\left( \bigg| \Lambda_{\hd}(\vartheta_n, \vartheta_0; D_n)  -\frac{g^{\rm{T}}}{\sqrt{n}}U_{\hd}(\vartheta_0;D_n) + \frac{1}{2}g^{\rm{T}} \mathscr{I}_{\hd, \vartheta_0} g \bigg| > \epsilon \right) \overset{n\to\infty}{\longrightarrow} 0.
\end{align*}
\end{proof}

\begin{lemma}
\label{lem:supT}
Fix $\vartheta, g\in \mathbb{R}^{p+2}$. If the conditions of Lemma \ref{lem:locasymnorm} hold, then 
\begin{align*}
   \sup_{\hd \in \mathcal{H},z\in\mathcal{Z}}T_{\hd, n}(z)^2 \rightarrow 0. 
\end{align*}
\end{lemma}
\begin{proof}
Fix $z \in \mathcal{Z}$ and $\hd \in \mathcal{H}$ and define the $\mathbb{R}^{p+2} \rightarrow \mathbb{R}$ function $v_{\hd, z}(\vartheta):= \bar{p}_{\hd,\vartheta}(z)^{1/2}$. We show that $v$ is Lipschitz continuous in $\vartheta$ for all $\hd \in \mathcal{H}$ and $z\in\mathcal{Z}$. We do so by showing that its derivative is uniformly upper bounded entrywise. First, 
\begin{align*}
\left[\pi(w)^a (1-\pi(w))^{1-a}\right]\frac{\partial}{\partial \vartheta} \bar{p}_{\hd,\vartheta}(z) = \frac{\partial}{\partial \vartheta} \left\{\pi_{\hd,\vartheta}(w)^a(1-\pi_{\hd,\vartheta}(w))^{1-a}\right\} = 
\begin{cases}
\frac{\partial}{\partial \vartheta}\pi_{\hd,\vartheta}(w) \text{ if } a= 1\\
-\frac{\partial}{\partial \vartheta}\pi_{\hd,\vartheta}(w) \text{ if } a= 0,
\end{cases}
\end{align*}
and because $\frac{\partial}{\partial \vartheta}\pi_{\hd,\vartheta}(w)= r_{\hd}(w) \rho(\vartheta^{\rm{T}} w)$, we have that,
\begin{align*}
   2 &\left|\left[\pi(w)^a (1-\pi(w))^{1-a}\right] ^{1/2}\frac{\partial}{\partial \vartheta}  \bar{p}_{\hd,\vartheta}(z)^{1/2} \right| = \left|\left[\pi(w)^a (1-\pi(w))^{1-a}\right] ^{1/2} \frac{\frac{\partial}{\partial \vartheta} \bar{p}_{\hd,\vartheta}(z)}{\bar{p}_{\hd,\vartheta}(z)^{1/2}} \right|  \\ 
   &= \left|\frac{r_{\hd}(w) \rho(\vartheta^{\rm{T}} r_{\hd}(w))}{\left\{\pi_{\hd,\vartheta}(w)^a [1-\pi_{\hd,\vartheta}(w)]^{1-a}\right\}^{1/2}}\right| =  \frac{\left|r_{\hd}(w) \rho(\vartheta^{\rm{T}} r_{\hd}(w))\right|}{\left\{\pi_{\hd,\vartheta}(w)^a [1-\pi_{\hd,\vartheta}(w)]^{1-a}\right\}^{1/2}}.
\end{align*}
By the definition of $\rho$, we know that $\rho(\vartheta^{\rm{T}} r_{\hd}(w))<1$ regardless of the value of $w$ and $\hd$. Because of the boundedness of $\mathcal{W}$ and $\mathcal{H}$, there exists a finite vector $C_1$ and a finite constant $C_2$ such that $\sup_{\hd,w}|r_{\hd}(w)\rho(\vartheta^{\rm{T}} r_{\hd}(w))| < C_1$ entrywise and $\sup_{\hd, z} \left\{\pi_{\hd,\vartheta}(w)^a [1-\pi_{\hd,\vartheta}(w)]^{1-a}\right\}^{-1/2} < C_2$. Finally, by the positivity assumption, there exists a $\delta>0$ such that $\pi(w)> \delta$ for all $w \in \mathcal{W}$. 
Hence, $\sup_{\hd, z}\left|\frac{\partial}{\partial \vartheta}  \bar{p}_{\hd,\vartheta}(z)^{1/2}\right| \leq \delta^{-1/2} C_1 C_2^{1/2}$, which implies that $v_{\hd, z}$ is Lipschitz for all $z \in \mathcal{Z}$ and $\hd \in \mathcal{H}$ with Lipschitz constant $\delta^{-1/2} C_1 C_2^{1/2}$. Therefore,
\begin{align*}
    0 \leq \sup_{\hd,z} T_{\hd, n}(z)^2 &=  \sup_{\hd,z}\,[\bar{p}_{\hd, n}(z)^{1/2} - \bar{p}_{\hd, n}(z)^{1/2}]^2 = \sup_{\hd,z}\, [v_{\hd, z}(\vartheta + g/\sqrt{n}) - v_{\hd, z}(\vartheta)]^2  \\
    &\leq \delta^{-1/2} C_1 C_2^{1/2}\| (\vartheta + g/\sqrt{n}) - \vartheta \|^2 = \delta^{-1/2} C_1 C_2^{1/2}\|g\|^2/\sqrt{n} \overset{n\to\infty}{\longrightarrow} 0,
\end{align*}
which is the desired result.
\end{proof}
Before stating the next Lemma, we state the definition of contiguity of the measures $P_{h_{0n},\vartheta_n}$  and  $P_{h_0, \vartheta_0}$. We say that the sequence of measures $P_{h_{0n},\vartheta_n}$ is contiguous to $P_{h_0, \vartheta_0}$ if $P_{h_0, \vartheta_0}(A_n ) \rightarrow 0$ implies that $P_{h_{0n},\vartheta_n}(A_n) \rightarrow 0$ for every sequence of measurable sets $A_n$. Recall that we treat each $h_{\cdot n}$ as an element of the deterministic sequence of functions $\{h_{0n}\}_{n = 1}^{\infty}$. If $P_{h_{0n},\vartheta_n}$ is contiguous to $P_{h_0, \vartheta_0}$ and $P_{h_{0n},\vartheta_n}(A_n) \rightarrow 0$ implies that $P_{h_0, \vartheta_0}(A_n) \rightarrow 0$ for every sequence of measurable sets $A_n$, then we say that  the sequence of measures $P_{h_{0n},\vartheta_n}$ is mutually contiguous to $P_{h_0, \vartheta_0}$. 

We also introduce additional notation. First, let $\|\cdot\|_{L^1(P_0)}$ denote the $L^1(P_0)$ norm and $\|\cdot\|$ denote the $L^2(P)$ norm. Moreover, for a given measurable function $f$, let $\|f\|_{L^1(P_{h_0, \vartheta_0})}:=E_{P_{h_0, \vartheta_0}}|f(W,A,Y)|$. Finally, for a given matrix $M$, we let $(M)_{l,k}$ denote the value in its $k$-th column and $l$-th row.

\begin{lemma}
\label{lem:lasymnseqh}
 Fix $g \in \mathbb{R}^{p+2}$ and let $\vartheta_n = \vartheta_0 + g/\sqrt{n}$. Assume that the conditions of Lemma \ref{lem:qmd} hold and $ \|h_{0n}-h_0\|_{L^1(P_{h_0, \vartheta_0})}\rightarrow 0$. 
Then, for every $\epsilon > 0$,
\begin{equation}
\label{eq:1lasymnseqh}
P_{h_{0n},\vartheta_n}\left(\left|\Lambda_{h_{0n}}(\vartheta_0,\vartheta_n ; D_n) + g^{\rm{T}}\frac{1}{\sqrt{n}}U_{h_{0n}}(\vartheta_n;D_n) + \frac{1}{2}g^{\rm{T}} \mathscr{I}_{h_0, \vartheta_0} g \right| > \epsilon \right)  \overset{n\to\infty}{\longrightarrow} 0.
\end{equation}
\end{lemma}
\begin{proof}

First, note that $P_{h_0, \vartheta_0} = P_{h_{0n},  \vartheta_0}$ and that, by Lemma \ref{lem:qmd}, $\{P_{\hd,\vartheta}: \vartheta \in \mathbb{R}^{p+2}, \hd \in \mathcal{H}\}$ is uniformly differentiable in quadratic mean. By the triangle inequality, 

\begin{align}
\label{eq:2lasymnseqh}
&P_{h_0, \vartheta_0}\left(\left|\Lambda_{h_{0n}}(\vartheta_0,\vartheta_n ; D_n) + g^{\rm{T}}\frac{1}{\sqrt{n}}U_{h_{0n}}(\vartheta_n;D_n) + \frac{1}{2} g^{\rm{T}} \mathscr{I}_{h_0, \vartheta_0} g \right| > \epsilon \right)  \nonumber  \\ \nonumber
=&P_{h_{0n}, \vartheta_0}\left(\bigg|\Lambda_{h_{0n}}(\vartheta_n,\vartheta_0; D_n) -g^{\rm{T}}\frac{1}{\sqrt{n}}U_{h_{0n}}(\vartheta_n;D_n) - \frac{1}{2}g^{\rm{T}} \mathscr{I}_{h_0, \vartheta_0} g \bigg| > \epsilon \right) \\ \nonumber
\leq &P_{h_{0n}, \vartheta_0}\left( \bigg|\Lambda_{h_{0n}}(\vartheta_n,\vartheta_0; D_n) -g^{\rm{T}}\frac{1}{\sqrt{n}}U_{h_{0n}}(\vartheta_0;D_n) + \frac{1}{2}g^{\rm{T}} \mathscr{I}_{h_{0n}, \vartheta_0} g \bigg| > \epsilon \right) \\ \nonumber
+&P_{h_0, \vartheta_0}\left( \bigg|g^{\rm{T}}\frac{1}{\sqrt{n}}U_{h_{0n}}(\vartheta_0;D_n) -g^{\rm{T}}\frac{1}{\sqrt{n}}U_{h_{0n}}(\vartheta_n;D_n) - 
\frac{1}{2}g^{\rm{T}} \mathscr{I}_{h_{0n}, \vartheta_0} g - 
\frac{1}{2}g^{\rm{T}} \mathscr{I}_{h_0, \vartheta_0} g \bigg| > \epsilon \right) \\ \nonumber
\leq &P_{h_{0n}, \vartheta_0}\left( \bigg|\Lambda_{h_{0n}}(\vartheta_n,\vartheta_0; D_n) -g^{\rm{T}}\frac{1}{\sqrt{n}}U_{h_{0n}}(\vartheta_0;D_n) + \frac{1}{2}g^{\rm{T}} \mathscr{I}_{h_{0n}, \vartheta_0} g \bigg| > \epsilon \right) \\ \nonumber
+& P_{h_0, \vartheta_0}\left( \left|g^{\rm{T}}\frac{1}{\sqrt{n}}U_{h_{0n}}(\vartheta_0;D_n) -g^{\rm{T}}\frac{1}{\sqrt{n}}U_{h_{0n}}(\vartheta_n;D_n) - 
g^{\rm{T}} \mathscr{I}_{h_{0n}, \vartheta_0}g
\right| > \epsilon \right) \\ \nonumber
+&I\left\{\frac{1}{2}\left|g^{\rm{T}} \mathscr{I}_{h_{0n}, \vartheta_0}g - g^{\rm{T}} \mathscr{I}_{h_0, \vartheta_0} g \right| > \epsilon \right\} \\ \nonumber
\leq \sup_{\hd \in \mathcal{H}} &P_{\hd, \vartheta_0}\left( \bigg|\Lambda_{\hd}(\vartheta_n,\vartheta_0; D_n) -g^{\rm{T}}\frac{1}{\sqrt{n}}U_{\hd}(\vartheta_0;D_n) + \frac{1}{2}g^{\rm{T}} \mathscr{I}_{\hd, \vartheta_0} g \bigg| > \epsilon \right) \\ \nonumber
+\sup_{\hd \in \mathcal{H}} & P_{\hd, \vartheta_0}\left( \left|g^{\rm{T}}\frac{1}{\sqrt{n}}U_{\hd}(\vartheta_0;D_n) -g^{\rm{T}}\frac{1}{\sqrt{n}}U_{\hd}(\vartheta_n;D_n) - 
g^{\rm{T}} \mathscr{I}_{h_0, \vartheta_0}g
\right| > \epsilon \right) \\ 
+ &I\left\{\frac{1}{2}\left|g^{\rm{T}} \mathscr{I}_{h_{0n}, \vartheta_0}g - g^{\rm{T}} \mathscr{I}_{h_0, \vartheta_0} g \right| > \epsilon \right\}.
\end{align}

By Lemmas \ref{lem:lanscore} and \ref{lem:locasymnorm}, the first two terms in the last inequality converge to 0 as $n \rightarrow \infty$. For the third term, one can use that $\vartheta_{0_{p+2}}=0$, and therefore $\pi_{h_{0n}, \vartheta_0}=\pi_{h_0, \vartheta_0}$, to show that
\begin{equation*}
\left[\mathscr{I}_{h_{0n}, \vartheta_0} - \mathscr{I}_{h_0, \vartheta_0}\right]_{l,k}=\begin{cases}
0 & l\leq p,\;k\leq p\\
E\left\{W_{k}\left[h_{0n}(W)-h_0(W)\right]\pi_{h_0, \vartheta_0}(W)(1-\pi_{h_0, \vartheta_0}(W))\right\} & l=p+2,\; k\leq p+1\\
E\left\{W_{l}\left[h_{0n}(W)-h_0(W)\right]\pi_{h_0, \vartheta_0}(W)(1-\pi_{h_0, \vartheta_0}(W))\right\}  & k=p+2,\; l\leq p+1\\
E\left\{[h_{0n}^{2}(W)-h_0^2(W)]\pi_{h_0, \vartheta_0}(W)(1-\pi_{h_0, \vartheta_0}(W))\right\} & k=p+2,\; l=p+2,
\end{cases}
\end{equation*}
where $W_k$ is the $k-th$ covariate. 
By similar arguments to those in Lemma \ref{lem:lanscore}, and because $\pi_{h_0, \vartheta_0}(w)$ is uniformly bounded, there exists a finite constant $M$ that does not depend on $n$ and is such that 
\begin{equation*}
\frac{1}{2}\left|g^{\rm{T}} \mathscr{I}_{h_{0n}, \vartheta_0}g - g^{\rm{T}} \mathscr{I}_{h_0, \vartheta_0} g \right| \leq M \left|E\left[h_{0n}(W) - h_0(W)\right]\right| \leq M \|h_{0n}-h_0\|_{L^1(P_{h_0, \vartheta_0})}. 
\end{equation*}
By the above display and because  $\|h_{0n}-h_0\|_{L^1(P_{h_0, \vartheta_0})}\rightarrow 0$, we have that 
\begin{align*}
I\left\{\frac{1}{2}\left|g^{\rm{T}} \mathscr{I}_{h_{0n}, \vartheta_0}g - g^{\rm{T}} \mathscr{I}_{h_0, \vartheta_0} g \right| > \epsilon \right\} \overset{n \rightarrow \infty}{\longrightarrow}0.
\end{align*}
Therefore, for every $\epsilon > 0$,
\begin{equation*}
    P_{h_0, \vartheta_0}\left(\left|\Lambda_{h_{0n}}(\vartheta_0,\vartheta_n ; D_n) + g^{\rm{T}}\frac{1}{\sqrt{n}}U_{h_{0n}}(\vartheta_n;D_n) + \frac{1}{2} g^{\rm{T}} \mathscr{I}_{h_0, \vartheta_0} g \right| > \epsilon \right) \overset{n\to\infty}{\longrightarrow} 0.
\end{equation*}
Because of Lemma \ref{lem:contiguity} the sequence of measures $(P_{h_{0n},\vartheta_n}^n)_{n=1}^\infty$ and $(P_{h_0, \vartheta_0}^n)_{n=1}^\infty$ are mutually contiguous. The latter implies that 
\begin{equation*}
    P_{h_{0n},\vartheta_n}\left(\left|\Lambda_{h_{0n}}(\vartheta_0,\vartheta_n ; D_n) + g^{\rm{T}}\frac{1}{\sqrt{n}}U_{h_{0n}}(\vartheta_n;D_n) + \frac{1}{2} g^{\rm{T}} \mathscr{I}_{h_0, \vartheta_0} g \right| > \epsilon \right) \overset{n\to\infty}{\longrightarrow} 0,
\end{equation*}
which is the desired result.
\end{proof}   
\begin{lemma}
\label{lem:lanscore}
Under the conditions of Lemma \ref{lem:lasymnseqh}, for every $\epsilon >0$,
\begin{align*}
\sup_{\hd \in \mathcal{H}} P_{h_0, \vartheta_0}\left( \bigg|-g^{\rm{T}}\frac{1}{\sqrt{n}}U_{\hd}(\vartheta_n;D_n) + g^{\rm{T}}\frac{1}{\sqrt{n}} U_{\hd}(\vartheta_0;D_n) - g^{\rm{T}}\mathscr{I}_{\hd, \vartheta_0}g \bigg| > \epsilon \right) \rightarrow 0
\end{align*}
\end{lemma}
\begin{proof}
Fix $\hd \in \mathcal{H}$ and $g \in \mathbb{R}^{p+2}$. A Taylor expansion of the function $\vartheta \mapsto \frac{g^{\rm{T}}}{\sqrt{n}}U_{\hd}(\vartheta; d_n)$ around $\vartheta_0$ yields
\begin{align*}
    \frac{g^{\rm{T}}}{\sqrt{n}}U_{\hd}(\vartheta_n, d_n) &=   \frac{g^{\rm{T}}}{\sqrt{n}}U_{\hd}(\vartheta_0; d_n) +   \frac{g^{\rm{T}}}{\sqrt{n}}\left(K_{\hd}(\vartheta;d_n)\bigg|_{\vartheta = \vartheta_0}\right)(\vartheta_n - \vartheta_0) + R_\hd(\vartheta_n, \vartheta_0,d_n) \\
    &= \frac{g^{\rm{T}}}{\sqrt{n}}U_{\hd}(\vartheta_0;d_n) + \frac{g^{\rm{T}}}{n}\left(K_{\hd}(\vartheta;d_n)\bigg|_{\vartheta = \vartheta_0} \right)g+ R_\hd(\vartheta_n, \vartheta_0,d_n),
\end{align*}
where $R_\hd(\vartheta_n, \vartheta_0,d_n)$ is the remainder term of the Taylor expansion, and $K_{\hd}(\vartheta;d_n)$ is the Jacobian of $U_{\hd}(\vartheta; d_n)$. First, we will show that
\begin{align*}
  \sup_{\hd \in \mathcal{H}} P_{\hd, \vartheta_0}\left( \left| g^{\rm{T}}\left[ \frac{1}{n}K_{\hd}(\vartheta;D_n)\bigg|_{\vartheta = \vartheta_0} - \mathscr{I}_{\hd, \vartheta_0} \right]g  \right| > \epsilon \right)  \overset{n\to\infty}{\longrightarrow} 0.
\end{align*}
For a given $d_n$,
\begin{align*}
    K_{\hd}(\vartheta;d_n)\bigg|_{\vartheta = \vartheta_0} = \sum_{i = 1}^n \left\{ r_\hd(w_i)r_\hd(w_i)^{\rm{T}}[\pi_{\hd, \vartheta_0}(w_i)(1-\pi_{\hd, \vartheta_0}(w_i))]\right\}.
\end{align*}
Now, for a fixed $w \in \mathcal{W}$, define the $\mathcal{H} \rightarrow \mathbb{R}$ function
\begin{equation*}
f_{w,g}(\hd) = g^{\rm{T}}\left[r_{\hd}(w)r_{\hd}(w)^{\rm{T}}[\pi_{\hd, \vartheta_0}(w)(1-\pi_{\hd, \vartheta_0}(w))] - \mathscr{I}_{\hd, \vartheta_0}\right]g.
\end{equation*}
We will show that there exists a finite constant $Q$ such that for $\hd_1, \hd_2 \in \mathcal{H}$
\begin{align*}
    E\left[|f_{w,g}(\hd_1) - f_{w,g}(\hd_2)|\right] \leq 2 Q \int |\hd_1(w) - \hd_2(w)|d\bar{P}_W(w).
\end{align*}
Given a matrix $M$, let $(M)_{k,l}$ its value in row $k$ and column $l$, and note that 
\begin{equation*}
\left[r_{\hd_1}(w) r_{\hd_1}(w)^{\rm{T}}-r_{\hd_2}(w) r_{\hd_2}(w)^{\rm{T}}\right]_{l,k}=\begin{cases}
0 & l \leq p,\;k\leq p\\
w_{k}\left[\hd_1(w)-\hd_2(w)\right] & l=p+2,\;1\leq k\leq p+1\\
w_{l}\left[\hd_1(w)-\hd_2(w)\right] & k=p+2,\;1\leq l\leq p+1\\
\hd_1(w)^2-\hd_2(w)^2 & k=p+2,\;l=p+2,
\end{cases}
\end{equation*}
where $w_j$ is the $j$-th covariate. Given the above, and because $g \in \mathbb{R}^{p+2}$, $W$ has bounded support, and $\mathcal{H}$ is bounded, we have that there exists a constant $K< \infty$ such that, for  $\hd_1, \hd_2 \in \mathcal{H}$ 
\begin{align*}
g^{\rm{T}}|r_{\hd_1}(w)r_{\hd_1}(w)^{\rm{T}} - r_{\hd_2}(w)r_{\hd_2}(w)^{\rm{T}}|g\leq K|\hd_1(w) - \hd_2(w)|.
\end{align*}
Additionally, using that $\mathcal{W}$ is bounded and the positivity assumption, we have that for all $w \in \mathcal{W}$ $\pi_{\hd_1, \vartheta_0}(w)(1-\pi_{\hd_1, \vartheta_0}(w)) < C$ for some $C < \infty $. Note also that $\pi_{\hd_1, \vartheta_0}(w) = \pi_{\hd_2, \vartheta_0}(w)$ since $\vartheta_{0_{p+2}}=0$. 
Combining the preceding arguments, we have that
\begin{align*}
&g^{\rm{T}}[r_{\hd_1}(w)r_{\hd_1}(w)^{\rm{T}}[\pi_{\hd_1, \vartheta_0}(w)(1-\pi_{\hd_1, \vartheta_0}(w))] -[r_{\hd_2}(w)r_{\hd_2}(w)^{\rm{T}}[\pi_{\hd_2, \vartheta_0}(w)(1-\pi_{\hd_2, \vartheta_0}(w))]g \\ 
&=[\pi_{\hd_1, \vartheta_0}(w)(1-\pi_{\hd_1, \vartheta_0}(w))]g^{\rm{T}}[r_{\hd_1}(w)r_{\hd_1}(w)^{\rm{T}} - r_{\hd_2}(w)r_{\hd_2}(w)^{\rm{T}}]g \\ \nonumber
&\leq (K \times C)|\hd_1(w) - \hd_2(w)|.
\end{align*}
Because $\mathscr{I}_{\hd, \vartheta_0} = E\left[r_{\hd}(W)r_{\hd}(W)^{\rm{T}}\pi_{\hd, \vartheta_0}(W)(1-\pi_{\hd, \vartheta_0}(W))\right]$ and by the aforementioned arguments, we have that, for $\hd_1, \hd_2 \in \mathcal{H}$ 
\begin{align*}
|g^{\rm{T}}[\mathscr{I}_{\hd_1, \vartheta_0} - \mathscr{I}_{\hd_2, \vartheta_0}]g| \leq K \times C  \|\hd_1(w) - \hd_2(w)\|_{L^1(P_0)}.
\end{align*}
Therefore, letting $Q = K \times C$ we have that,
\begin{align*}
E[\left|f_{W,g}(\hd_1) - f_{W,g}(\hd_2)\right| ] \leq  2  Q \int|\hd_1(w) - \hd_2(w)|d\bar{P}_{W}(w). 
\end{align*}
Next, because $\mathcal{W}$ and $\mathcal{H}$ are bounded,  $E\left[\left|g^{\rm{T}} r_{\hd}(W)r_{\hd}(W)^{\rm{T}}\{\pi_{\hd, \vartheta_0}(W)(1-\pi_{\hd, \vartheta_0}(W))\}g\right|\right] < \infty$. Hence, by the Weak Law of Large Numbers, for all $\hd \in \mathcal{H}$,
\begin{align}
\label{eq1lanscore}
P_{\hd, \vartheta_0}\left(\left| \frac{1}{n} \sum_{i = 1}^n \left\{g^{\rm{T}} \left[ r_\hd(W_i)r_\hd(W_i)^{\rm{T}}[\pi_{\hd, \vartheta_0}(W_i)(1-\pi_{\hd, \vartheta_0}(W_i))] - \mathscr{I}_{\hd, \vartheta_0} \right]g \right\} \right| > \epsilon\right)  \overset{n\to\infty}{\longrightarrow} 0.
\end{align}
Fix $\delta>0$ and let $\mathcal{H}_{\delta}$ be a minimal $\delta/Q$-cover of $\mathcal{H}$ with respect to the $L^1(P_0)$ pseudometric (this cover has finite cardinality since $\mathcal{H}_\delta$ is totally bounded in $L^2(P_0)$ and the $L^2(P_0)$ norm is stronger than the $L^1(P_0)$ norm). Define a $\mathcal{H}\rightarrow \mathcal{H}_{\delta}$ map $H_{\delta}$ so that, for all $\hd \in \mathcal{H}$, $H_{\delta}(\hd)$ is such that $\|\hd-H_{\delta}(\hd)\|_{L^1(\bar{P}_W)}\le \delta/Q$.
Hence, by the triangle inequality and Markov's inequality, we have that
\begin{align*}
       0&\le \sup_{\hd \in \mathcal{H}} P_{\hd, \vartheta_0} \left( \left| g^{\rm{T}}\left[ \frac{1}{n}K_{\hd}(\vartheta;D_n)\bigg|_{\vartheta = \vartheta_0} - \mathscr{I}_{\hd, \vartheta_0} \right]g  \right| > \epsilon \right) \\
       &\leq \sup_{\hd \in \mathcal{H}}P_{\hd, \vartheta_0}\left(\left|f_{w, g}(\hd)-f_{w, g}\circ H_{\epsilon}(\hd)\right| > \epsilon \right) + \sup_{\hd \in \mathcal{H}_\delta} P_{\hd, \vartheta_0} \left( \left| g^{\rm{T}}\left[ \frac{1}{n}K_{\hd}(\vartheta;D_n)\bigg|_{\vartheta = \vartheta_0} - \mathscr{I}_{\hd, \vartheta_0} \right]g  \right| > \epsilon \right) \\
       &\leq \epsilon^{-1}\sup_{\hd \in \mathcal{H}}\int \left|f_{w, g}(\hd)-f_{w, g}\circ H_{\epsilon}(\hd)\right| d\bar{P}_W(w) + \sup_{\hd \in \mathcal{H}_\delta} P_{\hd, \vartheta_0} \left( \left| g^{\rm{T}}\left[ \frac{1}{n}K_{\hd}(\vartheta;D_n)\bigg|_{\vartheta = \vartheta_0} - \mathscr{I}_{\hd, \vartheta_0} \right]g  \right| > \epsilon \right) \\
       &\leq \epsilon^{-1} \sup_{\hd \in \mathcal{H}} Q \int |\hd(w) - H_{\epsilon}(\hd)(w)|d\bar{P}_W(w)  +  \sup_{\hd \in \mathcal{H}_\delta} P_{\hd, \vartheta_0} \left( \left| g^{\rm{T}}\left[ \frac{1}{n}K_{\hd}(\vartheta;D_n)\bigg|_{\vartheta = \vartheta_0} - \mathscr{I}_{\hd, \vartheta_0} \right]g  \right| > \epsilon \right)\\
       &\leq \delta/\epsilon + \sup_{\hd \in \mathcal{H}_\delta} P_{\hd, \vartheta_0} \left( \left| g^{\rm{T}}\left[ \frac{1}{n}K_{\hd}(\vartheta;D_n)\bigg|_{\vartheta = \vartheta_0} - \mathscr{I}_{\hd, \vartheta_0} \right]g  \right| > \epsilon \right)
\end{align*}
Taking a limit superior as $n\rightarrow\infty$, Equation \ref{eq1lanscore} and the fact that $\mathcal{H}_\delta$ has finite cardinality show that
\begin{align*}
    0&\le \limsup_{n\rightarrow\infty}\sup_{\hd \in \mathcal{H}} P_{\hd, \vartheta_0} \left( \left| g^{\rm{T}}\left[ \frac{1}{n}K_{\hd}(\vartheta;d_n)\bigg|_{\vartheta = \vartheta_0} - \mathscr{I}_{\hd, \vartheta_0} \right]g  \right| > \epsilon \right)\le \delta/\epsilon.
\end{align*}
Since $\delta>0$ was arbitrary, the limit superior above is zero. Next, we show that
\begin{align}
  \sup_{\hd \in \mathcal{H}} P_{\hd, \vartheta_0}\left(\left|R_\hd(\vartheta_n, \vartheta_0,D_n) \right| > \epsilon\right)  \overset{n\to\infty}{\longrightarrow} 0. \label{eq:RhSup}
\end{align}
Fix $\hd \in \mathcal{H}$ and $d_n\in \mathcal{Z}^n$. Let $\mathbb{B}_r(x)$ denote a p+2-dimensional open ball of radius $r$ centered at $x$. Writing $U(\vartheta;d_n)_j$ to denote the $j$-th entry of $U(\vartheta;d_n)$, we note that, for each $j$, the function $\vartheta\mapsto U(\vartheta;d_n)_j$ is twice differentiable with continuous Hessian on $\mathbb{B}_{R_n}(\vartheta_0)$, where $R_n = 2 g/\sqrt{n}$. Moreover, for arbitrary $j, k \in \{1,,...p+2\}$ and letting $w_{i,j}$ denote the value of the j-th covariate of the i-th observation, it can be shown that
\begin{align*}
  \frac{\partial^2}{\partial \vartheta_k \partial \vartheta_j}\left\{ \frac{g^{\rm{T}}}{\sqrt{n}}U_\hd(\vartheta;d_n)\right\}= \frac{g^{\rm{T}}}{\sqrt{n}}\sum_{i = 1}^n \left\{w_{i,j}w_{i,k}r_\hd(w_{i})\rho(\vartheta^{\rm{T}} r_\hd(w_i))\left[\frac{1-\exp(\vartheta^{\rm{T}} r_\hd(w_i)}{1+\exp(\vartheta^{\rm{T}} r_\hd(w_i)}\right]\right\}.
\end{align*}
Because $g \in \mathbb{R}$, $\mathcal{W}$ and $\mathcal{H}$ are bounded, and $[1-\exp(x)]/[(1+\exp(x)] <1$ for all $x \in \mathbb{R}$, there exists a finite constant $M$ such that 
\begin{align*}
  \frac{\partial^2}{\partial \vartheta_k \partial \vartheta_j}\left\{ \frac{g^{\rm{T}}}{\sqrt{n}}U_\hd(\vartheta;d_n)\right\} \leq \frac{M}{\sqrt{n}}n = M\sqrt{n}.
\end{align*}
Hence, the second partial derivatives of the function $\vartheta \mapsto g^{\rm{T}}/\sqrt{n}U_{\hd}(\vartheta;d_n)$ can be bounded (up to $\sqrt{n}$) by a finite constant $M$ which does not depend on $\hd$ or $d_n$. By Taylor's Theorem we have that 
\begin{align*}
\left|R_\hd(\vartheta_n, \vartheta_0, d_n)\right|\leq \frac{M\sqrt{n}}{2!}\left\|\frac{g}{\sqrt{n}}\right\|_1^{2} = \frac{M}{2!}n^{-1/2}\|g\|_1^2 \leq \sqrt{p+2}  \frac{M}{2!}n^{-1/2}\|g\|^2,
\end{align*}
where $\|\cdot\|_1$ denotes the $\ell_1$ norm and, as throughout, $\|\cdot\|$ denotes the $\ell_2$ norm. 
Taking a supremum on both sides yields
\begin{align*}
\sup_{\hd\in \mathcal{H}}\left|R_\hd(\vartheta_n, \vartheta_0, d_n)\right| \leq \sqrt{p+2}  \frac{M}{2!}n^{-1/2}\|g\|^2.
\end{align*}
Therefore, for any $\hd \in \mathcal{H}$,
\begin{align*}
    P_{\hd, \vartheta_0}\left(  \sup_{\hd' \in \mathcal{H}}\left|R_{\hd'}(\vartheta_n, \vartheta_0,D_n) \right| > \epsilon\right)&\le I\left\{\sqrt{p+2}  \frac{M}{2!}n^{-1/2}\|g\|^2>\epsilon\right\}.
\end{align*}
Taking a supremum over $\hd \in \mathcal{H}$ and a limit as $n\rightarrow\infty$ then shows that
\begin{align*}
     \sup_{\hd \in \mathcal{H}} P_{\hd, \vartheta_0}\left(  \sup_{\hd' \in \mathcal{H}}\left|R_{\hd'}(\vartheta_n, \vartheta_0,D_n) \right| > \epsilon\right)&\le I\left\{\sqrt{p+2}  \frac{M}{2!}n^{-1/2}\|g\|^2>\epsilon\right\}\rightarrow 0.
\end{align*}
Combining this with the fact that
\begin{align*}
    \sup_{\hd \in \mathcal{H}} P_{\hd, \vartheta_0}\left( \left|r_{\hd}(\vartheta_n, \vartheta_0,D_n) \right| > \epsilon\right)\le \sup_{\hd \in \mathcal{H}} P_{\hd, \vartheta_0}\left(  \sup_{\hd' \in \mathcal{H}}\left|R_{\hd'}(\vartheta_n, \vartheta_0,D_n) \right| > \epsilon\right)
\end{align*} 
establishes \eqref{eq:RhSup}.
\end{proof}

\begin{lemma}
\label{lem:contiguity}
Assume that the conditions of Lemma \ref{lem:qmd} hold and that the $\mathcal{H}$-valued sequence $(h_{0n})_{n=1}^\infty$ is such that $\|h_{0n}-h_0\|_{L^2(P_{h_0, \vartheta_0})}\rightarrow 0$ as $n\rightarrow\infty$. 
Let $g \in \mathbb{R}^{p+2}$, and $\vartheta_n= \vartheta_0 + g/\sqrt{n}$. Then, the sequences of probability measures $(P_{h_{0n},\vartheta_n}^n)_{n=1}^\infty$ and $(P_{h_0, \vartheta_0}^n)_{n=1}^\infty$ are mutually contiguous, where, for a distribution $Q$, $Q^n$ denotes the $n$-fold product measure of $Q$. 
\end{lemma}
\begin{proof} 
Fix a sequence $(h_{0n})_{n=1}^\infty$ that is such that $\|h_{0n}-h_0\|_{L^2(P_{h_0, \vartheta_0})}\rightarrow 0$. We show that the ratio of the log-likelihoods, namely $\log\left\{\prod_{i=1}^n\frac{p_{h_{0n},\vartheta_n}(Z_i)}{p_{h_0, \vartheta_0}(Z_i)}\right\}$, is asymptotically normal under $P_{h_0, \vartheta_0}$. Recall that, because $\vartheta_{0_{p+2}} = 0$, it holds that $P_{h_0, \vartheta_0} = P_{h_{0n}, \vartheta_0}$. A Taylor expansion of the log-likelihood ratios gives that, for some intermediate value $\vartheta'(D_n)$ between $\vartheta_0$ and $\vartheta_n$,
\begin{align}
\label{loglkhd}
\log\left\{\prod_{i=1}^n\frac{p_{h_{0n},\vartheta_n}(Z_i)}{p_{h_0, \vartheta_0}(Z_i)}\right\} &= \sum_{i = 1}^n \left\{\log\left(p_{h_{0n},\vartheta_n}(Z_i)\right)- \log\left(P_{\hd, \vartheta_0}(Z_i)\right) \right\}  \nonumber \\ \nonumber
&= \sum_{i = 1}^n \left\{\log\left(p_{h_{0n},\vartheta_n}(Z_i)\right)- \log\left(p_{h_{0n}, \vartheta_0}(Z_i)\right) \right\} \\ \nonumber
&= (\vartheta_n - \vartheta_0)^{\rm{T}} \left\{\sum_{i=1}^n U_{h_{0n}}(\vartheta_0;Z_i)\right\} +\frac{1}{2}(\vartheta_n - \vartheta_0)^{\rm{T}}\left\{\sum_{i=1}^n K_{\hd_n}(\vartheta'(D_n);Z_i) \right\} (\vartheta_n- \vartheta_0) \\
&= n^{1/2}(\vartheta_n- \vartheta_0)^{\rm{T}} \left\{n^{-1/2}\sum_{i=1}^n U_{h_{0n}}(\vartheta_0;Z_i)\right\} \\ \nonumber
&\quad+ \frac{1}{2} n^{1/2}(\vartheta_n- \vartheta_0)^{\rm{T}}\left\{n^{-1}\sum_{i=1}^n K_{h_{0n}}(\vartheta'(D_n);Z_i)\right\} n^{1/2} (\vartheta_n- \vartheta_0) 
\end{align}
where $K(\vartheta,z)$ is the Jacobian of the score function evaluated at $(\vartheta, z)$ as defined in \eqref{defK}. We analyze the two sums in Equation \ref{loglkhd} as follows. For the first sum, we have that
\begin{align*}
&n^{1/2}(\vartheta_n- \vartheta_0)^{\rm{T}} \left\{n^{-1/2}\sum_{i=1}^n U_{h_{0n}}(\vartheta_0;Z_i)\right\}= n^{1/2}(\vartheta_n- \vartheta_0)^{\rm{T}} \left\{ n^{-1/2}\sum_{i=1}^n U_{h_0}(\vartheta_0;Z_i) \right\} \\
&+n^{1/2}(\vartheta_n- \vartheta_0)^{\rm{T}} \left\{ n^{-1/2}\sum_{i=1}^n \left\{U_{h_{0n}}(\vartheta_0;Z_i) - U_{h_0}(\vartheta_0;Z_i)\right\} \right\}.
\end{align*}
Next, for all $z=(w,a)$,
\begin{align*}
U_{h_{0n}}(\vartheta_0;z)&= r_{h_{0n}}(w) \frac{a -\pi_{h_{0n}, \vartheta_0}(w)}{\pi_{h_{0n}, \vartheta_0}(w)\left(1-\pi_{h_{0n}, \vartheta_0}(w)\right)}\rho(\vartheta_0^{\rm{T}}r_{h_{0n}}(w)) \\
&= r_{h_{0n}}(w)\frac{ -\pi_{h_0, \vartheta_0}(w)}{\pi_{h_0, \vartheta_0}(w)\left(1-\pi_{h_0, \vartheta_0}(w)\right)}\rho(\vartheta_0^{\rm{T}}r_{h_0}(w)).
\end{align*}
The latter implies
\begin{align*}
		U_{h_{0n}}(\vartheta_0;z)- 	U_{h_0}(\vartheta_0;z)= \begin{pmatrix}0\\
\vdots\\
0\\
(h_{0n}(w)-h_0(w))\frac{a -\pi_{h_0, \vartheta_0}(w)}{\pi_{h_0, \vartheta_0}(w)\left(1-\pi_{h_0, \vartheta_0}(w)\right)}\rho(\vartheta_0^{\rm{T}}r_{h_0}(w))
\end{pmatrix}.
\end{align*}
Then, because $\sqrt{n}(\vartheta_n-\vartheta_0) = g$, we have that
\begin{align}
\label{eq:difu}
&n^{1/2}(\vartheta_n- \vartheta_0)^{\rm{T}} \left( n^{-1/2}\sum_{i=1}^n \left\{U_{h_{0n}}(\vartheta_0;Z_i)- 	U_{h_0}(\vartheta_0;Z_i)\right\} \right) = \\ \nonumber
&= g_{p+2} n^{-1/2}\sum_{i=1}^n \left\{[h_{0n}(W_{i})-h_0(W_{i})]\frac{A_{i} -\pi_{h_0, \vartheta_0}(W_{i})}{\pi_{h_0, \vartheta_0}(W_{i})\left(1-\pi_{h_0, \vartheta_0}(W_{i})\right)}\rho(\vartheta_0^{\rm{T}}r_{h_0}(W_{i}))\right\},
\end{align}
The expectation of each element in the above sum is 
\begin{align*}
& E\left[[h_{0n}(W)-h_0(W)]\frac{A -\pi_{h_0, \vartheta_0}(W)}{\pi_{h_0, \vartheta_0}(W)\left(1-\pi_{h_0, \vartheta_0}(W)\right)}\rho(\vartheta_0^{\rm{T}}r_{h_0}(W))\right] \\
&= E\left\{ E\left[[h_{0n}(W)-h_0(W)]\frac{A -\pi_{h_0, \vartheta_0}(W)}{\pi_{h_0, \vartheta_0}(W)\left(1-\pi_{h_0, \vartheta_0}(W)\right)}\rho(\vartheta_0^{\rm{T}}r_{h_0}(W))|W \right] \right\} \\
&= E\left\{\frac{[h_{0n}(W)-h_0(W)]\rho(\vartheta_0^{\rm{T}}r_{h_0}(W))}{\pi_{h_0, \vartheta_0}(W)\left(1-\pi_{h_0, \vartheta_0}(W)\right)} E\left[A -\pi_{h_0, \vartheta_0}(W)|W \right] \right\} = 0.
\end{align*}
Therefore, the expectation of \eqref{eq:difu} is equal to 0.
Next, the variance of the right-hand side in Equation \ref{eq:difu} is equal to
\begin{align}
\label{eq1:lemcontiguitycondvar}
&\text{var}\left(g_{p+2} n^{-1/2}\sum_{i=1}^n \left\{[h_{0n}(W_{i})-h_0(W_{i})]\frac{A_{i} -\pi_{h_0, \vartheta_0}(W_{i})}{\pi_{h_0, \vartheta_0}(W_{i})\left(1-\pi_{h_0, \vartheta_0}(W_{i})\right)}\rho(\vartheta_0^{\rm{T}}r_{h_0}(W_{i}))\right\}\right) \\ \nonumber
&=g_{p+2}^2 E\left\{[h_{0n}(W)-h_0(W)]^2\left(\frac{A -\pi_{h_0, \vartheta_0}(W)}{\pi_{h_0, \vartheta_0}(W)\left(1-\pi_{h_0, \vartheta_0}(W)\right)}\rho(\vartheta_0^{\rm{T}}r_{h_0}(W))\right)^2\right\}.
\end{align}
We show that the above variance converges to 0 as $n \rightarrow \infty$. Because $\mathcal{W}$ is bounded and functions in $\mathcal{H}$ are uniformly bounded, and since $0<\pi_{h_0, \vartheta_0}(w)<1$ for all $w \in \mathcal{W}$, the quantity $\left(\frac{a -\pi_{h_0, \vartheta_0}(w)}{\pi_{h_0, \vartheta_0}(w)\left(1-\pi_{h_0, \vartheta_0}(w)\right)}\rho(\vartheta_0^{\rm{T}}r_{h_0}(w))\right)^2$ can be uniformly bounded by a constant $M< \infty$. Next, by assumption, $\|h_{0n}-h_0\|_{L^2(P_{h_0, \vartheta_0})}\rightarrow 0$. Hence, the variance in \eqref{eq1:lemcontiguitycondvar} can be upper bounded by a term that converges to 0, therefore it also converges to 0.
Because the mean of \eqref{eq:difu} is equal to 0 and its variance converges to 0, Chebyshev's inequality shows that, for any $\epsilon>0$,
\begin{align}
\label{eq2:lemcontiguity}
&P_{h_0, \vartheta_0}\left(\left|n^{1/2}(\vartheta_n - \vartheta_0)^{\rm{T}} \left\{n^{-1/2}\sum_{i=1}^n \left[U_{h_{0n}}(\vartheta_0;Z_i) - U_{h_0}(\vartheta_0;Z_i)\right] \right\}\right| > \varepsilon \right) \nonumber \\
&\le \frac{\text{var}\left(g_{p+2} n^{-1/2}\sum_{i=1}^n \left\{[h_{0n}(W_{i})-h_0(W_{i})]\frac{A_{i} -\pi_{h_0, \vartheta_0}(W_{i})}{\pi_{h_0, \vartheta_0}(W_{i})\left(1-\pi_{h_0, \vartheta_0}(W_{i})\right)}\rho(\vartheta_0^{\rm{T}}r_{h_0}(W_{i}))\right\}\right)}{\varepsilon^2} \nonumber \\
&\overset{n\to\infty}{\longrightarrow} 0.
\end{align}
Next, by the central limit theorem, under $P_{h_0, \vartheta_0}$, 
\begin{align}
\label{eq3:lemcontiguity}
n^{-1/2}\sum_{i=1}^n U_{h_0}(\vartheta_0;Z_i)  \rightsquigarrow N\left(\vec{0},\mathscr{I}_{h_0, \vartheta_0}\right).
\end{align}
By continuous mapping theorem,
\begin{equation*}
n^{1/2}(\vartheta_n - \vartheta_0)^{\rm{T}} \left\{n^{-1/2}\sum_{i=1}^n U_{h_0}(\vartheta_0;Z_i) \right\}\rightsquigarrow N\left(0, g^{\rm{T}} \mathscr{I}_{h_0, \vartheta_0} g\right).
\end{equation*}
Combining the arguments shown in \eqref{eq2:lemcontiguity} and \eqref{eq3:lemcontiguity}, we have that, under $P_{h_0, \vartheta_0}$,
\begin{align}
\label{eq4:lemcontiguity}
  (\vartheta_n - \vartheta_0)^{\rm{T}} \left\{\sum_{i=1}^n U_{h_{0n}}(\vartheta_0;Z_i)\right\} \rightsquigarrow N\left(0,g^{\rm{T}} \mathscr{I}_{h_0, \vartheta_0} g\right).
\end{align}
This concludes the study of the leading term in \eqref{loglkhd}.
Now, the second term in \eqref{loglkhd} can be reexpressed as 
\begin{align}
\label{eq5:lemcontiguity}
	&\frac{1}{2}n^{1/2}(\vartheta_n - \vartheta_0)^{\rm{T}}\left\{n^{-1}\sum_{i=1}^n K_{h_{0n}}(\vartheta'(D_n);Z_i)\right\}n^{1/2} (\vartheta_n - \vartheta_0) \nonumber \\
	&= \frac{1}{2}n^{1/2}(\vartheta_n - \vartheta_0)^{\rm{T}}\left\{n^{-1}\sum_{i=1}^n K_{h_0}(\vartheta_0;Z_i)\right\} n^{1/2} (\vartheta_n - \vartheta_0)  \\
	&+ \frac{1}{2}n^{1/2}(\vartheta_n - \vartheta_0)^{\rm{T}}\left[n^{-1}\sum_{i=1}^n \left\{K_{h_{0n}}(\vartheta_0;Z_i) - K_{h_0}(\vartheta_0;Z_i) \right\}\right] n^{1/2} (\vartheta_n - \vartheta_0) \nonumber \\
	&+ \frac{1}{2}n^{1/2}(\vartheta_n - \vartheta_0)^{\rm{T}}\left[n^{-1}\sum_{i=1}^n \left\{K_{h_{0n}}(\vartheta'(D_n);Z_i) - K_{h_{0n}}(\vartheta_0;Z_i) \right\}\right]n^{1/2} (\vartheta_n - \vartheta_0). \nonumber
\end{align}
By the weak law of large numbers, $n^{-1}\sum_{i=1}^n K_{h_0}(\vartheta_0;Z_i)\longrightarrow  -\mathscr{I}_{\vartheta_0 h_0}$ in probability under $P_{h_0, \vartheta_0}$. Therefore, by the continuous mapping theorem,  
\begin{equation}
\label{eq6:lemcontiguity}
P_{h_0, \vartheta_0}\left(\left|\frac{1}{2} n^{1/2}(\vartheta_n - \vartheta_0)^{\rm{T}}\left\{n^{-1}\sum_{i=1}^n K_{h_0}(\vartheta_0;Z_i)\right\} n^{1/2} (\vartheta_n - \vartheta_0) + \frac{1}{2}g^{\rm{T}} \mathscr{I}_{h_0, \vartheta_0}g\right|>\varepsilon \right) \overset{n\rightarrow \infty}{\longrightarrow} 0 
\end{equation}
Next, note that for all $z$,
\begin{align*}
    K_{h_{0n}}(\vartheta_0;z)&= [\pi_{h_{0n}, \vartheta_0}(w)][1-\pi_{h_{0n}, \vartheta_0}(w)]r_{h_{0n}}(w)r_{h_{0n}}(w)^{\rm{T}} \\
    &=[\pi_{h_0, \vartheta_0}(w)][1-\pi_{h_0, \vartheta_0}(w)]r_{h_{0n}}(w)r_{h_{0n}}(w)^{\rm{T}},
\end{align*}
where we have used that $\vartheta_{0_{p+2}}=0$ implies that $\pi_{\hd, \vartheta_0}=\pi_{h_0, \vartheta_0}$ for any $\hd \in \mathcal{H}$.
Hence,
\begin{align*}
&\frac{1}{2} n^{1/2}(\vartheta_n - \vartheta_0)^{\rm{T}}\left[n^{-1}\sum_{i=1}^n \left\{ K_{h_{0n}}(\vartheta_0;Z_i) -  K_{h_0}(\vartheta_0;Z_i)\right\}\right]n^{1/2} (\vartheta_n - \vartheta_0) = \\
&= \frac{g^{\rm{T}}}{2}\left[n^{-1}\sum_{i=1}^n \left(r_{h_{0n}}(W_i) r_{h_{0n}}(W_i)^{\rm{T}}-r_{h_0}(W_i) r_{h_0}(W_i)^{\rm{T}}\right)[\pi_{h_0, \vartheta_0}(W_i)(1-\pi_{h_0, \vartheta_0}(W_i))]\right]g. 
\end{align*}
We now show that the above term converges to 0 in probability by showing that the expected value of each entry of the matrix
\begin{align*}
\bigg|n^{-1}\sum_{i=1}^n\left(r_{h_{0n}}(W_i) r_{h_{0n}}(W_i)^{\rm{T}}-r_{h_0}(W_i) r_{h_0}(W_i)^{\rm{T}}\right)[\pi_{h_0, \vartheta_0}(W_i)(1-\pi_{h_0, \vartheta_0}(W_i))]\bigg|
\end{align*}
converges to 0, where $|\cdot|$ denotes the entrywise absolute value. First, for a given $w$, the matrix $r_{h_{0n}}(w)r_{h_{0n}}(w)^{\rm{T}}-r_{h_0}(w) r_{h_0}(w)^{\rm{T}}$ consists of the following entries:
\begin{equation*}
\left(r_{h_{0n}}(w) r_{h_{0n}}(w)^{\rm{T}}-r_{h_0}(w) r_{h_0}(w)^{\rm{T}}\right)_{l,k}=\begin{cases}
0 & l \leq p,\;k \leq p\\
w_{k}\left(h_{0n}(w)-h_0(w)\right) & l=p+2,\;1\leq k\leq p+1\\
w_{l}\left(h_{0n}(w)-h_0(w)\right) & k=p+2,\;1\leq l\leq p+1\\
h_{0n}(w)^{2}-h_0(w)^2 & k=p+2,\;l=p+2
\end{cases}
\end{equation*}
where $w_j$ is the $j-th$ observed covariate. For those entries that are different from 0, we apply Cauchy-Schwartz inequality:
\begin{align*}
	&E\left|\left(r_{h_{0n}}(W) r_{h_{0n}}(W)^{\rm{T}}-r_{h_0}(W) r_{h_0}(W)^{\rm{T}}\right)_{l,k}\left[\pi_{h_0, \vartheta_0}(W)(1-\pi_{h_0, \vartheta_0}(W))\right]\right| \\
	&\leq E\left\{\left(r_{h_{0n}}(W) r_{h_{0n}}(W)^{\rm{T}}-r_{h_0}(W) r_{h_0}(W)^{\rm{T}}\right)_{l,k}^2 \right\}^{1/2}E\left\{\left[\pi_{h_0, \vartheta_0}(W)(1-\pi_{h_0, \vartheta_0}(W))\right]^2   \right\}^{1/2}
\end{align*}
Because $\|h_{0n}-h_0\|_{L^2(P_{h_0, \vartheta_0})}\rightarrow 0$ and functions in $\mathcal{H}$ are uniformly bounded, we have that
$E\{[h_{0n}(W)^{2}-h_0(W)^{2}]^2\}= E\{[h_{0n}(W)+h_0(W)]^2[h_{0n}(W)-h_0(W)]^2\} \overset{n\to\infty}{\longrightarrow} 0$. Additionally, because $\mathcal{W}$ is bounded, $E\{[W_{l}(h_{0n}(W)-h_0(W))]^2\} \overset{n\to\infty}{\longrightarrow} 0$ for all $l \in \{1,...,p+1\}$. The latter statements combined with the fact that that $E\left\{\left[\pi_{h_0, \vartheta_0}(W)(1-\pi_{h_0, \vartheta_0}(W))\right]^2  \right\} < 1$ imply that
\begin{align*}
E\left[\left(r_{h_{0n}}(W) r_{h_{0n}}(W)^{\rm{T}}-r_{h_0}(W) r_{h_0}(W)^{\rm{T}}\right)_{l,k}^2 \right]^{1/2}E\left\{\left[\pi_{h_0, \vartheta_0}(W)(1-\pi_{h_0, \vartheta_0}(W))\right]^2 \right\}^{1/2}\overset{n\to\infty}{\longrightarrow} 0.
\end{align*}
Let $0_{(p+2)\times(p+2)}$ denote a $(p+2)\times(p+2)$ matrix of 0's. Then, the above limit and Markov's inequality imply that
\begin{align*}
n^{-1}\sum_{i=1}^n \left(r_{h_{0n}}(W_i) r_{h_{0n}}(W_i)^{\rm{T}}-r_{h_0}(W_i) r_{h_0}(W_i)^{\rm{T}}\right)\left[\pi_{h_0, \vartheta_0}(W_i)(1-\pi_{h_0, \vartheta_0}(W_i))\right] \overset{n\rightarrow\infty}{\longrightarrow}0_{(p+2)\times(p+2)}
\end{align*}
in probability under $P_{h_0, \vartheta_0}$. By the continuous mapping theorem,
\begin{align}
\label{eq7:lemcontiguity}
\frac{1}{2} n^{1/2}(\vartheta_n - \vartheta_0)^{\rm{T}}\left[n^{-1}\sum_{i=1}^n \left\{K_{h_{0n}}(\vartheta_0;Z_i) -  K_{h_0}(\vartheta_0;Z_i)\right\}\right]n^{1/2} (\vartheta_n - \vartheta_0) = o_{P_{h_0, \vartheta_0}}(1).
\end{align}
Finally, we show that the last term in Equation \ref{eq5:lemcontiguity} converges to 0 in probability. This term is equal to 
\begin{align*}
 &\frac{1}{2}n^{1/2}(\vartheta_n - \vartheta_0)^{\rm{T}}\left[n^{-1}\sum_{i=1}^n \left\{K_{h_{0n}}(\vartheta'(D_n);Z_i) -  K_{h_{0n}}(\vartheta_0;Z_i)\right\}\right]n^{1/2} (\vartheta_n - \vartheta_0)\\
	&=\frac{g^{\rm{T}}}{2}\left[n^{-1}\sum_{i=1}^n r_{h_{0n}}(W_{i})r_{h_{0n}}(W_{i})^{\rm{T}}\left(\left[\pi_{\vartheta'(D_n),h_{0n}}(W_i)(1-\pi_{\vartheta'(D_n), h_{0n}}(W_i))\right]-\left[\pi_{h_{0n}, \vartheta_0}(W_i)(1-\pi_{h_{0n}, \vartheta_0}(W_i))\right]\right)\right]g.
\end{align*}
Fix $w \in \mathcal{W}$ and $\hd \in \mathcal{H}$ and define the $\mathbb{R}^{p+2} \rightarrow \mathbb{R}$ function $v_{\hd, w}(\vartheta):= \pi_{\hd,\vartheta}(w_i)(1-\pi_{\hd,\vartheta}(w_i))$. We show that $v$ is Lipschitz continuous in $\vartheta$ uniformly over $\hd \in \mathcal{H}$ and $w \in\mathcal{W}$. We do so by showing that its derivative is uniformly upper bounded entrywise. The absolute value of the derivative is
\begin{align*}
    \left|\frac{\partial}{\partial \vartheta}v_{\hd, w}(\vartheta)\right| &= \left|\frac{\partial}{\partial \vartheta} \pi_{\hd,\vartheta}(w) - \frac{\partial}{\partial \vartheta}\pi_{\hd,\vartheta}(w)^2 \right| = \left|\rho(\vartheta^{\rm{T}}r_{\hd}(w))r_{\hd}(w) - 2\pi_{\hd,\vartheta}(w)\rho(\vartheta^{\rm{T}}r_{\hd}(w))r_{\hd}(w)\right| \\
    &= \left|\rho(\vartheta^{\rm{T}}r_{\hd}(w))r_{\hd}(w)[1-2\pi_{\hd,\vartheta}(w)] \right|.
\end{align*}
By definition $|\rho(\vartheta^{\rm{T}}r_{\hd}(w))| <1$ for all $w \in \mathcal{W}$ and $\hd \in \mathcal{H}$. Next, because $w \in \mathcal{W}$ and $\hd \in \mathcal{H}$ are bounded, we can upper bound each entry of $r_{\hd}(w)$ by a finite constant $L$ that depends neither on $w$ nor $\hd$. Finally, because $0<\pi_{\hd,\vartheta}(w)<1$ regardless of $w$ or $\hd$, we have that $|1-2\pi_{\hd,\vartheta}(w)|\leq 1$. Hence, for all $\vartheta$,
\begin{align*}
    \sup_{w \in \mathcal{W}, \hd \in \mathcal{H}}\left|\frac{\partial}{\partial \vartheta}v_{\hd, w}(\vartheta)\right| \leq L,
\end{align*}
which implies that $v_{\hd, w}(\cdot)$ is L-Lipschitz uniformly over $w \in \mathcal{W}$ and $\hd \in \mathcal{H}$. Combining the preceding arguments,
\begin{align*}
    &E\left|\frac{g^{\rm{T}}}{2}\left[n^{-1}\sum_{i=1}^n r_{h_{0n}}(W_{i})r_{h_{0n}}(W_{i})^{\rm{T}}\left(\left[\pi_{\vartheta'(D_n), h_{0n}}(W_i)(1-\pi_{\vartheta'(D_n), h_{0n}}(W_i))\right]-\left[\pi_{h_{0n}, \vartheta_0}(W_i)(1-\pi_{h_{0n}, \vartheta_0}(W_i))\right]\right)\right]g \right|\\
    &\leq L^2\|g\|^2 E\left[\|\vartheta'(D_n) - \vartheta_0\|\right].
\end{align*}
By Markov's inequality and the above display, 
\begin{align}
\label{eq8:lemcontiguity}
 &P_{h_0, \vartheta_0}\left(\left|\frac{1}{2} n^{1/2}(\vartheta_n- \vartheta_0)^{\rm{T}}\left[n^{-1}\sum_{i=1}^n \left\{K_{h_{0n}}(\vartheta'(D_n);Z_i) -  K_{h_{0n}}(\vartheta_0;Z_i)\right\}\right] n^{1/2} (\vartheta_n- \vartheta_0)\right| > \varepsilon \right) \\ \nonumber
 &\leq \varepsilon^{-1}L^2\|g\|^2E\left[\|\vartheta'(D_n) -  \vartheta_0\|\right] \leq \varepsilon^{-1}L^2\|g\|^2 \|\vartheta_n-  \vartheta_0\| = \varepsilon^{-1}L^2\frac{\|g\|^3}{\sqrt{n}} \overset{n\rightarrow \infty}{\longrightarrow} 0. 
 \end{align}
Therefore, by Equations \ref{eq5:lemcontiguity},  \ref{eq6:lemcontiguity}, \ref{eq7:lemcontiguity}, and \ref{eq8:lemcontiguity}, we have that 
\begin{align}
\label{eq9:lemcontiguity}
   P_{h_0, \vartheta_0}\left( \left|\frac{1}{2} n^{1/2}(\vartheta_n - \vartheta_0)^{\rm{T}}\left\{n^{-1}\sum_{i=1}^n K_{h_{0n}}(\vartheta'(D_n);Z_i)\right\} n^{1/2} (\vartheta_n - \vartheta_0) + \frac{1}{2}g^{\rm{T}} \mathscr{I}_{h_0, \vartheta_0}g\right| > \varepsilon \right) \overset{n \rightarrow \infty }{\longrightarrow }0 
\end{align}
Combining the equality in \ref{loglkhd} along with the results stated in Equations \ref{eq4:lemcontiguity} and \ref{eq9:lemcontiguity}, implies that
\begin{align*}
\log\left\{\prod_{i = 1}^n\frac{p_{h_{0n},\vartheta_n}(Z_i)}{p_{h_0, \vartheta_0}(Z_i)}\right\} \rightsquigarrow N\left(-\frac{1}{2}g^{\rm{T}} \mathscr{I}_{h_0, \vartheta_0}g, g^{\rm{T}} \mathscr{I}_{h_0, \vartheta_0}g\right)
\end{align*}
under $P_{h_0, \vartheta_0}$. Hence, by Le Cam's first lemma (see Example~6.5 in \cite{van2000asymptotic} for an appropriate version) $(P_{h_{0n},\vartheta_n}^n)_{n=1}^\infty$ and $(P_{h_0, \vartheta_0}^n)_{n=1}^\infty$ are mutually contiguous. 
\end{proof}
The following lemmas describe the asymptotic behavior of the maximum likelihood estimator of $\vartheta_0$ along the sequence of models $\{P_{h_{0n},\vartheta} : \vartheta \in \mathbb{R}^{p+2}\}_{n=1}^\infty$. Recall that $\ell_{\hd}(\vartheta; d_n)$ is the log-likelihood function given $d_n$ and $\hd \in (h_{0n})_{n=1}^{\infty}$,
\begin{equation}
\label{def:loglklhd}
\ell_{\hd }(\vartheta; d_n) := \frac{1}{n}\sum_{i = 1}^n \left[a_i\log \pi_{\hd,\vartheta} (w_i) + (1-a_i)\log (1-\pi_{\hd,\vartheta}(w_i))\right].
\end{equation}
Then, the maximum likelihood estimator of $\vartheta$ based on $d_n$ and $\hd$ is defined as
\begin{align*}
  \vartheta_{\hd, n}:= \argmax_{\vartheta}\ell_{\hd}(\vartheta; d_n),
\end{align*}
where this maximizer exists and is unique with probability one under our conditions. Finally, let $L(\vartheta, h)$ denote the expectation of $\ell_{\hd}(\vartheta; d_n)$,
\begin{equation*}
L(\vartheta, \hd) := E\left[A \log \pi_{\hd,\vartheta} (W) + (1-A)\log (1-\pi_{\hd,\vartheta} (W))\right].
\end{equation*}
The following lemma shows that $\vartheta_{h_{0n},n}$ is consistent for $\vartheta_0$.
\begin{lemma}
\label{lem:consistency} 
Assume that the $\mathcal{H}$-valued sequence $(h_{0n})_{n=1}^\infty$ is such that $\|h_{0n}-h_0\|_{L^2(P_{h_0, \vartheta_0})}\rightarrow 0$ as $n\rightarrow\infty$ and that the conditions of Proposition \ref{prop1} hold. Let $\vartheta_{h_{0n},n}=  \underset{\vartheta}{\mathrm{argmax}}\, \ell_{h_{0n}}(\vartheta; D_n)$ be the maximum likelihood estimator of $\vartheta_0$. Then, for every $\epsilon > 0$, 
\begin{equation}
P_{h_0, \vartheta_0}\left( \|\vartheta_{h_{0n},n} - \vartheta_0 \| \ge \epsilon\right) \overset{n \rightarrow \infty}{\longrightarrow} 0. 
\end{equation}
\end{lemma}
\begin{proof}
Let $D_n$ be a dataset arising as $n$ iid draws from $P_{h_0, \vartheta_0}$. 
It can be shown that, under the conditions of this lemma, it is true that (i) $L(\cdot,h_0)$ is uniquely maximized at $\vartheta_0$, (ii) $\ell_{h_{0n}}(\cdot,D_n)$ is concave almost surely, (iii) $\ell_{h_{0n}}(\vartheta,D_n)\rightarrow L(\vartheta,h_0)$ in probability for all $\vartheta\in\mathbb{R}^{p+2}$ (details are omitted here for brevity). Hence, Theorem~2.7 in \cite{newey1994large} gives the desired result.
\end{proof}
One can also show that the maximum likelihood estimator based on $h_0$, $\vartheta_{h_0,n}$, is consistent for $\vartheta_0$ by setting the $\mathcal{H}$-valued sequence in Lemma \ref{lem:consistency} to $(h_0)_{n=1}^{\infty}$. 
Next, the following lemma shows that the maximum likelihood estimator $\vartheta_{h_{0n},n}$ is asymptotically linear. 
\begin{lemma}
\label{lem:asymnorm} 
Assume that the $\mathcal{H}$-valued sequence $(h_{0n})_{n=1}^\infty$ is such that $\|h_{0n}-h_0\|_{L^2(P_{h_0, \vartheta_0})}\rightarrow 0$ as $n\rightarrow\infty$ and that the conditions of Proposition \ref{prop1} hold. Let $\vartheta_{h_{0n},n} = \underset{\vartheta}{\mathrm{argmax}}\, \ell_{h_{0n}}(\vartheta; D_n)$ be the maximum likelihood estimator of $\vartheta_0$, $g \in \mathbb{R}^{p+2}$, and $\vartheta_n= \vartheta_0 + g/\sqrt{n}$. Then, under sampling from $P_{h_{0n},\vartheta_n}^n$, 
\begin{align*}
\sqrt{n}[\vartheta_{h_{0n},n} - \vartheta_n] &= \mathscr{I}_{h_0, \vartheta_0}^{-1}\frac{1}{\sqrt{n}}U_{h_{0n}}(\vartheta_n; D_n) + o_p(1).
\end{align*}
\end{lemma}
\begin{proof}
We first show that, under $P^n_{h_0, \vartheta_0}$, the following holds:
\begin{align*}
\sqrt{n}[\vartheta_{h_{0n},n}  - \vartheta_0] = \mathscr{I}_{h_0, \vartheta_0}^{-1} \frac{1}{\sqrt{n}}U_{h_0}(\vartheta_0; D_n) + o_{P_{h_0, \vartheta_0}}(1).
\end{align*}
 Let $Z_{0}(\vartheta, \hd):=\frac{\partial}{\partial \vartheta}L(\vartheta, \hd)$, $ \vartheta_{h_{0},n} := \underset{\vartheta}{\mathrm{argmax}}\,\ell_{{h_0}}(\vartheta; d_n)$ be the maximum likelihood estimator when $h_0$ is known, and $Z_n(\vartheta,\hd):=U_{h}(\vartheta;d_n)$. Applying Triangle Inequality and by definition of the MLE's we have that
\begin{align}
\label{eq1:asymnorm}
    \nonumber  \|Z_n(\vartheta_{h_{0n},n},h_0) - Z_n(\vartheta_{h_{0},n},h_0)\|&=\|Z_n(\vartheta_{h_{0n},n},h_0) - Z_n(\vartheta_{h_{0n},n},h_{0n})\| \\ \nonumber 
    &\le \|Z_n(\vartheta_{h_{0n},n},h_0) - Z_0(\vartheta_{h_{0n},n},h_0) + Z_0(\vartheta_{h_{0n},n},h_{0n}) - Z_n(\vartheta_{h_{0n},n},h_{0n})\| \\
    &\quad+ \|Z_0(\vartheta_{h_{0n},n},h_0) - Z_0(\vartheta_{h_{0n},n},h_{0n})\|.
\end{align}
We begin by showing that, for every $\epsilon >0$,
\begin{equation}
\label{eq2:asymnorm}
   P_{h_0,\vartheta_0}\left( \|Z_n(\vartheta_{h_{0n},n},h_0) - Z_0(\vartheta_{h_{0n},n},h_0) + Z_0(\vartheta_{h_{0n},n},h_{0n}) - Z_n(\vartheta_{h_{0n},n},h_{0n})\| > n^{-1/2}\epsilon\right)\overset{n \to \infty}{\longrightarrow}0.
\end{equation}
By consistency of $\vartheta_{h_{0n},n}$ and the assumption on  $(h_{0n})_{n=1}^{\infty}$ we know that there exists a sequence $\delta_n$ that converges to zero sufficiently slowly so that $P_{h_0, \vartheta_0}\left(\|\vartheta_{h_{0n},n}-\vartheta_0\|>\delta_n\right)\rightarrow 0$ and $\|h_{0n} - h_0\| \leq \delta_n$. Thus, it suffices to show that 
\begin{align*}
    P_{h_0, \vartheta_0}\left(\sup_{\vartheta \in D_{\delta_n}}\|Z_n(\vartheta,h_0) - Z_0(\vartheta,h_0) + Z_0(\vartheta,h_{0n})  - Z_n(\vartheta,h_{0n})  \|> n^{-1/2}\epsilon \right) \overset{n \to \infty}{\longrightarrow}0.
\end{align*}
where $D_{\delta} = \{\vartheta \in \mathbb{R}^{p+2}: \|\vartheta - \vartheta_0\| \leq \delta\}$. Now, because
\begin{align*}
&\sup_{\vartheta \in D_{\delta_n}}\|Z_n(\vartheta,h_0) - Z_0(\vartheta,h_0) + Z_0(\vartheta,h_{0n})  - Z_n(\vartheta,h_{0n})\|  \\
=&\sup_{\vartheta \in D_{\delta_n}}\|Z_n(\vartheta,h_0) - Z_0(\vartheta,h_0) +[Z_n(\vartheta_0,h_0) - Z_n(\vartheta_0,h_0)]+ Z_0(\vartheta,h_{0n}) - Z_n(\vartheta,h_{0n})| \\
&\leq \sup_{\vartheta \in D_{\delta_n}} \|Z_n(\vartheta,h_0) - Z_n(\vartheta_0,h_0) -Z_0(\vartheta,h_0))\| + \sup_{\vartheta \in D_{\delta_n}} \|Z_n(\vartheta_0,h_0) -Z_n(\vartheta,h_{0n}) + Z_0(\vartheta,h_{0n})\| 
\end{align*}
and since $h_0 \in \mathcal{H}_{\delta_n}$, we can instead show that
\begin{align}
\label{eq3:asymnorm}
\sup_{\hd \in\mathcal{H}_{\delta_n}}P_{h_0, \vartheta_0}\left(\sup_{\vartheta \in D_{\delta_n}}\|Z_n(\vartheta,\hd) - Z_n(\vartheta_0,h_0) - Z_0(\vartheta,\hd) \| > n^{-1/2}\epsilon \right)\overset{n \to \infty}{\longrightarrow}0.
\end{align}
Fix $\hd \in \mathcal{H}_{\delta_n}$ and let $$f_{\hd,\vartheta}(z) := a \frac{\rho(\vartheta^{\rm{T}} r_{\hd}(w))}{\pi_{\hd,\vartheta}(w)}r_{\hd}(w)+ (a-1) \frac{\rho(\vartheta^{\rm{T}} r_{\hd}(w))}{1-\pi_{\hd,\vartheta}(w)}r_{\hd}(w)- a \frac{\rho(\vartheta_0^{\rm{T}} r_{h_0}(w))}{\pi_{h_0, \vartheta_0}(w)}r_{\hd}(w)-(a-1) \frac{\rho(\vartheta_0^{\rm{T}} r_{h_0}(w))}{1-\pi_{h_0, \vartheta_0}(w)}r_{h_0}(w).$$
Additionally, let $f^k_{\hd,\vartheta}(z)$ denote the k-th entry of the function $f_{\hd,\vartheta}(\cdot)$ and $(M)_k$ denote the k-th column of a matrix $M$. Further, we also define the class of functions $\mathcal{F}^k_{\hd,\delta} = \{z\mapsto f_{\hd,\vartheta}^k(z) :  \vartheta \in \mathbb{R}^{p+2}, \vartheta \in D_{\delta}\}$, and let $N_{[]}(\epsilon, \mathcal{F}_{\hd,\delta}^k, L_{2}(P))$ denote its bracketing number of radius $\epsilon$ with respect to the $L_{2}(P)$ norm. We will show that $E[|\sup_{f \in \mathcal{F}_{\delta_n,h}^k}\mathbb{G}_n f|] \longrightarrow 0$ as $n \rightarrow \infty$. Fix $z \in \mathcal{Z}$ and $k \in \{1,...,p+2\}$ and let $\nu_{z,\hd}^k: \vartheta \mapsto f_{\hd,\vartheta}^k(z)$ be a function that maps $\vartheta$ to the k-th entry of $f_{\hd,\vartheta}(z)$. Each entry of the derivative
\begin{align*}
    \left|\frac{\partial }{\partial \vartheta} \nu_{z,\hd}^k\right|
    &= \left|(\rho(\vartheta^{\rm{T}}r_{\hd}(w))r_{\hd}(w)r_{\hd}(w)^{\rm{T}})_{k}\right|
\end{align*}
can be uniformly bounded by a finite constant $C>0$ because both $\mathcal{W}$ and $\mathcal{H}$ are bounded. Therefore, the function $\nu_{z,\hd}^k(\vartheta)$ is Lipschitz in $\vartheta$. Hence, for $\vartheta_1,\vartheta_2 \in D_{\delta_n}$ we have that
\begin{align}
\label{eq4:asymnorm}
    |f_{\hd, \vartheta_1}^k(z) -  f_{\hd, \vartheta_2}^k(z)| \leq C\|\vartheta_1 - \vartheta_2\| \leq 2 C  \delta_n
\end{align}
Then, by Example 19.7 in \cite{van2000asymptotic}, we have that there exists a constant $K$ that depends on $p$ and such that
\begin{align*} 
N_{[]}(\epsilon C, \mathcal{F}_{\hd,\delta_n}^k, L_{2}(P)) \leq K \left(\frac{2\delta_n}{\epsilon}\right)^{p+2}, \textnormal{ for every } 0 < \epsilon < 2 \delta_n.
\end{align*}
Next, if $\vartheta = \vartheta_0$, then $f_{\hd,\vartheta}^k = 0$. Applying the bound in Equation \ref{eq4:asymnorm} with $\vartheta_2 = \vartheta_0$, we have that $|f_{\vartheta_1, \hd}^k(z)| \leq C\|\vartheta_1-\vartheta_0\| \leq  C \delta_n$. Because the previous bound is irrespective of $\vartheta_1$ and $\hd$, we have that for every $z \in \mathcal{Z}$,
$\sup_{\hd\in\mathcal{H}}\sup_{f_1\in\mathcal{F}_{\hd,\delta_n}^k}|f_1(z)| \leq  C\delta_n$. Then, let $F_n^k(z) = C\delta_n$ be the envelope function of the function class $\mathcal{F}_{\hd, \delta_n}^k$. Therefore, for every $0<\epsilon < 2$, the bracketing integral is finite because it can be upper bounded as follows,
\begin{align*} 
\int_{0}^{1} \sqrt{1 + \log N_{[]}(\epsilon C\delta_n, \mathcal{F}_{\hd, \delta_n}^k, L_2(P))}d\epsilon
\leq \int_{0}^{1}  \sqrt{1 +\log K + (p+2)\log   \left(\frac{2}{\epsilon}\right)}d\epsilon =: K'<\infty. 
\end{align*}
By Theorem 2.14.2 in \cite{van1996weak}, $E|\sup_{f \in \mathcal{F}_{\hd,\delta_n}^k}\mathbb{G}_n f| \leq K' C \delta_n \overset{n\to \infty}{\longrightarrow} 0$. Because the index $k$ was arbitrary and the upper bound $K' C \delta_n$ does not depend on $\hd$, we can establish that \eqref{eq3:asymnorm} holds via Markov's inequality.
Next, we show that 
\begin{equation}
\label{eq3:lemZrootncons}
  \|Z_0(\vartheta_{h_{0n},n},h_0) - Z_0(\vartheta_{h_{0n},n},h_{0n})\| =  o_p(\|\vartheta_{h_{0n},n} - \vartheta_0\|).
\end{equation}
Using similar notation for multivariate functions as before, we will let $Z_0^k(\vartheta,\hd)$ denote the k-th value of one of its outputs. First, 
\begin{align*}
    Z^k_0(\vartheta_{h_{0n},n},h_0) - Z^k_0(\vartheta_{h_{0n},n},h_{0n})&= Z^k_0(\vartheta_{h_{0n},n},h_0) - Z^k_0(h_0,\vartheta_0) + Z^k_0(h_{0n}, \vartheta_0) - Z^k_0(\vartheta_{h_{0n},n},h_{0n})
\end{align*}
since $Z^k_0(\vartheta_0,h_0)=Z^k_0(\vartheta_0, h_{0n})=0$. Next, by the fundamental theorem of calculus, we have that, for any $\hd$, $\vartheta$,
\begin{align*}
    Z^k_0(\vartheta, \hd)-Z^k_0(\vartheta_0,\hd)&= \int_0^{\|\vartheta-\vartheta_0\|} \left.\frac{\partial}{\partial\delta} Z^k_0\left(\vartheta_0 + \delta \frac{\vartheta-\vartheta_0}{\|\vartheta-\vartheta_0\|},\hd\right) \right|_{\delta=\epsilon}\, d\epsilon.
\end{align*}
Hence,
\begin{align*}
    &Z^k_0(\vartheta_{h_{0n},n},h_0) - Z^k_0(\vartheta_0,h_0) + Z^k_0(\vartheta_0,h_{0n}) - Z^k_0(\vartheta_{h_{0n},n},h_{0n}) \\
    &= \int_0^{\|\vartheta_{h_{0n},n}-\vartheta_0\|} \left.\frac{\partial}{\partial\delta} \left[Z_0^k\left(\vartheta_0 + \delta \frac{\vartheta_{h_{0n},n}-\vartheta_0}{\|\vartheta_{h_{0n},n}-\vartheta_0\|},h_0\right) - Z_0^k\left(\vartheta_0 + \delta \frac{\vartheta_{h_{0n},n}-\vartheta_0}{\|\vartheta_{h_{0n},n}-\vartheta_0\|},h_{0n}\right)\right] \right|_{\delta=\epsilon}\, d\epsilon \\
    &\leq\int_0^{\|\vartheta_{h_{0n},n}-\vartheta_0\|}\left|\left.\frac{\partial}{\partial\delta} \left[Z_0^k\left(\vartheta_0 + \delta \frac{\vartheta_{h_{0n},n}-\vartheta_0}{\|\vartheta_{h_{0n},n}-\vartheta_0\|},h_0\right) - Z_0^k\left(\vartheta_0 + \delta \frac{\vartheta_{h_{0n},n}-\vartheta_0}{\|\vartheta_{h_{0n},n}-\vartheta_0\|},h_{0n}\right)\right] \right|_{\delta=\epsilon} \right|  d\epsilon \\
    &\leq \|\vartheta_{h_{0n},n}-\vartheta_0\| \sup_{\epsilon\in [0,1]}\left.\left|\frac{\partial}{\partial\delta} \left[Z_0^k\left(\vartheta_0 + \delta\frac{  [\vartheta_{h_{0n},n}-\vartheta_0]}{\|\vartheta_{h_{0n},n}-\vartheta_0\|},h_0\right) - Z_0^k\left(\vartheta_0 + \delta \frac{\vartheta_{h_{0n},n}-\vartheta_0}{\|\vartheta_{h_{0n},n}-\vartheta_0\|},h_{0n}\right)\right]\right|_{\delta=\epsilon}\right| \\
    &= \|\vartheta_{h_{0n},n}-\vartheta_0\|\sup_{\epsilon\in [0,1]}\left|\frac{  [\vartheta_{h_{0n},n}-\vartheta_0]}{\|\vartheta_{h_{0n},n}-\vartheta_0\|}\left[ (\mathscr{I}_{\vartheta_0 + \epsilon[\vartheta_{h_{0n},n} -\vartheta_0]\|\vartheta_{h_{0n},n}-\vartheta_0\|^{-1}, h_{0n}}-\mathscr{I}_{\vartheta_0 + \epsilon[\vartheta_{h_{0n},n}-\vartheta_0]\|\vartheta_{h_{0n},n}-\vartheta_0\|^{-1}, h_0})_k\right]\right| \\
    &\leq \sup_{\epsilon\in [0,1]} \|\vartheta_{h_{0n},n}-\vartheta_0\|\| (\mathscr{I}_{\vartheta_0 + \epsilon[\vartheta_{h_{0n},n} -\vartheta_0]\|\vartheta_{h_{0n},n}-\vartheta_0\|^{-1}, h_{0n}}-\mathscr{I}_{\vartheta_0 + \epsilon[\vartheta_{h_{0n},n}-\vartheta_0]\|\vartheta_{h_{0n},n}-\vartheta_0\|^{-1}, h_0})_k\| \\
    &= \|\vartheta_{h_{0n},n}-\vartheta_0\|\sup_{\epsilon\in [0,1]}\| (\mathscr{I}_{\vartheta_0 + \epsilon[\vartheta_{h_{0n},n} -\vartheta_0]\|\vartheta_{h_{0n},n}-\vartheta_0\|^{-1}, h_{0n}}-\mathscr{I}_{\vartheta_0 + \epsilon[\vartheta_{h_{0n},n}-\vartheta_0]\|\vartheta_{h_{0n},n}-\vartheta_0\|^{-1}, h_0})_k\|
\end{align*}
By similar arguments outlined in Lemma \ref{lem:lasymnseqh}, one can show that there exists a neighborhood of $\vartheta_0$ in which $\hd \rightarrow \mathscr{I}_{\hd,\vartheta}$ is Lipschitz with constant $M$, where $M$ does not depend on the value of $\vartheta$. Hence,
\begin{equation}
Z^k_0(\vartheta_{h_{0n},n},h_0) - Z^k_0(\vartheta_0,h_0) + Z^k_0(\vartheta_0,h_{0n}) - Z^k_0(\vartheta_{h_{0n},n},h_{0n}) 
\leq  M\|\vartheta_{h_{0n},n}-\vartheta_0\| \|h_{0n} - h_0\|_{L_1(P)}
\end{equation} 
implies Equation \ref{eq3:lemZrootncons} by Markov's inequality.
The inequality in Equation \ref{eq1:asymnorm}, along with the results from Equations \ref{eq2:asymnorm} and \ref{eq3:lemZrootncons} imply that
\begin{align}
\label{eq:9localsymn}
\|Z_n(\vartheta_{h_{0n},n},h_0) - Z_n(\vartheta_{h_{0},n},h_0)\| = o_{P_{h_0, \vartheta_0}}(\|\vartheta_{h_{0n},n} - \vartheta_0\| + n^{-1/2}).
\end{align}
Next, by the boundedness of $\mathcal{Z}$ and $\mathcal{H}$, we have that $E[\|Z_{0}(\vartheta,h_{0n})\|^2]< \infty$. 
Moreover, by Lemma \ref{lem:locasymnorm} along with regularity conditions of the model $\{P_{\hd,\vartheta}: \vartheta \in \mathbb{R}^{p+2}, \hd \in \mathcal{H}\}$ the mapping $\vartheta \mapsto E[U_{h_0}(\vartheta, z)]$ is differentiable at $\vartheta_0$ with derivative matrix $-\mathscr{I}_{h_0, \vartheta_0}$, which is invertible by assumption. These properties, along with the Lipschitz result shown in equation \eqref{eq3:asymnorm} and the consistency of $\vartheta_{h_{0n},n}$ imply that the conditions in Theorem 5.21 in \cite{van2000asymptotic} are satisfied, with the exception that our rate in Equation \ref{eq:9localsymn} is not necessarily $o_{P_{h_0, \vartheta_0}}(n^{-1/2})$, but instead $o_{P_{h_0, \vartheta_0}}(\|\vartheta_{h_{0n},n} - \vartheta_0\| + n^{-1/2})$. However, by inspecting the proof of Theorem 5.21 in that reference, it is clear that, in the notation of that work, this result holds under our condition that 
$\mathbb{P}_n\psi_{\hat{\vartheta}_{n}}=o_P(\|\hat{\theta}_{n}-\theta_0\| + n^{-1/2})$, which holds in our setting. Therefore, we have that 
\begin{equation}
    \label{eq:10localsymn}
    \sqrt{n}[\vartheta_{h_{0n},n} - \vartheta_0] = -\mathscr{I}_{h_0, \vartheta_0}^{-1} \frac{1}{\sqrt{n}}U_{h_0}(\vartheta_0; D_n) + o_{P_{h_0, \vartheta_0}}(1).
\end{equation}
By the results shown in Lemma \ref{lem:contiguity} (particularly the identity in Equation \ref{eq:difu} and the convergence result in Equation \ref{eq2:lemcontiguity}), it is true that
\begin{equation*}
 \frac{1}{\sqrt{n}}U_{h_0}(\vartheta_0; D_n) = \frac{1}{\sqrt{n}}U_{h_{0n}}(\vartheta_0; D_n) + o_{P_{h_0, \vartheta_0}}(1),
\end{equation*}
and so, by the continuous mapping theorem,
\begin{equation}
\label{eq:11localsymn}
\mathscr{I}_{h_0, \vartheta_0}^{-1} \frac{1}{\sqrt{n}}U_{h_0}(\vartheta_0; D_n) = \mathscr{I}_{h_0, \vartheta_0}^{-1} \frac{1}{\sqrt{n}}U_{h_{0n}}(\vartheta_0; D_n) + o_{P_{h_0, \vartheta_0}}(1).
\end{equation}
Combining Equations \ref{eq:10localsymn} and \ref{eq:11localsymn}, and because $\vartheta_n= \vartheta_0 + g/\sqrt{n}$, it follows that
\begin{equation*}
\sqrt{n}[\vartheta_{h_{0n},n} - \vartheta_n] = \mathscr{I}_{h_0, \vartheta_0}^{-1} \frac{1}{\sqrt{n}} U_{h_{0n}}(\vartheta_0; D_n) - g + o_{P_{h_0, \vartheta_0}}(1).
\end{equation*}
Using very similar techniques as those that show Lemma \ref{lem:lanscore} and the fact that $\mathscr{I}_{h_{0n}, \vartheta_0} \longrightarrow \mathscr{I}_{h_0, \vartheta_0}$, it can be shown that 
\begin{equation*}
\frac{1}{\sqrt{n}}U_{h_{0n}}(\vartheta_n; D_n) - \frac{1}{\sqrt{n}}U_{h_{0n}}(\vartheta_0; D_n) + \mathscr{I}_{h_0, \vartheta_0} g = o_{P_{h_0, \vartheta_0}}(1),
\end{equation*}
which implies that 
\begin{equation*}
\frac{1}{\sqrt{n}}\mathscr{I}_{h_0, \vartheta_0}^{-1} U_{h_{0n}}(\vartheta_0; D_n) - g = \frac{1}{\sqrt{n}}\mathscr{I}_{h_0, \vartheta_0} ^{-1}U_{h_{0n}}(\vartheta_n; D_n) + o_{P_{h_0, \vartheta_0}}(1).
\end{equation*}
Hence, 
\begin{equation*}
\sqrt{n}[\vartheta_{h_{0n},n} - \vartheta_n] = \mathscr{I}_{h_0, \vartheta_0}^{-1} \frac{1}{\sqrt{n}}U_{h_{0n}}(\vartheta_n; D_n) + o_{P_{h_0, \vartheta_0}}(1).
\end{equation*}
Lemma~\ref{lem:contiguity} shows that $(P^n_{h_{0n},\vartheta_n})_{n=1}^\infty$ is contiguous with respect to $(P^n_{h_0, \vartheta_0})_{n=1}^\infty$, and applying this to the above yields the desired result.
\end{proof}
The following lemmas pertain to the asymptotic distribution of the matching estimator. Given $\vartheta \in \mathbb{R}^{p+2}$, $\hd \in \mathcal{H}$, and a sample $D_n$ drawn from $P_{\hd,\vartheta}$ and consisting of $n$ iid observations, the matching estimator $\psi_n(\pi_{\hd,\vartheta})$ is defined as 
\begin{equation}
\psi_n(\pi_{\hd,\vartheta})= \frac{1}{n}\sum_{i=1}^{n}(2a_i-1)\left(y_i - \frac{1}{M}\sum_{j = 1}^{n}I(j \in \mathcal{J}_M(i,\pi_{\hd,\vartheta}))y_j\right),
\end{equation}
where $\mathcal{J}_M(i,\pi_{\hd,\vartheta})$ is the matched set of observation $i$ based on the propensity score $\pi_{\hd,\vartheta}(\cdot)$: 
$$\mathcal{J}_M(i,\pi_{\hd,\vartheta}) = \left\{j: a_i = 1-a_j, \sum_{k = 1}^{n}I(a_i = 1-a_k)I\left(|\pi_{\hd,\vartheta}(w_j) -\pi_{\hd,\vartheta}(w_i) |\leq |\pi_{\hd,\vartheta}(w_k) -\pi_{\hd,\vartheta}(w_i) |\right) \leq M \right\}.$$
If we let $K_{M, \pi_{\hd,\vartheta}}(i)$ denote the number of times observation $i$ is used as a match, 
\begin{equation*}
K_{M, \pi_{\hd,\vartheta}}(i) = \sum_{j= 1}^{n}I(i \in \mathcal{J}_M(j,\pi_{\hd,\vartheta})),
\end{equation*}
then the matching estimator can be represented as 
\begin{equation*}
\psi_n(\pi_{\hd,\vartheta}) = \frac{1}{n}\sum_{i=1}^{n}(2A_i-1)\left(1 + \frac{K_{M, \pi_{\hd,\vartheta}}(i)}{M}\right)y_i.
\end{equation*}
Let $\bar{\mu}_{\hd,\vartheta}(a,p) := E_{\hd,\vartheta}[Y|A = a, \pi_{\hd,\vartheta}(W) = p]$ denote the conditional expectation, under $P_{\hd,\vartheta}$, of $Y$ given $a$ and propensity score $p$, and let $\mu_{\hd,\vartheta}(a,w) := E_{\hd,\vartheta}[Y|A = a, W = w]$ denote the conditional mean of $Y$ under treatment $a$ and covariates $w$. Similarly, we let $\sigma^2(a, w) := \text{var}(Y|W = w, A = a)$ be the conditional mean and variance of $Y$ given treatment level $a$ and covariates $w$ and $\bar{\sigma}^2(a, \pi_{\hd,\vartheta}(w)) := \text{var}(Y|A = a, \pi_{\hd,\vartheta}(W) = \pi_{\hd,\vartheta}(w))$ be the conditional mean and variance of $Y$ given treatment level $a$ and propensity score value $\pi_{\hd,\vartheta}(w)$. 
Finally, let $\psi_0 := E[E(Y|A = 1, W)- E(Y|A = 0, W)]$, which under our assumptions also equals $E[E(Y|A = 1, \pi_{\hd,\vartheta}(W) = p)- E(Y|A = 0, \pi_{\hd,\vartheta}(W) = p)]$, where expectations are under $P_{\hd,\vartheta}$.
Following \cite{abadie2016}, the root-$n$ scaled matching estimator can be expressed as the sum of two terms, namely $B_n$ and $R_n$, as follows:
\begin{equation*}
    \sqrt{n}[\psi_n(\pi_{\hd,\vartheta})-\psi_0]= B_n(\hd, \vartheta) + R_n(\hd, \vartheta),
\end{equation*}
where
\begin{align*}
B_n(\hd, \vartheta) &:= \frac{1}{\sqrt{n}}\sum_{i = 1}^n\Bigg\{\bar{\mu}_{\hd,\vartheta}\left(1,\pi_{\hd,\vartheta}(w_i)\right)- \bar{\mu}_{\hd,\vartheta}\left(0,\pi_{\hd,\vartheta}(w_i)\right) - \psi_0 \\
&\hspace{6em}+ (2a_i - 1)\left( 1+ \frac{K_{M,\pi_{\hd,\vartheta}}(i)}{M}\right)\left(Y_i - \mu_{\hd,\vartheta}\left(a_i,w_i\right)\right)\Bigg\}, \\
R_n(\hd, \vartheta) &:= \frac{1}{\sqrt{n}}\sum_{i = 1}^n\left\{(2a_i - 1)\left(\bar{\mu}_{\hd,\vartheta}\left(1-a_i,\pi_{\hd,\vartheta}(w_i)\right)- \frac{1}{M}\sum_{\substack{j \in \mathcal{J}_M  (i,\pi_{\hd,\vartheta}})} \bar{\mu}_{\hd,\vartheta}\left(1-a_i,\pi_{\hd,\vartheta}\left(w_j\right)\right)\right)\right\}. 
\end{align*}
The following lemma shows that $R_n$ converges to 0 in probability under sampling from $P_{h_{0n},\vartheta_n}$, and that $B_n$ and $n^{-1/2}U_{h_{0n}}(\vartheta_n; D_n)$ are jointly asymptotically normal where the covariance of their asymptotic distribution will depend on the following quantities: 
\begin{align*}
    c_{h_0, \vartheta_0}&:=E_{h_0, \vartheta_0}\Bigg[\text{Cov}[r_{h_0}(W),\mu_{h_0, \vartheta_0}(A,W)|\pi_{h_0, \vartheta_0}(W),A]\rho_{h_0, \vartheta_0}(W) \\
    &\hspace{5em} \left(\frac{A}{\pi_{h_0, \vartheta_0}(W)^2} + \frac{1-A}{(1-\pi_{h_0, \vartheta_0}(W))^2}\right)\Bigg] \\
\sigma_M^2 &:= E_{h_0, \vartheta_0}\left[\left(\bar{\mu}_{h_0, \vartheta_0}(1,\pi_{h_0, \vartheta_0}(W)) -\bar{\mu}_{h_0, \vartheta_0}(0,\pi_{h_0, \vartheta_0}(W)) -\psi_0 \right)^2\right] \\
&\; +E_{h_0, \vartheta_0}\left[\bar{\sigma}^2(1,\pi_{h_0, \vartheta_0}(W)) \left(\frac{1}{\pi_{h_0, \vartheta_0}(W) } + \frac{1}{2M} \left(\frac{1}{\pi_{h_0, \vartheta_0}(W)}-\pi_{h_0, \vartheta_0}(W) \right)\right)\right], \\
&\; +E_{h_0, \vartheta_0}\Bigg[\bar{\sigma}^2(0,\pi_{h_0, \vartheta_0}(W)) \Bigg(\frac{1}{1-\pi_{h_0, \vartheta_0}(W) } \\
& \hspace{4em} +\frac{1}{2M} \left(\frac{1}{1-\pi_{h_0, \vartheta_0}(W)}-(1-\pi_{h_0, \vartheta_0}(W)) \Bigg)\right)\Bigg].
\end{align*}

\begin{lemma}
\label{lem:brasym}
Assume that the $\mathcal{H}$-valued sequence $(h_{0n})_{n=1}^\infty$ is such that $\|h_{0n}-h_0\|_{L^2(P_{h_0, \vartheta_0})}\rightarrow 0$ as $n\rightarrow\infty$ and that the conditions of Proposition \ref{prop1} hold. Let $g \in \mathbb{R}^{p+2}$, and $\vartheta_n= \vartheta_0 + g/\sqrt{n}$. Then, under sampling from $P_{h_{0n},\vartheta_n}$,
\begin{align}
\label{eq1:brasym}
\begin{pmatrix}B_{n}(h_{0n},\vartheta_n)\\
\frac{1}{\sqrt{n}}U_{h_{0n}}(\vartheta_n; D_n)
\end{pmatrix} &\rightsquigarrow  N\left(\begin{pmatrix}0\\
0
\end{pmatrix},\begin{pmatrix}\sigma_M^{2} & c_{h_0, \vartheta_0}\\
c_{h_0, \vartheta_0} & \mathscr{I}_{h_0,\vartheta_0}
\end{pmatrix}\right)
\end{align}
and $R_n(h_{0n},\vartheta_n)\rightarrow 0$ in probability.
\end{lemma}
\begin{proof}
Fix $g \in \mathbb{R}^{p+2}$. Let $z_1 \in \mathbb{R}$, $z_2 \in \mathbb{R}^{p+2}$, and define the linear combination $C_n := z_1 B_n(h_{0n},\vartheta_n) + z_2^{\rm{T}}\frac{1}{\sqrt{n}} U_{h_{0n}}(\vartheta_n; D_n)$. Note that \begin{align*}
C_n(h_{0n},\vartheta_n) &= z_1 \frac{1}{\sqrt{n}}\sum_{i = 1}^n \left\{ \bar{\mu}_{h_{0n},\vartheta_n}\left(1, \pi_{h_{0n},\vartheta_n}\left(W_{i}\right)\right) - \bar{\mu}_{h_{0n},\vartheta_n}\left(0, \pi_{h_{0n},\vartheta_n}\left(W_{i}\right)\right) - \psi_0 \right\}\\
&\quad+ z_1 \frac{1}{\sqrt{n}} \sum_{i = 1}^n\left\{ (2A_{i} -1)\left(1 + \frac{K_{M, \pi_{h_{0n},\vartheta_n}}(i)}{M}\right) \times \left[Y_{i} - \bar{\mu}_{h_{0n},\vartheta_n}\left(A_{i}, \pi_{h_{0n},\vartheta_n}\left(W_{i}\right)\right) \right]\right\}\\
&\quad+ z_2^{\rm{T}} \frac{1}{\sqrt{n}} \sum_{i = 1}^n\left\{ r_{h_{0n}}(W_{i})\frac{A_{i}- \pi_{h_{0n},\vartheta_n}\left(W_{i}\right)}{\pi_{h_{0n},\vartheta_n}\left(W_{i}\right)\left(1-\pi_{h_{0n},\vartheta_n}\left(W_{i}\right)\right)} \rho_{h_{0n},\vartheta_n}\left(W_{i}\right)\right\}.
\end{align*}
Further, $C_n(h_{0n},\vartheta_n)$ can be re-expressed as $C_n(h_{0n},\vartheta_n) = \sum_{k = 1}^{3n}\xi_{n,k}(h_{0n},\vartheta_n)$ where for $1 \leq k \leq n$,
\begin{align*}
\xi_{n,k}(h_{0n},\vartheta_n)  &= z_1 \frac{1}{\sqrt{n}}\left\{\bar{\mu}_{h_{0n},\vartheta_n}\left(1, \pi_{h_{0n},\vartheta_n}\left(W_{k}\right)\right)- \bar{\mu}_{h_{0n},\vartheta_n}\left(0, \pi_{h_{0n},\vartheta_n}\left(W_{k}\right) \right) - \psi_0 \right\} \\
& \quad + \frac{1}{\sqrt{n}} z_2^{\rm{T}} E\left[r_{h_{0n}}(W_{k})\mid\pi_{h_{0n},\vartheta_n}\left(W_{k}\right)\right] \left[ \frac{A_{k}-\pi_{h_{0n},\vartheta_n}\left(W_{k}\right)}{\pi_{h_{0n},\vartheta_n}\left(W_{k}\right)\left(1-\pi_{h_{0n},\vartheta_n}\left(W_{k}\right)\right)}\rho_{h_{0n},\vartheta_n}\left(W_{k}\right)\right],
\end{align*}
for $n+1\leq k \leq 2n,$
\begin{align*}
\xi_{n,k}(h_{0n},\vartheta_n)  &= z_2^{\rm{T}}\frac{1}{\sqrt{n}}\left(r_{h_{0n}}(W_{k-n})- E\left[r_{h_{0n}}(W_{k-n})\mid\pi_{h_{0n},\vartheta_n}\left(W_{k-n}\right)\right]\right) \\
& \quad \left[\frac{A_{k-n} - \pi_{h_{0n},\vartheta_n}(W_{k-n})}{\pi_{h_{0n},\vartheta_n}\left(W_{k-n}\right)\left(1-\pi_{h_{0n},\vartheta_n}\left(W_{k-n}\right)\right)}\rho_{h_{0n},\vartheta_n}\left(W_{k-n}\right)\right] \\
&\quad+ z_1 \frac{1}{\sqrt{n}}(2A_{k-n}-1)\left(1+\frac{K_{M, \pi_{h_{0n},\vartheta_n}}(k-n)}{M}\right) \\
&\quad \quad\cdot\left[\mu_{h_{0n},\vartheta_n}\left(A_{k-n},W_{k-n}\right)- \bar{\mu}_{h_{0n},\vartheta_n}\left(A_{k-n},\pi_{h_{0n},\vartheta_n}\left(W_{k-n}\right)\right)\right],
\end{align*}
and, for $2n+1\leq k \leq 3n$,
\begin{align*}
\xi_{n,k}(h_{0n},\vartheta_n)  &= z_1 \frac{1}{\sqrt{n}} (2A_{k-2n} - 1)\left( 1+ \frac{K_{M, \pi_{h_{0n},\vartheta_n}}(k-2n)}{M}\right) (Y_{k-2n}  - \mu_{h_{0n},\vartheta_n}((A_{k-2n},W_{k-2n})).
\end{align*}
Now, it can be shown that
\begin{equation}
\label{eq2:lembrasym}
\left\{\sum_{j = 1}^{i} \xi_{n,j}(h_{0n},\vartheta_n), \mathcal{F}_{n,i}, 1 \leq i \leq 3n\right\}	
\end{equation}
is a martingale difference sequence with respect to the filtrations: \begin{equation*} \begin{cases}
\mathcal{F}_{n,k} = \sigma\{A_{1},...A_{k}, \vartheta_n^{\rm{T}} r_{h_{0n}}(W_{1}),..., \vartheta_n^{\rm{T}} r_{h_{0n}}(W_{k}),(h_{0j})_{j= 1}^n\}, \textnormal{ for } 1 \leq k \leq n, \\ 
\mathcal{F}_{n,k} = \sigma\{r_{h_{0n}}(W_{1}),...,r_{h_{0n}}(W_{k-n}),A_{1},...A_{n}, \vartheta_n^{\rm{T}} r_{h_{0n}}(W_{1}),...,\vartheta_n^{\rm{T}} r_{h_{0n}}(W_n),(h_{0j})_{j= 1}^n\}, \textnormal{ for }  n+1 \leq k \leq 2n, \\
\mathcal{F}_{n,k} = \sigma\{r_{h_{0n}}(W_{1}),...,r_{h_{0n}}(W_{n}),A_{1},...A_{n}, Y_{1},...Y_{k-2n},(h_{0j})_{j= 1}^n\}, \textnormal{ for }  2n+1 \leq k \leq 3n;
\end{cases}
\end{equation*}
Next, let
\begin{align}
\label{eqdefl:lembrasymn}
\lambda_{n} &:= \sum_{k=1}^{3n} E_{h_{0n},\vartheta_n}[\xi_{n,k}(h_{0n},\vartheta_n)^2|\mathcal{F}_{n,k-1}], \text{ and}\\
\label{eqdeflstr:lembrasymn}
&\lambda^*:= z_1^2 \sigma_M^2 + z_2^{\rm{T}} \mathscr{I}_{h_0, \vartheta_0}z_2 + 2 z_2^{\rm{T}} c_{h_0, \vartheta_0} z_1.
\end{align}
By assumption, for every $\epsilon > 0$,
\begin{equation}
\label{eq3:lembrasymn}
P_{h_{0n},\vartheta_n}\left(\left|\frac{\lambda_{n}}{\lambda^*} - 1\right| > \epsilon\right) \overset{n \rightarrow \infty }{\longrightarrow} 0.
\end{equation}
Additionally, by Lemma \ref{lfcond}, under $P_{h_{0n},\vartheta_n}$, for each $\epsilon > 0$, 	
\begin{align}
\label{eq4:lembrasymn}
\sum_{k=1}^{\infty}E_{h_{0n},\vartheta_n}\left[(\lambda^*)^{-1/2}\xi_{n,k}(h_{0n},\vartheta_n)^2I(|\xi_{n,k}(h_{0n},\vartheta_n)| > \epsilon)\right] \longrightarrow 0.
\end{align}
Hence, for an $\epsilon>0$ by Markov's inequality and the preceding result we have that 
\begin{align*}
&\lim_{n\to \infty}P_{h_{0n},\vartheta_n}\left(\sum_{k=1}^{n}E_{h_{0n},\vartheta_n}\left[(\lambda^*)^{-1/2}\xi_{n,k}(h_{0n},\vartheta_n)^2I(|\xi_{n,k}(h_{0n},\vartheta_n)| > \epsilon)| \mathcal{F}_{n,k-1}\right] > \epsilon \right)\\
&\leq \lim_{n\to \infty}\frac{E_{h_{0n},\vartheta_n}\left[(\lambda^*)^{-1/2}\sum_{k=1}^{n}E_{h_{0n},\vartheta_n}\left[\xi_{n,k}(h_{0n},\vartheta_n)^2I(|\xi_{n,k}(h_{0n},\vartheta_n)| > \epsilon | \mathcal{F}_{n,k-1})\right]\right]}{\epsilon} \\
&= \lim_{n\to \infty}\frac{(\lambda^*)^{-1/2}E_{h_{0n},\vartheta_n}\left[\xi_{n,k}(h_{0n},\vartheta_n)^2I(|\xi_{n,k}(h_{0n},\vartheta_n)|>\epsilon)\right]}{\epsilon} \longrightarrow 0
\end{align*}
Therefore, the conditions stated in Theorem 2 from \cite{gaenssler1986martingale} hold. Hence, under $P_{h_{0n},\vartheta_n}$, \begin{align*}
 \sum_{k = 1}^{\infty} (\lambda^*)^{-1/2}\xi_{n,k}(h_{0n},\vartheta_n)  \rightsquigarrow  N(0,1).
\end{align*}
By the Cramer-Wold device,
\begin{align*}
\begin{pmatrix}B_{n}(h_{0n},\vartheta_n)\\
n^{-1/2}U_{h_{0n}}(\vartheta_n; D_n)
\end{pmatrix} &\rightsquigarrow  N\left(\begin{pmatrix}0\\
0
\end{pmatrix},\begin{pmatrix}\sigma_M^{2} & c_{h_0, \vartheta_0}\\
c_{h_0, \vartheta_0} & \mathscr{I}_{h_0,\vartheta_0}
\end{pmatrix}\right).
\end{align*}
which shows Equation \ref{eq1:brasym}.
Next, we show that, for all $\epsilon > 0$
\begin{equation}
\label{eq2:brasym}
P_{h_{0n},\vartheta_n}\left(|R_n(h_{0n},\vartheta_n)| > \epsilon\right)\overset{n \to \infty}{\longrightarrow} 0.
\end{equation} 
Let 
\begin{equation*}
	R_{n,a}(h_{0n},\vartheta_n) := \frac{(2a-1)}{\sqrt{n}}\sum_{\{i:a_i = a\}} \left\{\bar{\mu}_{h_{0n},\vartheta_n}(1-a, \pi_{h_{0n},\vartheta_n}(w_i))- \frac{1}{M}\sum_{j \in \mathcal{J}_{M}(i, \pi_{h_{0n},\vartheta_n)}}\bar{\mu}_{h_{0n},\vartheta_n}(1-a, \pi_{h_{0n},\vartheta_n}(w_j))\right\}.
\end{equation*}
Then $R_n(h_{0n},\vartheta_n) =R_{n,1}(h_{0n},\vartheta_n)+ R_{n,0}(h_{0n},\vartheta_n)$. We show that $R_{n,1}(h_{0n},\vartheta_n) = o_{P_{h_{0n},\vartheta_n}}(1)$, the proof for $R_{n,0}(h_{0n},\vartheta_n)= o_{P_{h_{0n},\vartheta_n}}(1)$ is analogous. Fix $\hd \in \mathcal{H}$, $\vartheta \in \mathbb{R}^{p+2}$, and let $n_a$ be the number of observations in treatment arm $a$. Note that $n_0$ is lower and upper bounded by $M \leq n_0 \leq n-M$, and that $n_1 = n - n_0$. Without loss of generality, assume that the observations are ordered such that the observations with values $a_i=1$ are first. Then, the expectation of $R_{n,1}(\hd, \vartheta)$ conditional on the event $N_0 = n_0$ can be bounded as follows:
\begin{align*}
	&E_{\hd,\vartheta}\left[|R_{n,1}(\hd, \vartheta)\| N_0 = n_0\right] \\
	&\leq \frac{1}{\sqrt{n}}\sum_{i=1}^{n_1} E_{\hd,\vartheta} \left[ \left| \bar{\mu}_{\hd,\vartheta}(0, \pi_{\hd,\vartheta}(W_i))- \frac{1}{M}\sum_{j \in \mathcal{J}_{M}(i, \pi_{\hd,\vartheta}) }\bar{\mu}_{\hd,\vartheta}(0, \pi_{\hd,\vartheta}(W_j))\right|  | N_0 = n_0 \right] \\
	&=\frac{1}{\sqrt{n}}\sum_{i=1}^{n_1} E_{\hd,\vartheta} \left[ \left| \frac{1}{M}\sum_{j \in \mathcal{J}_{M}(i, \pi_{\hd,\vartheta}) } \left\{ \bar{\mu}_{\hd,\vartheta}(0, \pi_{\hd,\vartheta}(W_i))-\bar{\mu}_{\hd,\vartheta}(0, \pi_{\hd,\vartheta}(W_j))\right\}\right| N_0 = n_0 \right] \\
	&\leq \frac{1}{\sqrt{n}}\sum_{i=1}^{n_1} \left\{\frac{1}{M}\sum_{j \in \mathcal{J}_{M}(i, \pi_{\hd,\vartheta}) } E_{\hd,\vartheta} \left[ \left|  \bar{\mu}_{\hd,\vartheta}(0, \pi_{\hd,\vartheta}(W_i))-\bar{\mu}_{\hd,\vartheta}(0, \pi_{\hd,\vartheta}(W_j))\right|  | N_0 = n_0 \right] \right\} \\
	&\leq \frac{L}{\sqrt{n}}\sum_{i=1}^{n_1} \left\{\frac{1}{M}\sum_{j \in \mathcal{J}_{M}(i, \pi_{\hd,\vartheta}) } E_{\hd,\vartheta} \left[ \left| \pi_{\hd,\vartheta}(W_i)- \pi_{\hd,\vartheta}(W_j)\right|  | N_0 = n_0 \right]\right\},
\end{align*}
where the final inequality follows by the assumption that $\bar{\mu}_{\hd,\vartheta}(a,p)$ is L-Lipschitz in $p$. Let $U_{n_0, n_1, i}(m)$ be the $mth$ order statistic of $\left\{\left|\pi_{\hd,\vartheta}(W_{i})- \pi_{\hd,\vartheta}(W_{j})\right|: A_j  = 0\right\}$. Then, the last expression, up to the constant $L$, is equal to 
\begin{align*}
&\frac{1}{\sqrt{n}}\sum_{i: a_i = 1}^n \frac{1}{M}\sum_{j = 1}^M E_{\hd,\vartheta} \left[ U_{n_1,n_0,i}(j) \right]  = E_{\hd,\vartheta} \left[  \frac{1}{\sqrt{n}}\sum_{i: a_i = 1}^n \frac{1}{M}\sum_{j = 1}^M U_{n_1,n_0,i}(j) \right]. 
\end{align*}
By our assumptions and Lemmas S.1 and S.2 in \cite{abadie2016}, there exists a finite constant $r$ that depends exclusively on the support of $\mathcal{W}$ and the boundedness of $\mathcal{H}$ such that 
\begin{align*}
E_{\hd,\vartheta} \left[  \frac{1}{\sqrt{n}}\sum_{i: A_i = 1}^n \frac{1}{M}\sum_{j = 1}^M U_{n_1,n_0,i}(j) \right] &\leq r \frac{n_1}{\sqrt{n}\lfloor n_0^{3/4 }\rfloor} + M\frac{n_1}{\sqrt{n}} n_0^{M-1/4} \exp(-n_0 ^{1/4}) \\
&= \frac{r n_1\lfloor n_0^{-3/4}\rfloor n^{3/4}+ Mn_1 n_0^{M+1/2} \exp(-n_0 ^{1/4})}{\sqrt{n} n_0^{3/4}}.
\end{align*}
Moreover, $\lfloor n_0^{-3/4}\rfloor n^{3/4}$, $n_0^{M+1/2}\exp(-n_0^{1/4})$, $r$, and $M$ are bounded. Therefore, there is a constant $K$ such that
\begin{align*}
E_{\hd,\vartheta} \left[  \frac{1}{\sqrt{n}}\sum_{i: A_i = 1}^n \frac{1}{M}\sum_{j = 1}^M U_{n_1,n_0,i}(j) \right] &\leq \frac{K}{n^{1/4}}\left(\frac{n_1}{n}\frac{n^{3/4}}{n_0^{3/4}}\right).
\end{align*}
Because $r$ does not depend on $\vartheta$ nor $\hd$, neither does $K$. Combining the previous results, we have that
\begin{align*}
	E_{\hd,\vartheta}\left[|R_{n,1}(\hd, \vartheta) | |N_0 = n_0\right] & \leq \frac{LK}{n^{1/4}}\left(\frac{n_1}{n}\frac{n^{3/4}}{n_0^{3/4}}\right).
\end{align*}
By the tower property and and because $n_1 / n < 1$, we have that 
\begin{align*}
	E_{\hd,\vartheta}\left[|R_{n,1}(\hd, \vartheta)| \right] & \leq \frac{LK}{n^{1/4}}E_{\hd,\vartheta}\left[\left(\frac{n}{n_0}\right)^{3/4}\right]
\end{align*}
By Lemma S.3 in \cite{abadie2016}, there exists a constant $C$ that does not depend on $\vartheta$ nor $\hd$ such that 
\begin{align*}
E_{\hd,\vartheta}\left[\left(\frac{n}{n_0}\right)^{3/4}\right] \leq C.
\end{align*}
Therefore,
\begin{align*}
E_{\hd,\vartheta}\left[|R_{n,1}(\hd, \vartheta)| \right] \leq \frac{LK}{n^{1/4}} C.
\end{align*}
Because the above upper bound is uniform over $\vartheta$ and $\hd$, we have that $\sup_{\vartheta, \hd}E_{\hd,\vartheta}\left[|R_{n,1}(\hd, \vartheta)|\right] \overset{n \to \infty }{\longrightarrow}0$. Similarly, $\sup_{\vartheta, \hd}E_{\hd,\vartheta}\left[|R_{n,0}(\hd, \vartheta)|\right] \overset{n \to \infty }{\longrightarrow}0$. Applying Markov's inequality followed by the triangle inequality to $P_{h_{0n},\vartheta_n}\left(|R_n(h_{0n},\vartheta_n)| > \epsilon\right)$ along with the latter convergence result shows Equation \ref{eq2:brasym}.
\end{proof}

\begin{lemma}
\label{lfcond}
Let 
 \begin{equation*}
\left\{\sum_{j = 1}^{i} \xi_{n,j}(h_{0n},\vartheta_n), \mathcal{F}_{n,i}, 1 \leq i \leq 3n\right\}	
\end{equation*} 
be the martingale difference sequence defined in Lemma \ref{lem:brasym}, Equation \ref{eq2:lembrasym}. Assume that the conditions in Lemma  \ref{lem:brasym} hold. Then, for every 
$\epsilon > 0$
\begin{equation}
\label{eq1:lem12}
\sum_{k=1}^{\infty}E_{h_{0n},\vartheta_n}\left[\xi_{n,k}(h_{0n},\vartheta_n)^2I(|\xi_{n,k}(h_{0n},\vartheta_n) > \epsilon|)\right] \longrightarrow 0
\end{equation}
\end{lemma}
\begin{proof}
Instead of showing Equation \ref{eq1:lem12} directly, we show Lyapunov's condition instead, which implies the desired result. That is, we show that
$$\lim_{n \rightarrow \infty }\sum_{k = 1}^{3n} E_{h_{0n},\vartheta_n}\left[|\xi_{n,k}(h_{0n},\vartheta_n)|^{2+\delta}\right] =  0, \; \textnormal{ for some } \delta>0.$$
First, let $\tilde{\epsilon}$ be as defined in Condition \ref{cond:bounds}. Define the following upper bounds: let $\bar{w}$ denote the p+2-dimensional such that $\bar{w} = \sup_{w \in \mathcal{W}, \hd \in \mathcal{H}}r_{\hd}(w)$ entrywise, $[\underline{p},\overline{p}]$ be the upper and lower bounds of $\pi_{\hd,\vartheta}(w)$ for all $w \in \mathcal{W}$ that are uniform over $\hd\in\mathcal{H}$ and $\vartheta$ belonging to the ball $D_{\tilde{\epsilon}} = \{\vartheta \in \mathbb{R}^{p+2}: \|\vartheta-\vartheta_0\| < \tilde{\epsilon}\}$, and let $C_{\bar{\mu}} = \sup_{a \in \{0,1\}, \vartheta \in D_{\tilde{\epsilon}},  \hd \in \mathcal{H}} \bar{\mu}_{\hd,\vartheta}(a,\pi_{\hd,\vartheta}(w))$. Note that $C_{\bar{\mu}}<\infty$ by Condition \ref{cond:bounds}. Further, recall that $0< \rho_{\vartheta}(r_{\hd}(w)) <1$ for all $\vartheta \in \mathbb{R}^{p+2}, \hd \in \mathcal{H}$, and $w \in \mathcal{W}$. 

Fix $\vartheta \in D_{\tilde{\epsilon}}$, $\hd \in \mathcal{H}$. Let $\delta$ be some positive real number such that there exists a $C_\delta<\infty$ such that $E_{\hd,\vartheta}[|Y|^{2+\delta}|A=a, \pi_{\hd,\vartheta}(W)=p]$ is uniformly bounded by $C_\delta$ (such a $\delta$ is guaranteed to exist by Condition \ref{cond:bounds}). In the following equations all expectations are taken with respect to $P_{\hd,\vartheta}$. 

We divide $\sum_{k=1}^{3n}E\left[|\xi_{n,k}(\hd, \vartheta)|^{2+\delta}\right]$ into three separate sums and show that each converge to 0. For the first sum, we have that
\begin{align}
\label{eq2:lem12}
\sum_{k = 1}^{n} E\left[|\xi_{n,k}(\hd, \vartheta)|^{2+\delta}\right] &\leq \sum_{k = 1}^{n}  E\left[\bigg| z_1 \frac{1}{\sqrt{n}}(\bar{\mu}_{\hd,\vartheta}(1, \pi_{\hd,\vartheta}(W_{k}))- \bar{\mu}_{\hd,\vartheta}(0, \pi_{\hd,\vartheta}(W_{k})) - \psi_0)\bigg|^{2+\delta}\right]  \nonumber \\ 
&\quad + \sum_{k = 1}^{n} E\bigg[ \bigg|\left(\frac{1}{\sqrt{n}} z_2^{\rm{T}} E[r_{\hd}(W_{k})|\pi_{\hd,\vartheta}(W_{k})]\right)  \nonumber\\
& \hspace{4em} \left(\frac{A_{k}-\pi_{\hd,\vartheta}((W_{k}))}{\pi_{\hd,\vartheta}(W_{k})(1-\pi_{\hd,\vartheta}(W_{k}))}\rho_{\vartheta}(r_{\hd}(W_{k}))\right) \bigg|^{2+\delta}\bigg]  \nonumber \\
&\leq \frac{|2 z_1 C_{\bar{\mu}} + z_1|\psi_0||^{2+\delta}}{n^{\delta/2}} +  \sum_{k = 1}^{n}E\bigg[ \bigg|\left(\frac{1}{\sqrt{n}} z_2^{\rm{T}} E[r_{\hd}(W_{k})|\pi_{\hd,\vartheta}(W_{k})]\right) \nonumber  \\
& \hspace{10em}\left(\frac{A_{n,k}-\pi_{\hd,\vartheta}(W_{k})}{\pi_{\hd,\vartheta}(W_{k})(1-\pi_{\hd,\vartheta}(W_{k}))}\rho_{\vartheta}(r_{\hd}(W_{k}))\right)\bigg|^{2+\delta}\bigg]  \nonumber \\ 
&\leq \frac{|2 z_1 C_{\bar{\mu}} + z_1|\psi_0||^{2+\delta}}{n^{\delta/2}} + \frac{|z_{2}^{\rm{T}}\bar{w}|^{2+\delta}}{n^{\delta/2}} \left(\frac{1}{\underline{p}(1-\overline{p})}\right)^{2+\delta},
\end{align}
where the first inequality holds due to the triangle inequality. The latter upper bound converges to 0 as $n\to \infty$. To bound the second sum, first let 
$\Delta_{\hd,\vartheta}(a,w,\pi) : = E[Y|A=a,W=w] - E[Y|A=a,\pi_{\hd,\vartheta}(W) = \pi_{\hd,\vartheta}(w)]$. Then, $E\left[\mid\Delta_{\hd,\vartheta}(A,W,\pi)\mid^{2+\delta} |\pi_{\hd,\vartheta}(W), A\right]$ can be bounded as follows:
\begin{align*}
E\left[|\Delta_{\hd,\vartheta}(A,W,\pi)|^{2+\delta} |\pi_{\hd,\vartheta}(W), A\right]
& \leq E\left[|E(Y|A, W)|^{2+\delta}|\pi_{\hd,\vartheta}, A \right] +  E\left[|E(Y|A, \pi_{\hd,\vartheta}(W)) |^{2+\delta}|\pi_{\hd,\vartheta}(W),A\right] \\
&\leq E\left[E(|Y|^{2 + \delta }|A, W) |\pi_{\hd,\vartheta}(W), A\right] + E[|Y|^{2+\delta}|\pi_{\hd,\vartheta}(W),A] \\ 
&= 2 E[|Y|^{2+\delta}|\pi_{\hd,\vartheta}(W)=p,A]] \leq 2C_{\delta}
\end{align*}
Let $C_{\Delta}:=2C_{\delta}$ where, by assumption, $C_{\Delta}$ does not depend on $\vartheta$, nor $\hd$, nor $p$. Then, we have that 
\begin{align*}
& \sum_{k=n+1}^{2n} E[|\xi_{n,k}(\hd, \vartheta)|^{2+\delta}] \leq \sum_{k=n+1}^{2n} E\bigg\{ \bigg| \left[z_2^{\rm{T}}\frac{1}{\sqrt{n}}(r_{\hd}(W_{k-n})- E[r_{\hd}(W_{k-n})|\pi_{\hd,\vartheta}(W_{k-n})])\right] \\
& \;\hspace{15em} \left[\frac{A_{k-n} - \pi_{\hd,\vartheta}(W_{k-n})}{\pi_{\hd,\vartheta}(W_{k-n})(1-\pi_{\hd,\vartheta}(W_{k-n}))}\rho_\vartheta(r_{\hd}(W_{k-n}))\right] \bigg|^{2+\delta}\bigg\}  \\
&\quad+\sum_{k=n+1}^{2n} E\left\{\bigg|\frac{z_1}{\sqrt{n}}(2A_{k-n}-1)\left(1 + \frac{K_{M, \pi_{\hd,\vartheta}}(k-n)}{M}\right) \Delta_{\hd,\vartheta}(A_{k-n},W_{k-n},\pi_{\hd,\vartheta}) \bigg|^{2+\delta}\right\}\\
& \leq \sum_{k=n+1}^{2n}  E \left\{ \bigg| z_2^{\rm{T}}\frac{1}{\sqrt{n}}(r_{\hd}(W_{k-n}))  \left(\frac{A_{k-n} - \pi_{\hd,\vartheta}(W_{k-n})}{\pi_{\hd,\vartheta}(W_{k-n})(1-\pi_{\hd,\vartheta}(W_{k-n}))}\right)\rho_\vartheta(r_{\hd}(W_{k-n})) \bigg|^{2+\delta}\right\}  \\
&\quad+\sum_{k=n+1}^{2n}  E\left\{ \bigg| z_2^{\rm{T}}\frac{1}{\sqrt{n}}(E[r_{\hd}(W_{k-n})|\pi_{\hd,\vartheta}(W_{k-n})]) \left(\frac{A_{k-n} - \pi_{\hd,\vartheta}(W_{k-n})}{\pi_{\hd,\vartheta}(W_{k-n})(1-\pi_{\hd,\vartheta}(W_{k-n}))}\right)\rho_\vartheta(r_{\hd}(W_{k-n}))\bigg|^{2+\delta}\right\}  \\
&\quad+\sum_{k=n+1}^{2n}E\left\{\bigg|\frac{z_1}{\sqrt{n}}(2A_{k-n}-1)\left(1+\frac{K_{M, \pi_{\hd,\vartheta}}(k-n)}{M}\right) \Delta_{\hd,\vartheta}(A_{k-n},W_{k-n},\pi_{\hd,\vartheta})\bigg|^{2+\delta}\right\} \\
& \leq 2 \frac{|z_{2}^{\rm{T}}\bar{w}|^{2+\delta}}{n^{\delta/2}}\bigg|\frac{1}{\underline{p}(1-\overline{p})}\bigg|^{2+\delta} \\
& \hspace{1em} +\sum_{k=n+1}^{2n}E\left\{\bigg|\frac{z_1}{\sqrt{n}}(2A_{k-n}-1)\left(1+\frac{K_{M, \pi_{\hd,\vartheta}}(k-n)}{M}\right) \Delta_{\hd,\vartheta}(A_{k-n},W_{k-n},\pi_{\hd,\vartheta})\bigg|^{2+\delta}\right\},
\end{align*}
where the first and second inequalities are due to the triangle inequality and the third is due to the boundedness of $\mathcal{W}$ and $\mathcal{H}$, and the fact that $\vartheta \in D_{\tilde{\epsilon}}$. Next, let $A^n$ denote the treatment values for all $n$ observations. Similarly, let $\pi_{\hd,\vartheta}^{n}$ denote all the propensity score values for all $n$ observations. When conditioning on $A^n$ and $\pi_{\hd,\vartheta}^{n}$, each $K_{M,\pi_{\hd,\vartheta}}(i)$ term is constant for all $i \in \{n+1,...,2n\}$. Hence,
\begin{align*}
&\sum_{k=n+1}^{2n}E\left\{\bigg|\frac{z_1}{\sqrt{n}}(2A_{k-n}-1)\left(1+\frac{K_{M, \pi_{\hd,\vartheta}}(k-n)}{M}\right) \Delta_{\hd,\vartheta}(A_{k-n},W_{k-n},\pi_{\hd,\vartheta})\bigg|^{2+\delta}\right\} \\
&=\sum_{k=n+1}^{2n}E\left[E\left\{\bigg|\frac{z_1}{\sqrt{n}}(2A_{k-n}-1)\left(1+\frac{K_{M, \pi_{\hd,\vartheta}}(k-n)}{M}\right) \Delta_{\hd,\vartheta}(A_{k-n},W_{k-n},\pi_{\hd,\vartheta})\bigg|^{2+\delta}\middle|A^n,\pi_{\hd,\vartheta}^{n}\right\}\right] \\
&=\frac{z_1^2}{n^{1+\delta/2}}\sum_{k=n+1}^{2n}E\left[ \bigg|(2A_{k-n}-1)\left(1+\frac{K_{M, \pi_{\hd,\vartheta}}(k-n)}{M}\right)\bigg|^{2+\delta}E\left\{\Delta_{\hd,\vartheta}(A_{k-n},W_{k-n},\pi_{\hd,\vartheta})^{2+\delta}|A^n,\pi_{\hd,\vartheta}^{n}\right\}\right] \\
&\leq \frac{C_{\Delta}z_1^2}{n^{1+\delta/2}}\sum_{k=n+1}^{2n}E\left[ \bigg|(2A_{k-n}-1)\left(1+\frac{K_{M, \pi_{\hd,\vartheta}}(k-n)}{M}\right)\bigg|^{2+\delta}\right].
\end{align*}
Nearly identical arguments to those used to prove Lemma S.8 in \cite{abadie2016} show that $E[K^{\lceil2+\delta\rceil}_{M,\pi_{\hd,\vartheta}}(i)|A_i = a]$ is uniformly bounded in $n$ by a finite constant that neither depends on $\vartheta$ nor $\hd$. We denote this constant by $L$. Moreover, because $K_{M, \pi_{\hd,\vartheta}}(i)$ is a positive integer, we have that $E[K^{2+\delta}_{M,\pi_{\hd,\vartheta}}(i)|A_i = a]\leq E[K^{\lceil2+\delta\rceil}_{M,\pi_{\hd,\vartheta}}(i)|A_i = a]$, where $\lceil\,\cdot\,\rceil$ denotes the ceiling function. Therefore, 
\begin{equation}
\label{eq3:lem12}
\sum_{k = n+1}^{2n} E[|\xi_{n,k}(\hd, \vartheta)|^{2+\delta}] \leq \frac{LC_{\Delta}z_1^2}{n^{\delta/2}},
\end{equation}
and recall that neither $L$ nor $C_{\Delta}$ depend on $\vartheta$ or $\hd$. The last sum can be bounded as follows:
\begin{align*}
&\sum_{k = 2n+1}^{3n} E[|\xi_{n,k}(\hd, \vartheta)|^{2+\delta}] \\
&= \sum_{k = 2n+1}^{3n} E\left[\bigg|z_1 \frac{1}{\sqrt{n}} (2A_{k-2n} - 1)\left( 1+ \frac{K_{M,\pi_{\hd,\vartheta}}(k-2n)}{M}\right) \times (Y_{k-2n}  - \mu(A_{k-2n}, W_{k-2n}))\bigg|^{2+\delta}\right] \\ 
&\leq \sum_{k = 2n+1}^{3n} E\left[\bigg|z_1 \frac{1}{\sqrt{n}} (2A_{k-2n} - 1)\left( 1+ \frac{K_{M,\pi_{\hd,\vartheta}}(k-2n)}{M}\right) Y_{k-2n}\bigg|^{2+\delta}\right] \\ 
&\quad+ \sum_{k = 2n+1}^{3n} E\left[\bigg|z_1 \frac{1}{\sqrt{n}} (2A_{k-2n} - 1)\left( 1+ \frac{K_{M,\pi_{\hd,\vartheta}}(k-2n)}{M}\right) \mu(A_{k-2n},W_{k-2n})\bigg|^{2+\delta}\right].
\end{align*}
Similar arguments as those employed to show Equation \ref{eq2:lem12} (broadly, conditioning on $A^n$ and $\pi_{\hd,\vartheta}^n$ followed by bounding the terms in the sum) can be used to show that
\begin{equation}
\label{eq4:lem12}
E\left[\bigg|\sum_{k = 2n+1}^{3n}|\xi_{n,k}(\hd, \vartheta)|^{2+\delta}\right] \leq \frac{z_1^2M_{\delta}}{n^{\delta/2}}, 
\end{equation}
where $M_{\delta}$ depends on neither $\vartheta$ nor $\hd$.

By the triangle inequality and since the bounds on Equations  \ref{eq1:lem12}, \ref{eq2:lem12}, \ref{eq3:lem12}, and \ref{eq4:lem12} are uniform over $\vartheta$ and $\hd$, we have that $\sup_{\vartheta \in D_{\tilde{\epsilon}}, \hd \in \mathcal{H}}E\left[\sum_{i=1}^{3n}|\xi_{n,k}(\hd, \vartheta)|^{2+\delta}\right] \rightarrow 0$. Finally, because $P_{h_{0n},\vartheta_n}\left(\vartheta_n \not \in D_{\tilde{\epsilon}}\right) \overset{n \to \infty}{\longrightarrow} 0$, 
\begin{align*}
\sum_{k = 1}^{3n} E_{h_{0n},\vartheta_n}\left[|\xi_{n,k}(h_{0n},\vartheta_n)|^{2+\delta}\right]\overset{n \to \infty}{\longrightarrow} 0.
\end{align*}
\end{proof}
\section{Proof of Proposition 1}\label{supsec:sec5}
Similar to the results in \cite{abadie2016}, our main asymptotic normality theorem provides guarantees for the discretized version of the maximum likelihood estimator. That is, we employ the same result from \cite{andreou2012alternative} that allows us to show the asymptotic distribution of the matching estimator when using the following transformation of the MLE to match. For a given $\hd \in \mathcal{H}$, let $\vartheta_{\hd,n}$ be the MLE of $\vartheta_0$ based on $\hd$ and $d_n$. Additionally, let $k>0$ be the discretization constant that we use to define a grid of cubes in $\mathbb{R}^{p+2}$ of sides of length $k/\sqrt{n}$. Informally, the discretized estimator $\vartheta_{\hd,n,k}$ is obtained by taking the midpoint of the cube $\vartheta_{\hd,n}$ belongs to. Formally, if $r:\mathbb{R}^{p+2}\rightarrow \mathbb{Z}^{p+2}$ is a function that rounds each entry of a vector to its nearest integer, then $\vartheta_{\hd,n,k}:= k r(\sqrt{n}\vartheta_{\hd,n}/k)/\sqrt{n}$. 
The following proposition shows that asymptotic distribution of the matching estimator when carrying matching via the discretized estimator $\vartheta_{h_{0n},n,k}$ is asymptotically normal under a set of conditions. It is worth noting that our asymptotic result is based on taking the limit with respect to the probability measure based on the shifted discretized version of the truth defined as $\vartheta_{0,n,k}:= k r(\sqrt{n}\vartheta_0 / k)/\sqrt{n} + g/\sqrt{n}$, where $g\in \mathbb{R}^{p+2}$. 
\begin{proposition}
\label{prop1}
Suppose that Conditions \ref{cond:pos:first}-\ref{cond:ps} and \ref{cond:bounds}-\ref{cond:ulan:last} hold. Let $(h_{0n})_{n=1}^{\infty}$ be an $\mathcal{H}-$ valued sequence such that $\|h_{0n} - h_0\|_{L^2(P_{h_0,\vartheta_0,})} \rightarrow 0$, and let $g\in \mathbb{R}^{p+2}$. Additionally, let: $\vartheta_{h_{0n},n}= \argmax_{\vartheta} \ell_{h_{0n}}(\vartheta;D_n)$ be the maximum likelihood estimator based on observations $D_n = \{A_i, W_i, h_{0n}(W_i)\}_{i = 1}^n$, let $\vartheta_{h_{0n},n,k}$ denote its discretized version with discretization constant $k$, and let $\psi_{n,h_{0n},k}=\psi_{n}(\pi_{h_{0n}, \vartheta_{h_{0n},n,k}})$ denote the matching estimator based on propensity score $\pi_{h_{0n}, \vartheta_{h_{0n},n,k}}$. Finally, let $\vartheta_{0,n,k}= k r(\sqrt{n}\vartheta_0 / k)/\sqrt{n} + g/\sqrt{n}$ be the shifted discretized version of $\vartheta_0$.
Then,
\begin{align*}
\lim_{k \downarrow 0}\lim_{n \to \infty}&P_{\vartheta_{0,n,k},h_{0n}}\left(\sqrt{n}(\sigma_M^2 - c_{h_0, \vartheta_0}^{\rm{T}}\mathscr{I}_{h_0, \vartheta_0}c_{h_0, \vartheta_0})^{-1/2}(\psi_{n,h_{0n},k} - \psi_0) \leq z\right) = \Phi(z).
\end{align*}
where $\Phi$ is the cumulative distribution function of a standard normal random variable.
\end{proposition}

\begin{proof}
Let $\vartheta_n = \vartheta_0+g/\sqrt{n}$. By Lemma \ref{lem:brasym}, we have that, under $P_{h_{0n},\vartheta_n}$,
\begin{align*}
\begin{pmatrix}B_{n}(h_{0n},\vartheta_n)\\
\frac{1}{\sqrt{n}}U_{h_{0n}}(\vartheta_n; D_n)
\end{pmatrix} &\rightsquigarrow  N\left(\begin{pmatrix}0\\
0 
\end{pmatrix},\begin{pmatrix}\sigma_M^{2} & c_{h_0, \vartheta_0} \\
c_{h_0, \vartheta_0} & \mathscr{I}_{h_0, \vartheta_0}
\end{pmatrix}\right).
\end{align*}
Moreover, also because of Lemma \ref{lem:brasym}, we have that $R_n(h_{0n},\vartheta_n)\rightarrow 0$ in probability. By Slutsky's Theorem, under $P_{h_{0n},\vartheta_n}$,
\begin{align}
\label{eq1:theo2}
\begin{pmatrix}\sqrt{n}(\psi_{n}(\pi_{h_{0n},\vartheta_n})-\psi_0) \\
\frac{1}{\sqrt{n}}U_{h_{0n}}(\vartheta_n; D_n)
\end{pmatrix} &\rightsquigarrow  N\left(\begin{pmatrix}0\\
0 
\end{pmatrix},\begin{pmatrix}\sigma_M^{2} & c_{h_0, \vartheta_0} \\
c_{h_0\vartheta_0} & \mathscr{I}_{h_0, \vartheta_0}
\end{pmatrix}\right).
\end{align}
Applying the continuous mapping theorem to \eqref{eq1:theo2} with the bivariate function $f(x,y) = (x,\mathscr{I}_{h_0, \vartheta_0}^{-1}y,-g^{\rm{T}}y  -g^{\rm{T}}\mathscr{I}_{h_0, \vartheta_0}g/2)$, we have that under $P_{h_{0n},\vartheta_n}$,
\begin{align*}
\begin{pmatrix}\sqrt{n}(\psi_{n}(\pi_{h_{0n},\vartheta_n})-\psi_0) \\
\mathscr{I}_{h_0, \vartheta_0}^{-1}\frac{1}{\sqrt{n}}U_{h_{0n}}(\vartheta_n;D_n)\\
-g^{\rm{T}}\frac{1}{\sqrt{n}}U_{h_{0n}}(\vartheta_n;D_n) - g^{\rm{T}}\mathscr{I}_{h_0, \vartheta_0}g/2
\end{pmatrix} \rightsquigarrow  N\left(\begin{pmatrix}0\\
0\\
-g^{\rm{T}}\mathscr{I}_{h_0,\vartheta_0}g/2
\end{pmatrix},\begin{pmatrix}\sigma_{M}^{2} & c_{h_0,\vartheta_0}^{\rm{T}}\mathscr{I}_{h_0,\vartheta_0}^{-1} & -c_{h_0,\vartheta_0}^{\rm{T}}g\\
\mathscr{I}_{h_0,\vartheta_0}^{-1}c_{h_0,\vartheta_0}  & \mathscr{I}_{h_0,\vartheta_0}^{-1} & -g\\
-g^{\rm{T}}c_{h_0,\vartheta_0} & -g^{\rm{T}} & g^{\rm{T}}\mathscr{I}_{h_0,\vartheta_0}g
\end{pmatrix}\right).
\end{align*}
By Slutsky's theorem, Lemma \ref{lem:lasymnseqh}, and Lemma \ref{lem:asymnorm}, we have that under $P_{h_{0n},\vartheta_n}$,
\begin{align*}
\begin{pmatrix}\sqrt{n}(\psi_{n}(\pi_{h_{0n},\vartheta_n})-\psi_0)\\
\sqrt{n}\left(\vartheta_{h_{0n},n} -\vartheta_n\right)\\
\Lambda\left(\vartheta_0|\vartheta_n\right)
\end{pmatrix}\rightsquigarrow N\left(\begin{pmatrix}0\\
0\\
-g^{\rm{T}}\mathscr{I}_{h_0, \vartheta_0}g/2
\end{pmatrix},\begin{pmatrix}\sigma_{M}^{2} & c_{h_0,\vartheta_0}^{\rm{T}}\mathscr{I}_{h_0, \vartheta_0}^{-1} & -c_{h_0,\vartheta_0}^{\rm{T}}g\\
\mathscr{I}_{h_0, \vartheta_0}^{-1}c_{h_0,\vartheta_0} & \mathscr{I}_{h_0, \vartheta_0}^{-1} & -g\\
-g^{\rm{T}}c_{h_0,\vartheta_0} & -g^{\rm{T}} & g^{\rm{T}}\mathscr{I}_{h_0, \vartheta_0}g
\end{pmatrix}\right).
\end{align*}
Therefore, Condition AN in \cite{andreou2012alternative} is satisfied. Next, because we assumed that Condition \ref{cond:ulan:last} is true, the (ULAN) condition in \cite{andreou2012alternative} applies in our setting. Then, we have that,
\begin{align*}
\lim_{k \downarrow 0}\lim_{n \to \infty}&P_{{h_{0n},\vartheta_{0,n,k}}}\left(\sqrt{n}(\sigma_M^2 - c_{h_0, \vartheta_0}^{\rm{T}}\mathscr{I}_{h_0, \vartheta_0}^{-1}c_{h_0, \vartheta_0})^{-1/2}(\psi_{n,h_{0n},k} - \psi_0) \leq z\right) = \Phi(z).
\end{align*}
\end{proof}

\section{Proof of Corollary~\ref{cor2}} \label{supsec:sec6}

In the following corollary we show the convergence statement from Proposition \ref{prop1} when the sequence of $\mathcal{H}$- valued functions is random and obtained from a dataset $S=\{O_{-i}\}_{i=1}^\infty$ consisting of infinitely many iid draws from $P_{h_0, \vartheta_0}$. Let $\{H_n\}_{n= 1}^{\infty}$ denote a sequence of operators that take as input a realization $s$ of the independent dataset and output an estimate $H_n(s)\in\mathcal{H}$ of $h_0$. We mostly have in mind the case where $H_n(S)$ is an estimate of $h_0$ based on the first $m_n<\infty$ observations $\{O_{-i}\}_{i=-1}^{-m_n}$ in $S$, where $(m_n)_{n=1}^\infty$ is some increasing sequence, 
though our specification of $\{H_n\}_{n= 1}^{\infty}$ enables consideration of more general specifications of the functions $H_n$. 
We will establish the convergence in probability of the sequence of matching estimators, assuming that $\|H_n(S) - h_0\|_{L^2(P_{h_0, \vartheta_0})} \rightarrow 0$ holds almost surely under sampling of $S$. 
Hereafter we will denote $H_n(S)$ as $h_n$. As in Proposition \ref{prop1}, we will establish a distributional result regarding the sequence of matching estimators under sampling from a discretized distribution, namely the distribution indexed by the unshifted discretized parameter $\vartheta_{0,n,k}^u:= k r(\sqrt{n}\vartheta_0 / k)/\sqrt{n}$. Since the final entry of $\vartheta_0$ is zero by assumption, the final entry of $\vartheta_{0,n,k}^u$ is zero as well. Hence, the value of $P_{\hd,\vartheta_{0,n,k}^u}$ does not depend on the value of $\hd \in\mathcal{H}$. Hereafter we shall write $P_{\vartheta_{0,n,k}^u}$, rather than $P_{\hd,\vartheta_{0,n,k}^u}$, to make this clear.

Our distributional result will hold conditional on $S$. 
Before showing the Corollary, we clarify what we mean by this notion of convergence. For each $n$, let $F_n$ denote a random $\mathcal{D}^n\rightarrow\mathbb{R}$ function that is measurable with respect to the $\sigma$-field generated by $S$, $\mathcal{D}:=\mathbb{R}^{p+2}\times\{0,1\}\times\mathbb{R}$ contain the support of each $P_{\hd,\vartheta}$, and $D_n$ denote a random dataset of $n$ iid draws from $P_{\vartheta_{0,n,k}^u}$, where this dataset is independent of $S$. 
We say that, under $P_{\vartheta_{0,n,k}^u}$ and conditionally on $S$, $F_n(D_n)$ converges to a random variable $Z$ if
\begin{align*}
    P_{h_0, \vartheta_0}\left(\lim_{k \downarrow 0 }\lim_{n \to \infty }{\rm Pr}(F_n(D_n) \leq z\mid S) =G_{Z}(z)\right ) = 1,
\end{align*}
where $G_Z$ is the cumulative distribution function of $Z$ and ${\rm Pr}(\,\cdot\mid S)$ denotes the conditional distribution of $(D_n,S)$ given $S$. The display above depends on the probability measure $P_{\vartheta_{0,n,k}^u}$ through $D_n$, which is an iid sample from $P_{\vartheta_{0,n,k}^u}$.
\begin{corollary}
\label{cor2}
Suppose that the conditions in Proposition \ref{prop1} hold. Suppose that $\|h_n - h_0\|_{L^2(P_{h_0, \vartheta_0})} \rightarrow 0$ almost surely. 
Let $\vartheta_{h_n,n}= \argmax_{\vartheta} \ell_{h_n}(\vartheta;D_n)$ be the maximum likelihood estimator based on observations $D_n = \{A_i, W_i, h_n(W_i)\}_{i = 1}^n$, let $\vartheta_{h_n,n,k} = k r(\sqrt{n}\vartheta_{h_n,n}/ k)/\sqrt{n}$ denote its discretized version with discretization constant $k$, and let $\psi_{n,h_n,k}=\psi_{n}(\pi_{h_n, \vartheta_{h_n,n,k}})$ denote the matching estimator based on propensity score $\pi_{h_n,\vartheta_{h_n,n,k}}$. 

Then, under $P_{\vartheta_{0,n,k}^u}$ and conditionally on $S$,
\begin{align*}
    \sqrt{n}(\sigma_M^2 - c_{h_0, \vartheta_0}^{\rm{T}}\mathscr{I}_{h_0, \vartheta_0}^{-1}c_{h_0, \vartheta_0})^{-1/2}(\psi_{n,h_n,k}-\psi_0)
\end{align*}
converges to a standard normal random variable.
\end{corollary}
\begin{proof}
Let $\Phi$ denote the cumulative distribution function of a standard normal random variable. By Proposition \ref{prop1}, taking the shift $g$ in that result to be equal to zero, we have that 
\begin{align*}
\lim_{k \downarrow 0}\lim_{n \to \infty}P_{\vartheta_{0,n,k}^u}\left(\sqrt{n}(\sigma_M^2 - c_{h_0, \vartheta_0}^{\rm{T}}\mathscr{I}_{h_0, \vartheta_0}^{-1}c_{h_0, \vartheta_0})^{-1/2}(\psi_{n,h_{0n},k} - \psi_0) \leq z \right) = \Phi(z),
\end{align*}
for every $\mathcal{H}$-valued deterministic sequence $\{h_{0n}\}_{n=1}^{\infty}$ such that $\|h_{0n} - h_0\|_{L^2(P_{h_0,\vartheta_0})}\rightarrow 0$. Hence,
\begin{align*}
    &\left\{s\in \mathcal{S}: \lim_{n \to \infty}\|H_n(s) - h_0\|_{L^2(P_{h_0, \vartheta_0)}} \rightarrow 0\right\} \\
    &\subseteq \left\{s\in\mathcal{S} : \lim_{k \downarrow 0}\lim_{n \to \infty}P_{\vartheta_{0,n,k}^u}\left(\sqrt{n}(\sigma_M^2 - c_{h_0, \vartheta_0}^{\rm{T}}\mathscr{I}_{h_0, \vartheta_0}^{-1}c_{h_0, \vartheta_0})^{-1/2}(\psi_{n,H_n(s),k} - \psi_0) \leq z \right) = \Phi(z) \right\}.
\end{align*}
By assumption, the event $\left\{s\in \mathcal{S}: \lim_{n \to \infty}\|H_n(s) - h_0\|_{L^2(P_{h_0,\vartheta_0})} \rightarrow 0\right\}$ has probability one under sampling the variates in $S$ independently from $P_{h_0, \vartheta_0}$. Hence, the event on the right-hand side also has probability one. Finally, because $S$ is independent of $D_n$, the probability in the event on the right-hand side is equal to
\begin{align*}
    {\rm Pr}\left(\sqrt{n}(\sigma_M^2 - c_{h_0, \vartheta_0}^{\rm{T}}\mathscr{I}_{h_0, \vartheta_0}^{-1}c_{h_0, \vartheta_0})^{-1/2}(\psi_{n,h_n,k}  - \psi_0) \leq z \,\mid\, S\right).
\end{align*}
\end{proof}


\end{document}